\documentclass[11pt,runningheads]{llncs}
\usepackage[a4paper, margin = 1in]{geometry}

\usepackage[T1]{fontenc}
\usepackage{booktabs}
\usepackage[ruled]{algorithm2e}

\SetAlFnt{\small}
\SetAlCapFnt{\small}
\SetAlCapNameFnt{\small}
\SetAlCapHSkip{0pt}
\IncMargin{-\parindent}

\usepackage{fix-cm}
\usepackage{graphicx}
\usepackage[titletoc]{appendix}
\usepackage{xcolor}

\usepackage{amsmath,amssymb}
\usepackage{xifthen}
\usepackage{stackengine}
\usepackage{comment}
\usepackage[utf8]{inputenc}
\usepackage{nicefrac}  
\usepackage{xfrac}
\usepackage{enumerate}
\usepackage{tikz}
\usetikzlibrary{arrows.meta}
\usepackage{CJKutf8} 
\usepackage[round,authoryear]{natbib}
\usepackage{hyperref}
\usepackage{cleveref} 
\hypersetup{
	colorlinks = true,
	linkcolor=blue!70!black,
	citecolor=blue!70!black,
	linkbordercolor = {white}
}

\newcommand{\squishlist}{
\begin{list}{{{\small{$\bullet$}}}}
{\setlength{\itemsep}{3pt}      \setlength{\parsep}{1pt}
\setlength{\topsep}{1pt}       \setlength{\partopsep}{0pt}
\setlength{\leftmargin}{1em} \setlength{\labelwidth}{1em}
\setlength{\labelsep}{0.5em} } }
\newcommand{\squishend}{  \end{list}}

\newcommand{\xhdr}[1]{\vspace{8pt}\noindent{\bf {#1.}}\ }
\newcommand{\sxhdr}[1]{%
  \par\vspace{5pt}%
  {\bfseries #1.}\enspace\ignorespaces%
}

\newcommand{\omt}[1]{} 

\usepackage{pgfplots}
\pgfplotsset{compat=newest}
\usepgfplotslibrary{groupplots}
\usetikzlibrary{patterns}
\usetikzlibrary{decorations.pathreplacing,arrows.meta}
\usepackage{tikz}
\usetikzlibrary{decorations.markings, decorations.pathreplacing}

\newcommand{\cc}[1]{\ensuremath{\mathsf{#1}}}

\newcommand{\Zplus}{\mathbb{Z}^+}

\newcommand{\prob}[2][]{\text{Pr}\ifthenelse{\not\equal{}{#1}}{_{#1}}{}\!\left[{\def\givenn{\middle|}#2}\right]}
\newcommand{\expect}[2][]{\mathbb{E}\ifthenelse{\not\equal{}{#1}}{_{#1}}{}\!\left[{\def\givenn{\middle|}#2}\right]}
\newcommand{\indicator}[2][]{\mathbf{1}\ifthenelse{\not\equal{}{#1}}{_{#1}}{}\!\left\{{\def\givenn{\middle|}#2}\right\}}
\newcommand{\Variance}[2][]{\text{Var}\ifthenelse{\not\equal{}{#1}}{_{#1}}{}\!\left[{\def\givenn{\middle|}#2}\right]}

\newcommand{\E}{\mathbb{E}}

\newcommand{\dd}{\ \mathrm{d}}

\newcommand{\supp}{\cc{supp}}

\def\expect{\Ex}

\newcommand{\MPC}{\cc{MPC}}

\newcommand{\IndexAlg}{\emph{Index Algorithm }}

\newcommand{\AutoAdjust}[3]{\mathchoice{ \left #1 #2  \right #3}{#1 #2 #3}{#1 #2 #3}{#1 #2 #3} }
\newcommand{\Xcomment}[1]{{}}

\newcommand{\InBrackets}[1]{\AutoAdjust{[}{#1}{]}}
\newcommand{\Ex}[2][]{\operatorname{\mathbb E}_{#1}\InBrackets{#2}}
\def\expect{\Ex}

\definecolor{deepblue}{rgb}{0,0,0.5}
\definecolor{deepred}{rgb}{0.6,0,0}
\definecolor{deepgreen}{rgb}{0,0.5,0}
\definecolor{halfgray}{gray}{0.55}

\makeatletter
\renewenvironment{proof}[1][\proofname]{\par
  \normalfont
  \topsep6\p@\@plus6\p@
  \trivlist
  \item[\hskip\labelsep\itshape #1.]\ignorespaces
}{%
  \qed\endtrivlist
}
\makeatother

\makeatletter
\renewcommand{\@fnsymbol}[1]{\number#1}
\makeatother

\usepackage{color-edits}
\addauthor[Zhicheng]{dzc}{red}
\addauthor[HF]{hf}{red}
\addauthor[Zihe]{zh}{red}
\begin{document}
\title{Algorithmic Information Design for Searchers \\ with Uncertain Alternatives\thanks{An earlier version of this paper was titled \emph{Competitive Information Design in Sequential Search}.}}

\author{
Zhicheng Du\thanks{Renmin University of China.  
Email: \texttt{duzhicheng@ruc.edu.cn}}  \and
Hu Fu\thanks{Shanghai University of Finance and Economics.
Email: \texttt{fu.hu.thu@gmail.com}} \and
Ying Qin\thanks{Renmin University of China.
Email: \texttt{qinying0420@ruc.edu.cn}} \and
Zihe Wang\thanks{Renmin University of China.
Email: \texttt{wang.zihe@ruc.edu.cn}}
}
\authorrunning{Z. Du et al.}
\institute{}
\date{\today}
%
%
%
%
\maketitle              
\begin{abstract}

Advertisements reveal information to consumers who decide on further information acquisition and eventual purchase.
\citet{anderson2006advertising} first modeled this problem using an information-design framework where the advertiser acts as a sender and the consumer as a receiver.
Due to search frictions and the consumer's outside option, search for additional information is not always worthwhile for the receiver.
Thus, the sender's information design is used to make further search attractive and ultimately induce purchase.
Following \citet{lyu2023information}, we study the optimal information design for a sender who sells search goods to a consumer with uncertain alternatives.
Instead of making relaxations to the sender's problem (as done in \citealp{lyu2023information}), we work directly on the joint distribution over realized values and signals.
Our contributions are twofold.
First, we give a method, based on duality arguments, to verify whether a given information strategy is optimal.
We illustrate the value of this verification framework in a competitive extension, where it certifies a non-trivial symmetric equilibrium for two senders with a common convex prior.
Second, on the algorithmic front, we develop an FPTAS that finds for the seller a signaling scheme whose utility differs from that of the optimal solution by at most an additive $\epsilon$ error, for any $\epsilon > 0$.




\end{abstract}



\section{Introduction}
\label{sec:intro}

Consumers are often faced with horizontally differentiated products without being fully informed of the attributes offered by different sellers.
This information gap is often bridged, at least partially, by the advertisements posted by sellers to compete for the purchase.
This paper takes an information design perspective to study the design of informative advertising for buyers engaged in a sequential search.


In \citet{anderson2006advertising}'s influential model of \emph{informative advertising}, an advertisement is viewed as a signal sent by the seller to directly, though perhaps only partially, inform a potential buyer about her valuation of a product.
The signal, bound by legal sanctions, is credible.
The unit-demand buyer, with posterior estimates from the advertisements, searches for a purchase. 

In this context, advertisements are only meaningful when there is search friction.  
\citeauthor{anderson2006advertising} assumed the buyer incurs an inspection cost in order to exactly learn her value for any product; 
they study a monopoly seller facing a buyer with a deterministic outside option.
Following \citet{lyu2023information}, we study an extended framework which allows the buyer to have a randomized value for her outside option. 
This extension, as we explain in Section~\ref{subsec:pandora}, makes the framework applicable to a seller in a competitive market for search goods, with the outside option summarizing the competition from all rival sellers.
This single-dimensional summarization is made possible by structural properties of the buyer's optimal search strategy, the celebrated \IndexAlg first discovered by \citet{weitzman1978optimal} for the \emph{Pandora Box} problem.


\citeauthor{anderson2006advertising} studied a seller that strategically sets a price and commits to an advertising signaling scheme.
As in \citet{lyu2023information}, we abstract away the pricing part and only consider strategic signaling: 
a seller gets utility~$1$ when the buyer makes a purchase from him, and $0$ otherwise.\footnote{Strictly speaking, the last part of \citet{lyu2023information} discusses pricing, but in essence does not go beyond saying that the seller's revenue is the product of his price and his demand (which is our utility, inherently a function of the price), and should be optimized by selecting the best price.  The mathematical substance of the work almost entirely lies in the part that considers an exogenous price.}
Unlike the case with deterministic outside option, the information design part in our setting is already very complicated.  
This simplification is also meaningful because decisions on pricing and promotion/marketing are often made by different departments in modern corporations; advertising and pricing may be decoupled decisions (It is for similar reasons that bidders in ad auctions are sometimes modeled as ``value maximizers'' rather than standard utility maximizers; see, e.g., \citealp{WilkensCN16,WilkensCN17,Deng24,Balseiro25}.)

We study the seller's strategic advertising as an instance of information design \citep{kamenica2011bayesian}.
A unit-demand buyer has match value~$v$ for the seller's good; $v$, drawn from a publicly known distribution $F$, is unknown to the buyer initially, but can be learned if the buyer inspects it at a cost $c>0$.
The buyer's value $o$ for her outside option is drawn from a publicly known distribution~$H$, but $o$ is private to the buyer.
The seller aims to maximize the probability that the buyer searches for and selects his item. 
To this end, the sender, prior to any inspection, commits to a signaling scheme that stochastically maps $v$ to a signal.
The signal acts as an advertisement that updates the buyer's belief about $v$ to a posterior distribution.
The buyer implements the Index Algorithm to decide whether to inspect and whether to take the seller's good.
In the language of information design, the buyer is the \emph{receiver}, and the seller the \emph{sender}.

\xhdr{Our Contributions}
We make contributions on two fronts of the sender's persuasion problem. 
First, we develop a procedure for verifying the optimality of the seller's signaling scheme and apply it to identify equilibria in markets with two advertising sellers.
This is the first time such equilibria are characterized in an oligopoly; while \citet{lyu2023information} was only able to identify equilibria when the number of sellers goes to infinity (the large-market limit).
On the algorithmic front, we develop a fully polynomial-time approximation scheme (FPTAS) which finds for the seller a signaling scheme whose yield utility differs from that of the optimal solution by at most an additive $\epsilon$ error, for any $\epsilon > 0$.
Underlying both results is a lossless representation of the sender's signaling as a value-index joint distribution.


As we explain in Section~\ref{sec:prelim}, this two-dimensional representation naturally results from the Index Algorithm.
Working with it directly poses technical difficulty: the feasible set of strategies is given by an infinite-dimensional linear program with a continuum of feasibility constraints.
\citeauthor{lyu2023information} avoided the difficulty by representing the signaling scheme as a \emph{one-dimensional} distribution of the so-called amortized values (see \Cref{defn:kappa}); they imposed on these distributions a relaxed, necessary feasibility condition.
Whereas this relaxation allows the use of tools from the persuasion literature (such as the techniques worked on mean-preserving contractions), it comes with loss.
Generally, solutions to this relaxed program may not be feasible.
\citeauthor{lyu2023information} derived conditions under which the relaxed program yields a feasible optimal solution; these conditions include quasi-concavity (and hence unimodality) of~$H$, the distribution of the outside option value.
Notably, when one applies the framework to study a competitive market with multiple sellers (which motivates the single-seller setting in the first place), a constraint on~$H$ in the single seller setting does not translate into a constraint on the value distribution~$F$, the only input in the market setting.  
This is because $H$ is used to summarize rival sellers' competition; in an equilibrium, it incorporates information about the rival sellers' endogenous signaling. 
Therefore, in a symmetric equilibrium, an assumption on~$H$ is really an assumption on the \emph{solution} of the program.
It is for this reason that \citeauthor{lyu2023information} was able to derive equilibria only when the number of sellers goes to infinity.

We tackle the formidable two-dimensional distribution head-on.
Notice first that our approach imposes neither relaxation nor restriction on the original problem.
On one front, we exploit the geometry of the value-index joint distributions and develop dual certificates for the optimality of signaling schemes (Theorem~\ref{thm:constructed-lambda-mu-optimal}).
The value of this verification framework goes beyond equilibrium certification, although it is one important application.
It is worth pointing out that the two-seller equilibrium we derive (Theorem~\ref{thm:symmetric-equilibrium-convex-prior-big-c}) does not satisfy \citeauthor{lyu2023information}'s unimodality assumption.
On the other front, we develop an additive FPTAS (Theorem~\ref{thm:fptas}) which uses a two-level discretization that performs the two dimensions, realized values and posterior indices, in different ways.
We also study the competitive extension with multiple senders in several additional directions, including equilibrium existence in the general game and comparisons of information disclosure across different competitive settings.
Due to space constraints, these results are deferred to \Cref{apx:existence,apx:sym-equilibrium,apx:comparison}.

\subsection{Related Literature}
\label{subsec:related work}
Our work contributes to two related strands of literature.

\xhdr{Sequential search and Pandora Box}
In our model, the receiver faces an uncertain outside option and a search good controlled by the sender.
She decides whether to pay the cost to learn the item's value and then choose the better option, or to take the outside option directly.
This sequential decision problem can be seen as an instance of the \emph{Pandora Box} problem introduced in \citet{weitzman1978optimal}.
The Pandora Box problem has motivated a rapidly growing algorithmic literature on sequential search and its variants \citep[e.g.][]{doval2018whether, beyhaghi2019pandora, boodaghians2020pandora, chawla2020pandora, FuLiLiu23, BeyhaghiCai23,Hajiaghayi25,banihashem2025delegated}.
We refer to \citet{BeyhaghiCaiSurvey} for a survey.
A central insight of \citet{weitzman1978optimal} is that the receiver's optimal search policy is governed by index values.
The Index Algorithm has also been widely used to model consumer search in markets with competing sellers.
See, for example, \citet{armstrong2017ordered} for a survey, and \citet{derakhshan22ranking} and \citet{friedler25buybox} for recent applications.
Our work differs from this literature by focusing on the sender's optimal pre-search information design.



\xhdr{Information design for search goods}
We model the senders' pre-search information revelation as an instance of Bayesian persuasion \citep{kamenica2011bayesian}.
This framework has inspired a large body of works on information design.
See, e.g., \citealp{bergemann2019information,dughmi2017algorithmic,kamenica2019bayesian} for some recent surveys.

A large strand of the literature studies information provision for experience goods, where the receiver cannot learn the exact value through inspection before purchase and must make the purchase decision based only on the information provided by the seller.
More recent work studies information disclosure and advertising in monopoly and competitive markets, including \citet{ivanov2013information,boleslavsky2017demonstrations,hwang2026competitive}.
Related persuasion tools include the concavification approach of \citet{kamenica2011bayesian}, the linear-programming approach of \citet{kolotilin2018optimal}, and the duality-based approaches of \citet{dworczak2019simple,dworczak2024persuasion}.

However, information design for search goods is less developed.
The main obstacle is that what the sender controls is the receiver's pre-search information.
Thus, posterior means alone are no longer sufficient for determining outcomes due to the inherently complicated search decision.
\citet{anderson2006advertising} initiate this line by studying advertising search goods to consumers with deterministic alternatives.
\citet{anderson2013advertising} extend the framework to allow for richer advertising content.
Several related papers focus on post-search information design instead of pre-search information design, including \citet{bar2012search,board2018competitive,whitmeyer2020persuasion,dogan2022consumer,au2023attraction}.
Closer to pre-search information design for search goods, \citet{choi2019optimal} characterize consumer-optimal and seller-worst signals in a binary-value model, while \citet{hinnosaar2020robust} study seller-worst pre-search information in a robust-pricing environment.
In contrast to these works, we study the seller's optimal pre-search information strategy with a continuous value distribution and a buyer with a privately known outside option.

The closest paper to ours is \citet{lyu2023information}, who first studies optimal pre-search information provision by a search-goods seller with privately known outside options.
He reduces the sender's persuasion problem to a one-dimensional relaxed problem over the distribution of the amortized values, and shows that the relaxed optimum solves the original problem only under unimodality conditions on the outside-option distribution.
Our paper complements this work by keeping the exact value-index formulation without making any relaxation and developing algorithmic tools for the original best-response problem.

\xhdr{Roadmap}
\Cref{sec:prelim} introduces the model and identifies the sender's persuasion problem.
\Cref{sec:verification} develops the dual-certificate framework for verifying optimality of candidate strategies.
\Cref{sec:approximation} provides an additive FPTAS for computing near-optimal information strategies.

\section{Preliminaries}
\label{sec:prelim}

\subsection{The Setup}
\label{subsec:setup}
We consider a single strategic sender (he) who faces a receiver (she) with uncertain alternatives. 
The sender controls one item whose match value to the receiver is $v\in[0,1]$. 
The value is drawn from a commonly known prior $F\in\Delta([0,1])$,\footnote{Throughout the paper, $\Delta(\cdot)$ denotes the set of all probability distributions; unless specified otherwise, a distribution is represented by its cumulative distribution function (CDF).} and is \emph{initially invisible} to the receiver. 


The receiver can pay an inspection cost $c>0$ to learn the realized value $v$ (i.e., the item is a search good). 
The receiver also has an outside option $o\in[0,1]$, drawn from a commonly known distribution $H\in\Delta([0,1])$.
The realized outside option is the receiver's private information. 


Search friction makes the additional information generated by inspection not always worthwhile.
The seller can design a pre-search signal about the realized value $v$ to attract the receiver.
Following the information design literature,  the sender commits to a signaling scheme $(\mathcal S,\pi)$, where $\mathcal S$ is the signal space and $\pi(\cdot\mid v)$ specifies the conditional distribution of signals given the realized value $v$.

The timing of the game is as follows.
\begin{enumerate}
    \item The sender publicly commits to a signaling scheme $(\mathcal S,\pi)$.
    \item The sender observes value $v\sim F$, and sends signal $s\sim \pi(\cdot\mid v)$ to the receiver.
    \item The receiver updates the posterior belief about $v$ after receiving signal $s$.
    \item The receiver observes outside option $o\sim H$, and decides whether to  search for the item.
    \begin{enumerate}
        \item[(4.1)] If not, she takes the outside option $o$ and leaves the game.
        \item[(4.2)] If yes, she incurs a search cost $c>0$ and learns $v$ exactly.
        Then she takes the larger value among $v$ and $o$ and leaves the game.
        Ties are broken uniformly at random.
    \end{enumerate}
\end{enumerate}

\subsection{The Receiver's Optimal Search}
\label{subsec:receiver-optimal-search}
\label{subsec:pandora}

We review \citet{weitzman1978optimal}'s Pandora Box problem which models a sequential searcher's problem and admits an elegant, index-based optimal strategy.
Our purpose is to illustrate (a) that the setup in Section~\ref{subsec:setup} is an abstraction of a sender's problem in a search good market with multiple competing sellers, and (b) that the outcome of the optimal search is completely determined by a comparison between the outside option and the sender's so-called {amortized values}.

In the Pandora Box problem, a searcher is faced with $n$~boxes, each with a value that becomes observable only after an inspection cost is paid.
Specifically, each box~$i$ has value $v_i$ drawn independently from a publicly known distribution~$F_i$, and is revealed at inspection cost~$c_i$.  
The searcher can take at most one box, and must optimize her search order and stopping time to maximize the expected value of the box taken minus the cumulated inspection costs.
\citeauthor{weitzman1978optimal} showed the following \emph{Index Algorithm} to be the optimal strategy.

\begin{definition}[Index]
\label{def:index}
    Box $i$'s \emph{index} $\theta_i$ is the unique solution to 
    \begin{equation}
    \label{eq:weitzman-equation}
        \E_{v_i\sim F_i}[(v_i-\theta_i)_+]=c_i\quad \mathrm{where}\quad (x)_+\triangleq\max \{x, 0\}~.\tag{Index Eq.}
    \end{equation}
\end{definition}

\begin{definition}[The Index Algorithm, \citealp{weitzman1978optimal}] 
Compute indices for all boxes.
Discard boxes with negative indices; sort the rest in decreasing order by their indices.
Inspect the boxes in this order until the largest observed value exceeds the indices of all the boxes not inspected, or when all the boxes have been inspected.
Take the box with the highest revealed value. 
\end{definition}

Both our observations ((a) and (b) above) are consequences of the discovery, by \citet{kleinberg2016descending}, that the searcher implementing the index algorithm always ends up taking the box with the highest \emph{amortized value}: 

\begin{definition}[Amortized value]
\label{defn:kappa}
For box~$i$ with index $\theta_i$ and realized value $v_i$, its \emph{amortized value}  
is $\kappa_i\triangleq \min\{v_i,\theta_i\}$.
\end{definition}

\begin{proposition}[\citealp{kleinberg2016descending}]
\label{prop:amortized}
A searcher implementing the index algorithm in the Pandora Box problem takes the box with the highest nonnegative amortized value.
\end{proposition}

By Proposition~\ref{prop:amortized}, for a seller in a competitive search-good market, the competition from rival sellers can be summarized by one number: the highest amortized value from these sellers.
This number is abstracted as $o$, the value of the outside option, in our setup in Section~\ref{subsec:setup}.
Our seller in question is chosen by the buyer if and only if his own amortized value is greater than~$o$.

When there is signaling, the buyer calculates the index using the posterior distribution, and chooses the sender if{f} his posterior amortized value is greater than~$o$.

\subsection{The Sender's Persuasion Problem}
\label{subsec:sender-persuasion-problem}
Given the receiver's optimal search rule, i.e., the Index Algorithm, the sender can reshape the receiver's belief about the invisible item and affect the corresponding index computed by the receiver through information disclosure.
Now, we formalize the sender's persuasion problem.

A signaling scheme is a joint distribution over realized values and signals. 
Under the Index Algorithm, both the sender's and the receiver’s utilities are determined by the realized values and the posterior indices. 
Since signals inducing the same posterior index are payoff-equivalent, we identify information strategies with joint distributions $G$ over the value-index pair $(v,\theta)$.

\begin{definition}[Feasible joint distribution]
\label{def:joint-distribution}
    A joint distribution $G\in\Delta([0,1]\times[-c,1-c])$ is feasible if its value marginal is $F$ and, for $G_\theta$-almost every $\theta\in[-c,1-c]$,\footnote{For such a joint distribution $G$, we write $G_\theta$ for its marginal distribution over indices.
    We write $G_{\theta\mid v}$ and $G_{v\mid\theta}$ for the conditional distributions of $\theta$ given $v$ and of $v$ given $\theta$, respectively.
}
$$
\int_0^1 \big((v-\theta)_+-c\big)\dd G(v\mid \theta)=0~.
$$
We denote the set of feasible joint distributions by $\mathcal G(F,c)$, which is convex and weakly compact.\footnote{See \Cref{pf:convex-compact} for the formal statement and the proof.}
\end{definition}

Our original game is payoff-equivalent to the following reduced game:
\begin{enumerate}
    \item The sender chooses a feasible joint distribution $G\in \mathcal G(F,c)$ as his strategy.
    \item The sender's value-index pair $(v, \theta)$ is drawn from~$G$;  
    its amortized value $\kappa$ is $\min\{v, \theta\}$.
    \item 
    The outside option $o$ is drawn from $H$.
    If $\kappa \ge o$, the receiver takes the sender's item; otherwise, she takes the outside option. 
\end{enumerate}

For the sender's utility, he gains utility $1$ if chosen by the receiver, and utility $0$ otherwise.
For a realized amortized value $\kappa\in[-c,1-c]$, the sender's interim utility is the probability of the realized outside option $o$ being smaller than $\kappa$, 
so we define
\begin{equation*}
    u(\kappa)\triangleq
    \begin{cases}
        0\quad & \mathrm{if\ }\kappa \in[-c,0)~,\\
        H(\kappa)\quad & \mathrm{if\ }\kappa \in[0,1-c]~.
    \end{cases}
\end{equation*}
The sender's persuasion problem is to choose an information strategy $G$ among the joint distribution space $\mathcal G(F,c)$ to maximize the expected utility:
\begin{equation}
\label{eq:best-response}
    \max_{G\in\mathcal G(F,c)} U(G)\triangleq \mathbb E_{(v,\theta)\sim G} \left[ u\left(\min\{v,\theta\}\right) \right]~.\tag{$\mathcal{P}_{\cc{BR}}$}
\end{equation}
It is an infinite-dimensional linear program (LP) over joint value-index distributions, subject to a continuum of feasibility constraints. 
Consequently, neither standard finite-dimensional LP methods nor conventional one-dimensional persuasion techniques apply directly.

\section{Optimal Strategy Verification}
\label{sec:verification}

In this section, we develop a dual-certificate approach to the sender's problem.
In \Cref{subsec:duality}, we show the existence of optimal solutions for both the primal and dual problems (\Cref{thm:optimal-solution-exists-in-dual-new}), and derive a version of complementary slackness conditions (\Cref{thm:primal-dual-optimal}).
In \Cref{subsec:verification}, we develop a method for deciding if a given primal solution is optimal (\Cref{thm:constructed-lambda-mu-optimal}); this is by constructing a dual solution (\Cref{alg:lambda,alg:mu}).
Finally, we apply the framework to a competitive extension and certify a symmetric equilibrium for two senders with a common convex prior (\Cref{thm:symmetric-equilibrium-convex-prior-big-c}).

\subsection{Duality and Optimality Certificates}
\label{subsec:duality}
We formulate the sender's problem \eqref{eq:best-response} as the primal problem. 
Define $p(v,\theta)\triangleq u(\min\{v,\theta\})$ and $q(v,\theta)\triangleq (v-\theta)_+-c$ for each value-index pair $(v,\theta)$.
To derive the dual problem, we introduce a dual variable $\lambda(v)$ for the value-marginal constraint at each $v\in[0,1]$, and a dual variable $\mu(\theta)$ for the index constraint at each $\theta\in[-c,1-c]$.
Then, 
the dual problem is
\begin{align}
\label{eq:dual-obj}
\tag{$\mathcal{D}_{\cc{BR}}$}
\min_{\lambda,\ \mu} \quad & \int_0^1 \lambda(v)f(v)\dd v\\
\text{s.t.} \quad & \label{eq:dual-con1}\lambda(v)+\mu(\theta)q(v,\theta)\ge p(v,\theta)~,\quad \forall \ (v,\theta)
\end{align}

For continuous interim utility function $u(\kappa)$, 
the primal must have an optimal solution $G^*$ by the weak compactness of the feasible space $\mathcal{G}(F,c)$ (\Cref{lem:convex-compact}).
We next show that, under mild assumptions, the dual also has an optimal solution, with additional well-behaved properties.

\begin{lemma}[Existence of optimal dual solutions]
\label{thm:optimal-solution-exists-in-dual-new}
If function $u(\kappa)$ is $L$-Lipschitz continuous for some $L>0$,\footnote{
Let $\mathcal X$ be a subset of a Euclidean space, a function $h:\mathcal X\to\mathbb R$ is $L$-Lipschitz continuous for some $L>0$ if
$|h(x)-h(y)|\le L\|x-y\|$ for all $x,y\in\mathcal X$.
}
the dual problem~\eqref{eq:dual-obj} has an optimal solution $(\lambda^*,\mu^*)$ such that (i) $\mu^*(\theta)\in[-L,0]$ for each $\theta\in[-c,1-c]$; and (ii) $\lambda^*$ is non-decreasing and continuous over $[0,1]$.
\end{lemma}

Such properties make the subsequent construction of dual variables tractable.
Complementary slackness also holds, and can be used to verify whether a given pair of feasible solutions is optimal both to the primal and dual problems.

\begin{lemma}[Complementary slackness]
\label{thm:primal-dual-optimal}
    Given a feasible primal solution $G^*$ and a feasible dual solution $(\lambda^*,\mu^*)$, $G^*$ and $(\lambda^*,\mu^*)$ are both optimal to the primal and dual if and only if
    \begin{equation}
    \label{eq:primal-dual-optimal}
    \int_{[0,1]\times[-c,1-c]}
    \left[
    \lambda^*(v)+\mu^*(\theta)q(v,\theta)-p(v,\theta)
    \right]
    \dd G^*(v,\theta)
    =
    0~.
    \end{equation}
\end{lemma}


The lemma links primal supports with binding dual constraints.
Dual feasibility implies that the slack $\lambda^*(v)+\mu^*(\theta)q(v,\theta)-p(v,\theta)$ is nonnegative everywhere.
Hence, Condition~\eqref{eq:primal-dual-optimal} implies that this slack must vanish $G^*$-almost everywhere.
Equivalently, for $G^*$-almost every pair $(v,\theta)$,
\begin{equation}
    \lambda^*(v)=\sup_{\theta'\in[-c,1-c]}\left\{p(v,\theta')-\mu^*(\theta')q(v,\theta')\right\}
    =p(v,\theta)-\mu^*(\theta)q(v,\theta)~,
\end{equation}
where $p(v,\theta)$ is the payoff of the value-index pair $(v,\theta)$, and $-\mu^*(\theta)q(v,\theta)$ captures its spillover effect on the other values with the same index.
Thus, $p(v,\theta)-\mu^*(\theta)q(v,\theta)$
can be viewed as the amortized payoff generated by $(v,\theta)$, while $\lambda^*(v)$ is the maximal amortized payoff among all indices.
The optimal strategy assigns mass only to indices that maximize this amortized payoff.

\subsection{Optimality Verification Procedures}
\label{subsec:verification}
Using techniques above, we develop a method to verify the optimality of a given primal solution.
Suppose the prior $F$ is atomless and strictly increasing over $[0,1]$.


We first derive a necessary condition for optimality to rule out some candidates.
Recall that $u(\kappa)=0$ whenever $\kappa<0$.
Thus low values within $[0,c]$ can generate positive utility only if they are pooled with higher values to induce nonnegative amortized values.
However, values sufficiently close to zero are not worth pooling: raising them to nonnegative amortized values requires too much distortion of higher values relative to the payoff they can generate.
The following lemma formalizes this observation through a threshold value.

\begin{lemma}[Positive threshold]
\label{lem:no-threshold-value}
For any feasible strategy $G\in\mathcal G(F,c)$, define
\begin{align*}
T\triangleq
\Bigl\{v\in[0,c]~\Bigm|~
&\supp(G_{v\mid\theta})\cap[0,v)=\emptyset
\ \  \mathrm{for\ all\ }\theta\in(v-c,1-c]~,\\
&\supp(G_{v\mid\theta})\cap(v,1]=\emptyset
\ \ \mathrm{for\ all\ }\theta\in[-c,v-c)
\Bigr\}~.
\end{align*}
The threshold $\underline v\triangleq\sup T$ if $T\neq \emptyset$; and $\underline v\triangleq 0$ otherwise.
If $G^*$ is optimal, then $\underline v>0$.
\end{lemma}



For candidate strategies that survive, we propose \Cref{alg:lambda,alg:mu}, which construct dual variables $\lambda$ and $\mu$ respectively (\Cref{alg:mu} is in \Cref{pf:alg-2}).
If the two algorithms terminate successfully, the input and their are both optimal for the respective programs.

\begin{theorem}[Optimality verification]
\label{thm:constructed-lambda-mu-optimal}
Assume that:
(i) $u(\kappa)$ is differentiable and Lipschitz continuous everywhere; and 
(ii) for every $(v,\theta)\in\supp(G)$, any open neighborhood of $(v,\theta)$ has strictly positive probability under $G$.
If $\underline{v}>0$ and a feasible solution $(\lambda,\mu)$ can be constructed by \Cref{alg:lambda,alg:mu}, then $G$ and $(\lambda,\mu)$ are optimal for the primal and dual programs, respectively.
\end{theorem}

Specifically, \Cref{alg:lambda} determines $\lambda(v)$ for each $v\in[0,1]$ using spatial relationships between the dual variables, and \Cref{alg:mu} checks the existence of $\mu(\theta)$
that constitutes a feasible dual solution along with $\lambda$ constructed.
\Cref{alg:mu} is relatively straightforward,
so we focus on explaining the construction of $\lambda$ in the main text.
A few definitions are in order.

\begin{definition}[Monotone support sequence]
\label{def:monotone-support-sequence}
Given a joint distribution $G$, a \emph{monotone support sequence} of $G$ is a sequence $\{(v^m,\theta^m)\}_{m\in\mathbb{Z}_+}\subseteq\supp(G)$ such that $v^m<v^{m+1}$ for every $m\in \mathbb{Z}_+$, and the sequence $\{\theta^m\}_{m\in \mathbb{Z}_+}$ is either nondecreasing or nonincreasing.
If $\{\theta^m\}_{m\in \mathbb{Z}_+}$ is strictly monotone, then $\{(v^m,\theta^m)\}_{m\in \mathbb{Z}_+}$ is called a \emph{strictly monotone support sequence}.
\end{definition}


\begin{algorithm}[h]
\caption{LAMBDA $(v)$: Construction of the dual variable $\lambda$.}
\label{alg:lambda}
\SetKwInOut{Require}{Input}
\SetKwInOut{Ensure}{Output}

\Require{~Candidate strategy $G$; interim utility $u$; cost $c$; threshold $\underline{v}$; and any value $v \in [0,1]$.}
\Ensure{~The corresponding value of the dual variable $\lambda(v)$.}

\BlankLine

\If{$v \in [0, \underline{v}]$}{
    $\lambda(v) \gets 0$ and \textbf{return} $\lambda(v)$ \tcp*[r]{Case 1}
}

\If{$(v, v-c) \in \mathrm{supp}(G)$}{
    $\lambda(v) \gets u(v-c)$ and \textbf{return} $\lambda(v)$ \tcp*[r]{Case 2}
}

\If{there is a monotone support sequence converging to $(v, v-c)$}{
    $\lambda(v) \gets u(v-c)$ and \textbf{return} $\lambda(v)$ \tcp*[r]{Case 3.1}
}
\ElseIf{there is a strictly monotone support sequence converging to $(v, \theta_1)$ with $\theta_1 \neq v-c$}{
    $\lambda(v) \gets u'(\theta_1)q(v, \theta_1) + p(v, \theta_1)$ and \textbf{return} $\lambda(v)$ \tcp*[r]{Case 3.2}
}
\Else{
    There is a $\theta_2 \neq v-c$ such that $(v-\epsilon, v) \subseteq \mathrm{supp}(G_{v|\theta_2})$ for some $\epsilon > 0$ \;
    $\underline{v}_{\theta_2} \gets \inf \mathrm{supp}(G_{v|\theta_2})$ \;
    $\lambda(v) \gets \frac{\text{LAMBDA}(\underline{v}_{\theta_2}) - p(\underline{v}_{\theta_2}, \theta_2)}{q(\underline{v}_{\theta_2}, \theta_2)} q(v, \theta_2) + p(v, \theta_2)$ and \textbf{return} $\lambda(v)$ \tcp*[r]{Case 3.3}
}
\end{algorithm}

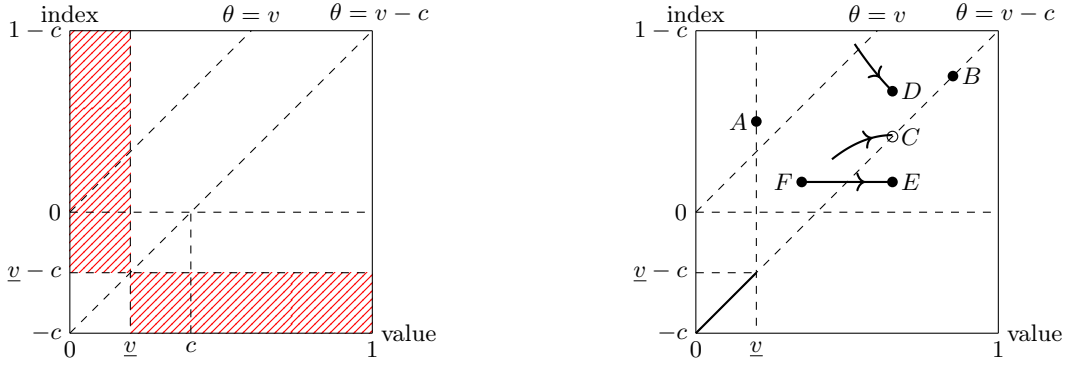
\begin{figure}[htbp]
    \centering
    \begin{minipage}[t]{0.48\textwidth}
        \centering
        \begin{tikzpicture}[scale=2.0]
        \fill[pattern=north east lines, pattern color=red] (0,-0.4) rectangle (0.4,1.2);
        \fill[pattern=north east lines, pattern color=red] (0.4,-0.8) rectangle (2,-0.4);
        \draw (0,0) node[left] {\footnotesize $0$};
        \draw (0,-0.8) node[below] {\footnotesize $0$};
        \draw (2,-0.8) node[below] {\footnotesize $1$};
        \draw (0,1.2) node[left] {\footnotesize $1-c$};
        \draw (0.8,-0.8) node[below] {\footnotesize $c$};
        \draw (0,-0.8) -- (0,1.2) node[above] {\footnotesize index};
        \draw (2,-0.8) -- (2,1.2);
        \draw (0,-0.8) -- (2,-0.8) node[right] {\footnotesize value};
        \draw (2,1.2) -- (0,1.2);
        \draw (0,-0.8) node[left] {\footnotesize $-c$};
        
        \draw[dashed] (0,0) -- (1.2,1.2) node[above] {\footnotesize $\theta=v$};
        \draw[dashed] (0,-0.8) -- (2,1.2) node[above] {\footnotesize $\;\; \theta=v-c$};
        \draw[dashed] (0,0) -- (2,0) node[above] {};
        \draw[dashed] (0.8,-0.8) -- (0.8,0) node[above] {};
        \draw[dashed] (0.4,1.2) -- (0.4,-0.8) node[below] {\footnotesize \textcolor{black}{$\underline{v}$}};
        \draw[dashed] (0,-0.4) node[left] {\footnotesize $\underline{v}-c$} -- (2,-0.4);

    \end{tikzpicture}
    \end{minipage}%
    \hfill
    \begin{minipage}[t]{0.48\textwidth}
        \centering
        \begin{tikzpicture}[scale=2.0,,decoration={markings, mark=at position 0.7 with {\arrow{>}}}]
        \draw (0,0) node[left] {\footnotesize $0$};
        \draw (0,-0.8) node[below] {\footnotesize $0$};
        \draw (2,-0.8) node[below] {\footnotesize $1$};
        \draw (0,1.2) node[left] {\footnotesize $1-c$};
        \draw (0,-0.8) -- (0,1.2) node[above] {\footnotesize index};
        \draw (2,-0.8) -- (2,1.2);
        \draw (0,-0.8) -- (2,-0.8) node[right] {\footnotesize value};
        \draw (2,1.2) -- (0,1.2);
        \draw (0,-0.8) node[left] {\footnotesize $-c$};
        
        \draw[dashed] (0,0) -- (1.2,1.2) node[above] {\footnotesize $\theta=v$};
        \draw[dashed] (0,-0.8) -- (2,1.2) node[above] {\footnotesize $\;\; \theta=v-c$};
        \draw[dashed] (0,0) -- (2,0) node[above] {};
        \draw[thick] (0,-0.8) -- (0.4,-0.4) node[above] {};
        
        \fill[black] (1.7,0.9) circle (1pt) node[right] {\footnotesize $B$};

        \draw (1.3,0.5) circle (1pt) node[right] {\footnotesize $C$};
        \draw[thick, postaction={decorate}, domain=0.9:1.3, samples=50] plot (\x, {-(\x - 1.3)^2+0.51});

        \fill[black] (1.3,0.8) circle (1pt) node[right] {\footnotesize $D$};\
        \draw[thick, postaction={decorate}, domain=1.05:1.3, samples=50] plot (\x, {(\x - 1.8)^2+0.55});
        
        \fill[black] (1.3,0.2) circle (1pt) node[right] {\footnotesize $E$};
        \draw[thick,, postaction={decorate}] (0.7,0.2) -- (1.3,0.2);
        \fill[black] (0.7,0.2) circle (1pt) node[left] {\footnotesize $F$};
        
        \fill[black] (0.4,0.6) circle (1pt) node[left] {\footnotesize $A$};
        \draw[dashed] (0.4,1.2) -- (0.4,-0.8) node[below] {\footnotesize $\underline{v}$};
        \draw[dashed] (0,-0.4) node[left] {\footnotesize $\underline{v}-c$} -- (0.4,-0.4);

        
    \end{tikzpicture}
    \end{minipage}
    \caption{
    The left panel shows the threshold value $\underline{v}$.
    The right panel illustrates \Cref{alg:lambda}. 
    The solid dots and black curves represent the support set of one candidate strategy~$G$. 
    The hollow dots represent the limit points of some monotone support sequence. 
    Case 1: Let $\lambda(v_A)=0$ since $v_A\le\underline{v}$;
    Case 2: Let $\lambda(v_B)=u(v_B-c)$ since $\theta_B=v_B-c$;
    Case 3.1: Let $\lambda(v_C)=u(v_C-c)$ since there is a monotone support sequence converging to $(v_C,\theta_C)$;
    Case 3.2: Let $\lambda(v_D)=u'(\theta_D)q(v_D,\theta_D)+p(v_D,\theta_D)$ since there is a strictly monotone support sequence converging to $(v_D,\theta_D)$;
    Case 3.3: Let $\lambda(v_E)=\frac{\lambda(v_F)-p(v_F,\theta_F)}{q(v_F,\theta_F)}q(v_E,\theta_E)+p(v_E,\theta_E)$ since $(v_F,\theta_E)\subseteq\supp(G_{v|\theta_E})$.
    }
    \label{fig:procedure}
\end{figure}

The construction of $\lambda$ is based on necessary conditions for $G$ and $(\lambda, \mu)$ to be optimal.
In other words, the deductions in the following are driven by the assumption that $G$ is optimal, and that a corresponding $\lambda$ satisfying the complementary slackness condition with $G$ can be constructed.
We provide an illustration for \Cref{alg:lambda} in the right panel of \Cref{fig:procedure}.
By our previous discussion following the complementary slackness conditions, if $G$ and $(\lambda,\mu)$ are both optimal, $\lambda(v)$ represents value $v$'s contribution to the utility.
For any $v \in [0,1]$, 
the determination of $\lambda(v)$ falls into one of five cases.
In the first two cases, we can directly determine $\lambda(v)$:  
\begin{itemize}
    \item \textbf{Case 1:}
    If $v\in[0,\underline v]$, then \Cref{alg:lambda} returns $\lambda(v)=0$.
    By \Cref{lem:no-threshold-value}, values at or below the threshold $\underline v$ make no contribution to the sender's payoff.
    \item \textbf{Case 2:}
    If $(v,v-c)\in\supp(G)$, then \Cref{alg:lambda} returns $\lambda(v)=u(v-c)$.
    Condition $\lambda(v)+\mu(v-c)q(v,v-c)-p(v,v-c)=0$ implies $\lambda(v)=u(v-c)$ since $q(v,v-c)=0$.
\end{itemize}

Now consider a value $v$ that falls into neither of these two cases.
As the prior distribution is assumed to have positive density everywhere on $[0,1]$, there must exist a monotone support sequence converging to the point $(v,\theta)$ for some $\theta\in[0,1-c]$.
Depending on the convergence properties of this sequence, we further divide the analysis into three cases:
\begin{itemize}
    \item \textbf{Case 3.1:}
    If there is a monotone support sequence $\{(v^m,\theta^m)\}_{m\in\Zplus}$ converging to the point $(v,v-c)$,
    by the continuity of $\lambda$ and the boundedness of $\mu$ (see \Cref{thm:optimal-solution-exists-in-dual-new}), we have $\lambda(v)=\lim\nolimits_{m\rightarrow \infty}-\mu(\theta^m)q(v^m,\theta^m)+p(v^m,\theta^m)=u(v-c)$.
    \item \textbf{Case 3.2:}
    If there is a strictly monotone support sequence $\{(v^m,\theta^m)\}_{m\in\Zplus}$ converging to the point $(v,\theta_1)$ for some $\theta_1\neq v-c$.
    By \Cref{thm:primal-dual-optimal}, $\lambda(v^m)=-\mu(\theta^m)q(v^m,\theta^m)+p(v^m,\theta^m)$ for each $m\in\Zplus$.
    The converging sequence gives infinitely many such equalities in the neighborhood of $(v,\theta_1)$.
    These equalities, combined with the specific formula of $-\mu(\theta^m)q(v^m,\theta^m)+p(v^m,\theta^m)$, allows us to conclude that $\lambda(v)=u'(\theta)q(v,\theta_1)+p(v,\theta_1)$.

    \item \textbf{Case 3.3:}
    If none of the above conditions hold, 
    there must exist $\theta_2\neq v-c$ such that $(v-\epsilon,v)\subseteq \supp(G_{v|\theta_2})$ for some $\epsilon>0$. 
    By the continuity of $\lambda$, we reduce the task of solving $\lambda(v)$ to solving $\lambda(\inf\supp(G_{\cdot|\theta_2}))$ and run \Cref{alg:lambda} again to solve this case.
\end{itemize}

\subsection{Application: Certifying Symmetric Equilibrium}

We now use the verification framework to certify symmetric equilibria in a competitive version of the model.
We first extend the single-sender model to a multi-sender game, which can be viewed as a Pandora Box problem where each box is controlled by a strategic sender.

\xhdr{Competitive information design game}
There are $n\ge2$ senders, each controlling one item.
All senders share the same prior $F$ and cost $c>0$.
Each sender $i$ chooses an information strategy $G_i\in\mathcal G(F,c)$.
The value-index pair $(v_i,\theta_i)$ is independently drawn from $G_i$ with the amortized value $\kappa_i\triangleq\min\{v_i,\theta_i\}$.
The receiver selects a sender with the highest nonnegative amortized value.
We denote the amortized value distribution under some strategy $G$ by
$$
K_G(\kappa)
\triangleq
\Pr
\left[
\min\{v,\theta\}\le\kappa
\right]=G(\kappa,1-c)+G(1,\kappa)-G(\kappa,\kappa)~.
$$
If all opponents use $G$, the induced interim utility of an individual sender is\footnote{
This expression ignores ties since ties cannot exist under equilibrium in most scenarios we consider.
}
$$
u_G(\kappa)
\triangleq
\begin{cases}
0\quad & \mathrm{if\ }\kappa<0~,\\
K_G(\kappa)^{N-1}\quad & \mathrm{if\ }\kappa\ge 0~.
\end{cases}
$$
A feasible strategy $G^*\in\mathcal G(F,c)$ forms a \emph{symmetric equilibrium} if and only if $G^*$ is a best response strategy to the competitive environment $u_{G^*}$.


Notice that \citet{lyu2023information} also studies competition but only in the large-market limit.
His dimensionality-reduction approach can not characterize finite-seller markets, which is precisely where our verification framework working on the joint distributions over $(v,\theta)$ becomes useful.

We consider a game where all senders share a common convex prior $F$.\footnote{We also provide results about the symmetric equilibrium under a common concave prior; see \Cref{apx:sym-equilibrium}.}
There are two cases depending on the inspection cost.
When $c\le \inf\supp(F)$, all realized values induce nonnegative full-revelation indices, and full revelation is the unique symmetric equilibrium; this case has a simpler characterization, so we defer it to \Cref{subsec:low-cost-equilibrium}.
We focus here on the high-cost case, where $c>\inf\supp(F)$.
For a two-sender game with a common convex prior $F$, we construct a symmetric equilibrium whose induced amortized-value distribution has a hinge-convex structure.

\begin{theorem}[Hinge-convex signaling equilibrium]
\label{thm:symmetric-equilibrium-convex-prior-big-c}
For two senders with a common convex prior~$F$ and cost $c>\inf \supp(F)$, there is a symmetric equilibrium $G^*$ that induces each sender's amortized value distribution $K$ of the following hinge-convex structure:
for certain parameters $\theta_1 \in (-c, 0)$, $\theta_2 \in (0, 1 - c]$, and $\rho > 0$ (uniquely determined in the proof),
\begin{equation*}
    K^*(\kappa)=
    \begin{cases}
        F(\kappa+c)\ &\emph{if}\ \kappa\in[-c,\theta_1]~,\\
        F(\theta_1+c)\ &\emph{if}\ \kappa\in(\theta_1,0]~,\\
        \min \left\{
        \rho \cdot\kappa+F(\theta_1+c), 1 \right\}\ 
        &\emph{if}\ \kappa\in(0,\theta_2]~, \\
        F(\kappa+c)\ &\emph{if}\ \kappa\in(\theta_2,1-c]~,
    \end{cases}
\end{equation*}
where $K^*(\theta_2) = F(\theta_2 + c)$ if $\theta_2 < 1 - c$ (the non-degenerate case); and 
$K^*(\theta_2) = 1$ if $\theta_2 = 1 - c$ (the degenerate case).

\end{theorem}

In this equilibrium, values below $\theta_1+c$ are fully revealed and yield zero utility.
Values in $(\theta_1+c,c)$ are pooled with higher values to generate positive amortized values, which are spread uniformly over an interval starting from $0$.
In the non-degenerate case, pooling stops at value $\theta_2+c$, and values above are fully revealed (see \Cref{fig:convex-eq-right}).
In the degenerate case, all values above $\theta_1+c$ join the pooling region, so the upper full-revelation region disappears.

\begin{figure}[ht]
    \centering
    \begin{minipage}[t]{0.48\textwidth}
        \centering
        \begin{tikzpicture}[scale=2.0]
            \draw[->] (-1.2,0) -- (1.5,0) node[right] {\footnotesize $\kappa$};
            \draw[-] (0,0) -- (0,2.2) node[left] {};
            \draw[->] (0,0) -- (0,2.3) node[left] {};

            \draw[rounded corners=3pt, fill=white, draw=black] (-1.3,1.3) rectangle (-0.3,1.8);
        
            \draw[red, thick] (-1.25,1.65) -- (-1.0,1.65);
            \draw (-1.05,1.65) node[anchor=west, font=\footnotesize] {$K^*(\kappa)$};
        
            \draw[black, thick] (-1.25,1.45) -- (-1.0,1.45);
            \draw (-1.05,1.45) node[anchor=west, font=\footnotesize] {$F(\kappa+c)$};
            \draw[thick, domain=-0.8:1.2, samples=50] plot (\x, {-(7-(\x+1.38113)^2)^0.5+2.5811});
            \draw[red, thick] (0,0.226) -- (0.9, 1.24078);
            \draw[red, thick, domain=-0.8:-0.2, samples=50] plot (\x, {-(7-(\x+1.38113)^2)^0.5+2.5811});
            \draw[red, thick, domain=0.9:1.2, samples=50] plot (\x, {-(7-(\x+1.38113)^2)^0.5+2.5811});
            \draw[red, thick] (-0.2,0.222) -- (0,0.222);
            \draw[dashed] (0,0.226) -- (-0.2,0);
            
            \draw (-0.8,0) node[below] {\footnotesize $-c$};
            \draw (0,0) node[below] {\footnotesize $0$};
            \draw[dashed] (-0.2,0.226) -- (-0.2,0) node[below] {\footnotesize $\theta_1$};
            \draw[dashed] (0.9,1.24078) -- (0.9,0) node[below] {\footnotesize $\theta_2$};
            \draw[dashed] (0.6,0.906) -- (0.6,0) node[below] {\footnotesize $\theta^*$};
            \draw[dashed] (1.2,2) -- (1.2,0);
            \draw (1.25,0) node[below] {\footnotesize $1-c$};
            \draw[dashed] (1.2,2) -- (0,2) node[left] {\footnotesize $1$};
        \end{tikzpicture}
    \end{minipage}
    \hspace{0.02\textwidth}
    \begin{minipage}[t]{0.48\textwidth}
        \centering
        \begin{tikzpicture}[scale=2.0]
            \fill[blue!10] (0.01,-0.79) rectangle (0.6,-0.2); 
            \fill[green!10] (0.6,-0.2) rectangle (1.7,0.9); 
            \fill[yellow!20] (1.7,0.9) rectangle (1.99,1.19);  
            
            \draw[dashed] (0,0) -- (2,0);
            \draw (-0.05,0) node[left] {\footnotesize $0$};
            \draw (0,-0.8) node[below] {\footnotesize $0$};
            \draw (0,1.2) node[left] {\footnotesize $1-c$};
            \draw (2,-0.8) node[below] {\footnotesize $1$};
            \draw (0,-0.8) -- (0,1.2) node[above] {\footnotesize index};
            \draw (2,-0.8) -- (2,1.2);
            \draw (0,-0.8) -- (2,-0.8) node[right] {\footnotesize value};
            \draw (2,1.2) -- (0,1.2);

            \draw[dashed] (0,0) -- (1.2,1.2) node[above] {\footnotesize $\theta=v$};
            \draw[dashed] (0,-0.8) -- (2,1.2) node[above] {\footnotesize $\;\; \theta=v-c$};
            
            \draw[red, thick] (1.7,0.9) -- (2,1.2);
            \draw[red, thick] (0,-0.8) -- (0.6,-0.2);
            \draw[red, thick, domain=0.6:1.4, samples=50] plot (\x, {0.79*(max(0, 0.5*\x-0.3))^0.3});
            \draw[red, thick] (1.4,0.6) -- (1.7,0.9);
            \draw[red, thick, domain=1.4:1.7, samples=50] plot (\x, {-22*(\x-1.4)^3+0.6});
            
            \draw[dashed] (0,-0.8) -- (0,-0.8) node[left] {\footnotesize $-c$};
            \draw[dashed] (1.7,0.9) -- (0,0.9) node[left] {\footnotesize $\theta_2$};
            \draw[dashed] (1.7,0.9) -- (1.7,-0.8) node[below, xshift=2pt] {\footnotesize $\theta_2+c\;\;$};
            \draw[dashed] (0.8,0) -- (0.8,-0.8) node[below, yshift=-2pt] {\footnotesize $c$};
            \draw[dashed] (0.6,-0.2) -- (0,-0.2) node[left] {\footnotesize $\theta_1$};
            \draw[dashed] (1.4,0.6) -- (1.4,-0.8);
            \draw (1.2,-0.8) node[below, xshift=2pt] {\footnotesize $\theta^*+c$};
            \draw[dashed] (0.6,0) -- (0.6,-0.8) node[below, xshift=-1pt] {\footnotesize $\theta_1+c\;\;$};
            \draw[dashed] (1.4,0.6) -- (0, 0.6) node[left] {\footnotesize $\theta^*$};
        \end{tikzpicture}
    \end{minipage}
        \caption{An example of the non-degenerate case in \Cref{thm:symmetric-equilibrium-convex-prior-big-c}.
        In the left panel, the red curve represents the amortized value distribution $K^*(\kappa)$ in equilibrium, and the black curve the shifted prior $F(\kappa+c)$.
        In the right panel, the red curves represent the support sets of the equilibrium strategy $G^*$.
        In the blue and yellow regions, low values and high values are fully revealed; while in the green region, lower values in $(\theta_1+c,\theta^*+c)$ are paired with higher values in $(\theta^*+c,\theta_2+c)$ to produce amortized values uniformly distributed over an interval.
        }
        \label{fig:convex-eq-right}
\end{figure}
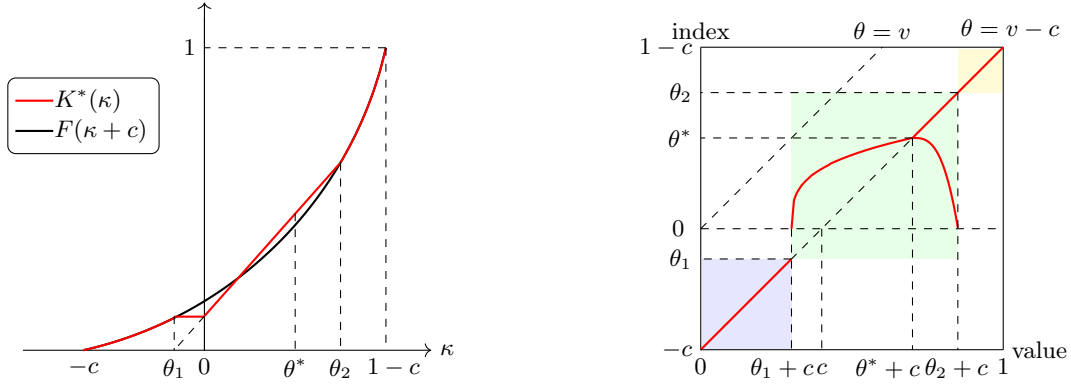

Our findings echo a common pattern in real-world advertising. 
Firms often use such a \emph{three-fold advertising regime}: they highlight the product (by truthfully revealing) when the consumer--product match is clearly good, stay silent when the match is clearly poor, and otherwise blur the match quality---benefiting most when rivals face similarly mediocre matches.

Our proof follows the \emph{identify--construct--verify} procedure.
We first identify the candidate equilibrium amortized-value distribution $K^*$, then construct a feasible joint distribution $G^*$ that implements it, and finally construct dual variables $(\lambda^*,\mu^*)$ via \Cref{alg:lambda,alg:mu} and invoke \Cref{thm:constructed-lambda-mu-optimal} to certify that the constructed strategy is a best response to itself.

\xhdr{Step 1: Identify the candidate amortized value distribution $K^*$}
For the target structure of the amortized value distribution $K^*$ in \Cref{thm:symmetric-equilibrium-convex-prior-big-c}, values in $(\theta_1+c,\theta_2+c]$ are pooled so that the induced amortized values are uniformly distributed over $[0,\theta_1']$ for some parameter $\theta_1'\le\theta_2$ to be determined.
These thresholds are pinned down by the following two equalities.

First, notice that the amortized value distribution $K$ of any feasible joint distribution forms a mean-preserving contraction (MPC) of the shifted prior $F(\kappa+c)$.\footnote{A distribution $H\in\Delta([0,1])$ forms an MPC of a distribution~$F\in\Delta([0,1])$ if and only if $\int_{0}^t F(x) \dd x\ge\int_{0}^t G(x) \dd x$ for any $t\in[0,1]$ with equality at $t=1$.
See Lemma~3.1 in \citet{lyu2023information} for a formal statement.}
Due to this observation, it holds $K^*(\theta_1)=F(\theta_1+c)$ and $K^*(\theta_1')=F(\theta_1'+c)$.
Because $K^*$ is assessed to be linear with slope $\rho$ on this region, we obtain the first equality: 
$$
\rho\theta_1'=F(\theta_1'+c)-F(\theta_1+c)~.
$$

Second, consider the boundary value $\theta_1+c$.
At equilibrium, this value must be indifferent between being pooled with higher values to produce amortized value $0$ and being truthfully revealed to produce amortized value $\theta_1<0$.
Thus, the gain from pooling, $F(\theta_1+c)$, must equal the loss from lowering higher values' amortized values by $|\theta_1|$.
Since $K^*$ is linear with slope $\rho$ in the pooled region, this loss is $\rho|\theta_1|$, giving the second equality: 
$$
F(\theta_1 + c) = \rho|\theta_1|~.
$$

We finally prove that there uniquely exist $\theta_1$ and $\theta_1'$ that solve the resulting system obtained by combining the above two equalities.
$$
\frac{F(\theta_1+c)}{|\theta_1|}=\frac{F(\theta_1'+c)-F(\theta_1+c)}{\theta_1'}=\rho~.
$$
The solution yields two cases:
If $K^*(\theta_1')<1$, the equilibrium is non-degenerate with $\theta_2=\theta_1'$; otherwise, it is degenerate with $\theta_2=1-c$.
We sketch the remaining proof for the more involved non-degenerate case.  

\xhdr{Step 2: Construct a joint distribution $G^*$ that implements $K^*$}
To implement the target amortized-value distribution $K^*$, we fully reveal values below $\theta_1+c$ and above $\theta_2+c$, and split the intermediate region $[\theta_1+c,\theta_2+c]$ into low and high parts.
Since the density $f$ is nondecreasing and $F(\theta_2+c)-F(\theta_1+c)=K^*(\theta_2)-K^*(\theta_1)=\rho\theta_2$, there uniquely exists a cutoff $\theta^*\in(\theta_1,\theta_2)$ such that each value above $\theta^*+c$ has density of size at least $\rho$. 
We call such values \emph{high}, and the remaining values in the interval \emph{low}.

High values serve two roles in the construction.
Each high value $v$ is truthfully revealed with probability $\rho/f(v)$, and its remaining mass is pooled with low values.
The pairing must simultaneously generate the desired amortized value distribution and preserve both the value-marginal constraint and the row-wise index constraints.

We construct the pairing through a system of ODEs. 
Each index $\theta\in[0,\theta^*]$ pools one low value $\alpha(\theta)$ with one high value $\beta(\theta)$. 
The initial condition pairs the lowest and highest values: $\alpha(0)=\theta_1+c$ and $\beta(0)=\theta_2+c$.
The ODEs are chosen so that each row satisfies the index equation and the induced amortized values have the target density. 
Their solution satisfies $\alpha'(\theta)\ge0$ and $\beta'(\theta)\le0$ for all $\theta\in[0,\theta^*]$, and terminates at $\alpha(\theta^*)=\beta(\theta^*)=\theta^*+c$.


\xhdr{Step 3: Certify equilibrium via the verification framework}
It remains to show that the constructed strategy is indeed a best response to itself. This is the point at which the verification framework becomes essential. Rather than directly solving the sender's best-response problem, given the interim utility induced by the hinge-shaped amortized-value distribution $K^*$, we construct dual functions $(\lambda^*,\mu^*)$ using \Cref{alg:lambda,alg:mu}. 
We then invoke \Cref{thm:constructed-lambda-mu-optimal} to certify the optimality of the constructed strategy and dual functions.





\section{Optimal Strategy Approximation}
\label{sec:approximation}

In this section, we analyze the sender's best-response problem from a complementary perspective.
We provide an additive FPTAS to compute a near-optimal strategy (\Cref{thm:fptas}).

\begin{theorem}[An additive FPTAS]
\label{thm:fptas}
Suppose the interim utility function $u(\kappa)$ is $L$-Lipschitz continuous for some $L>0$.
For every $\epsilon>0$, \Cref{alg:continuous-fptas} computes a feasible strategy $\widetilde G\in\mathcal G(F,c)$ whose payoff is at least $\mathrm{OPT}-\epsilon$ in polynomial time in $L$ and $1/\epsilon$.
In other words, the sender's persuasion problem \eqref{eq:best-response} admits a fully polynomial-time approximation scheme (FPTAS).
\end{theorem}

The theorem complements the verification framework developed above.
Since we keep the joint distributions over value-index pairs as the decision variables,
we have to preserve both the value marginal and the row-wise index constraints throughout the approximation.

Our approximation scheme uses a \emph{two-level} discretization.
We first discretize the value space to obtain a discrete prior, and then discretize the index support to obtain a finite LP (see \Cref{fig:approximation} for illustrations).
The proof compares three optimization problems: the original continuous-prior problem, the discretized-prior problem, and the finite LP after index discretization.
We solve for the finite LP and lift the discrete optimal solution back to the feasible region under the original prior as the output of our algorithm.
The key step is to bound the payoff loss in each transition between the three problems above.
We provide the proof sketch below.


\xhdr{Proof Sketch}
The first problem is the original sender's persuasion problem \eqref{eq:best-response}, which admits an optimal solution $G^*$ with optimal value $\mathsf{OPT}$.

The second problem is obtained by discretizing the value prior.
For a grid size $\delta\in(0,1)$, we uniformly partition the value space $[0,1]$ into intervals $\{I_r\}_{i\in M}$.
We then move all probability mass of $F$ in each interval $I_r$ to its midpoint $v_r$, thereby obtaining a discrete prior 
\begin{equation}
\label{eq:discrete-prior}
    F_\delta\triangleq\sum_{i=1}^M f_i^\delta\cdot \delta_{v_i}
    \quad \mathrm{where}\quad
    f_i^\delta\triangleq F(I_i)~,
\end{equation}
and $\delta_{v_i}$ denotes the Dirac measure at $v_i$.
We define the optimal value under the optimum $G^*_\delta$
\begin{equation}
    \mathsf{OPT}_\delta\triangleq\max\nolimits_{G_\delta\in\mathcal G(F_\delta,c)} U(G_\delta)=\mathbb E_{(v,\theta)\sim G_\delta} \left[ u\left(\min\{v,\theta\}\right) \right]~.\tag{$\mathcal{P}^\delta_{\cc{BR}}$} 
\end{equation}

The two problems are not directly comparable: their feasible sets are $\mathcal G(F,c)$ and $\mathcal G(F_\delta,c)$, neither of which contains the other.
Thus, there is no immediate ordering between their optimal values.
Nevertheless, their optimal values satisfy the following two-sided stability bound:
\begin{equation}
\label{eq:bound-1}
    |\mathsf{OPT}-\mathsf{OPT}_\delta|\le L\delta~.
\end{equation}
To prove this, we construct feasible solutions in both directions.
Starting from $G^*$, we move each row's value mass to the corresponding grid midpoints and recompute the row index, obtaining a feasible solution under $F_\delta$.
Since both values and recomputed indices move by at most $\delta$, the Lipschitz continuity of $u$ implies a payoff loss of at most $L\delta$.
Conversely, starting from $G_\delta^*$, we spread each midpoint mass back to its original value interval according to the conditional prior and recompute row indices.
The same argument also gives a payoff loss of at most $L\delta$.

The third problem is a finite LP obtained by further discretizing the index support.
Let $v_1<\cdots<v_M$ denote the support of the discretized prior $F_\delta$.
Define $\Theta_\delta\triangleq [v_1-c,v_M-c]$ as the effective index domain under $F_\delta$.
Fix a granularity $\epsilon\in(0,1)$, we first discretize the payoff range of the interim utility $u(\kappa)$.
For each $\ell_k\triangleq k\epsilon$, we compute the corresponding generalized inverse $\xi_k
\triangleq
\inf\{\theta\in\Theta_\delta:u(\theta)\ge \ell_k\}$.
We then identify the full index grids:
\begin{equation}
\label{eq:index-grid}
    \Gamma
    \triangleq
    \{v_i-c:i\in[M]\}
    \cup
    \left(\{v_i:i\in[M]\}\cap\Theta_\delta\right)
    \cup
    \{\xi_k:k=0,1,\ldots,\lceil 1/\epsilon\rceil\}~.
\end{equation}
Denote $R\triangleq |\Gamma|$.
The points $v_i$ ensure that the mapping $\theta\mapsto\min\{v_i,\theta\}$ has no kink inside any grid interval.
The points $v_i-c$ ensure that the function $q_i(\theta)=(v_i-\theta)_+-c$ does not change sign inside any grid interval.
The payoff-grid points $\xi_k$ ensure that moving mass between adjacent grid indices causes at most $\epsilon$ utility loss.
Let $\boldsymbol x=(x_{ir})_{i\in[M],r\in[R]}$ be the decision variable, where $x_{ir}\ge 0$ denotes the probability mass assigned to $(v_i,\theta_r)$.
We consider the following finite LP:
\begin{equation}
   \tag{$\mathcal{P}^{\delta,\epsilon}_{\cc{BR}}$} 
\begin{aligned}
\mathsf{OPT}_{\delta,\epsilon}
\triangleq
\max_{\boldsymbol x\ge0}\quad
&\sum_{i=1}^M\sum_{r=1}^R
u(\min\{v_i,\theta_r\})\ x_{ir}\\
\mathrm{s.t.}\quad
&\sum_{r=1}^R x_{ir}=f_i^\delta~,
&&\forall i\in[M]~,\\
&\sum_{i=1}^M
\big((v_i-\theta_r)_+-c\big)x_{ir}=0~,
&&\forall r\in[R]~.
\end{aligned} 
\end{equation}
Since every solution is also feasible for the discretized-prior problem \eqref{eq:best-response-delta}, it holds $\mathsf{OPT}_{\delta,\epsilon}\le\mathsf{OPT}_\delta$.
The converse approximation follows from a rounding argument.
Any feasible strategy under $F_\delta$ can be rounded to the grid $\Gamma$ by splitting each non-grid row between the two adjacent grid rows.
The construction ensures that the total loss is at most $\epsilon$.
Hence
\begin{equation}
\label{eq:bound-2}
    \mathsf{OPT}_\delta-\epsilon
    \le
    \mathsf{OPT}_{\delta,\epsilon}
    \le
    \mathsf{OPT}_\delta~.
\end{equation}

Finally, we solve the finite LP and lift its optimal solution $G_{\delta,\epsilon}^*$ back to the original prior $F$.
By the same lifting argument as above, this step incurs at most $L\delta$ payoff loss.
Combining the loss bounds \eqref{eq:bound-1} and \eqref{eq:bound-2} achieved above, the resulting feasible strategy $\widetilde G\in\mathcal G(F,c)$ satisfies
$$
U(\widetilde G)
\ge
\mathsf{OPT}_{\delta,\epsilon}-L\delta
\ge
\mathsf{OPT}-\epsilon-2L\delta~.
$$
Thus, the constructed strategy $\widetilde G$ is the output of our FPTAS, which achieves any constant factor approximation to the optimal solution $G^*$.

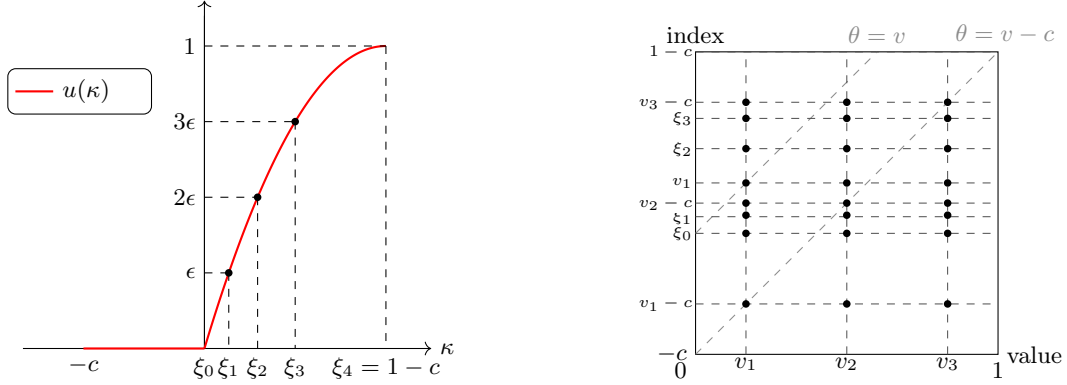
\begin{figure}[ht]
    \centering
    \begin{minipage}[t]{0.48\textwidth}
        \centering
        \begin{tikzpicture}[scale=2.0]
    
            \draw[->] (-1.2,0) -- (1.5,0) node[right] {\footnotesize $\kappa$};
            \draw[-] (0,0) -- (0,2.2) node[left] {};
            \draw[->] (0,0) -- (0,2.3) node[left] {};
    
            \draw[rounded corners=3pt, fill=white, draw=black] (-1.3,1.55) rectangle (-0.35,1.85);
            \draw[red, thick] (-1.25,1.70) -- (-1.00,1.70);
            \draw (-1.05,1.70) node[anchor=west, font=\footnotesize] {$\ u(\kappa)$};
    
            \draw[red, thick] (-0.8,0) -- (0,0);
            \draw[red, thick, domain=0:1.2, samples=100]
                plot (\x,{2*(2*(\x/1.2)-(\x/1.2)^2)});
    
            \draw[dashed] (0,0.5) -- (0.1608,0.5);
            \draw[dashed] (0,1.0) -- (0.3515,1.0);
            \draw[dashed] (0,1.5) -- (0.6000,1.5);
            \draw[dashed] (0,2.0) -- (1.2000,2.0);
    
            \draw[dashed] (0.1608,0.5) -- (0.1608,0);
            \draw[dashed] (0.3515,1.0) -- (0.3515,0);
            \draw[dashed] (0.6000,1.5) -- (0.6000,0);
            \draw[dashed] (1.2000,2.0) -- (1.2000,0);
    
            \fill[black] (0.1608,0.5) circle (0.025);
            \fill[black] (0.3515,1.0) circle (0.025);
            \fill[black] (0.6000,1.5) circle (0.025);
    
            \node[left] at (0,0.5) {\footnotesize $\epsilon$};
            \node[left] at (0,1.0) {\footnotesize $2\epsilon$};
            \node[left] at (0,1.5) {\footnotesize $3\epsilon$};
            \node[left] at (0,2.0) {\footnotesize $1$};
    
            \draw (-0.8,0) node[below] {\footnotesize $-c$};
            \draw (0,0) node[below] {\footnotesize $\xi_0$};
            \draw (0.1608,0) node[below] {\footnotesize $\xi_1$};
            \draw (0.3515,0) node[below] {\footnotesize $\xi_2$};
            \draw (0.6000,0) node[below] {\footnotesize $\xi_3$};
            \draw (1.20,0) node[below] {\footnotesize $\xi_4=1-c$};
        \end{tikzpicture}
    \end{minipage}
    \hspace{0.02\textwidth}
    \begin{minipage}[t]{0.48\textwidth}
        \centering
        \begin{tikzpicture}[scale=2.0]
            \draw[black] (0,-0.8) rectangle (2,1.2);
    
            \draw[dashed, black!45] (0,0) -- (1.2,1.2)
                node[above] {\footnotesize $\theta=v$};
            \draw[dashed, black!45] (0,-0.8) -- (2,1.2)
                node[above] {\footnotesize $\;\;\theta=v-c$};
    
            \foreach \y/\lab in {
                -0.4667/$ v_1-c$,
                 0.0000/$\xi_0$,
                 0.1100/$\xi_1$,
                 0.2000/$ v_2-c$,
                 0.3333/$ v_1$,
                 0.5600/$\xi_2$,
                 0.7600/$\xi_3$,
                 0.8667/$ v_3-c$,
                 1.2000/$1-c$
            }{
                \draw[dashed, black!70, line width=0.25pt] (0,\y) -- (2,\y);
                \node[left, fill=white, inner sep=1pt] at (0,\y) {\tiny \lab};
            }
    
            \foreach \x/\lab in {
                0.3333/$ v_1$,
                1.0000/$ v_2$,
                1.6667/$ v_3$
            }{
                \draw[dashed, black!70, line width=0.25pt] (\x,-0.8) -- (\x,1.2);
                \node[below, fill=white, inner sep=1pt] at (\x,-0.8) {\footnotesize \lab};
            }
    
            \node[below left] at (0,-0.8) {\footnotesize $0$};
            \node[below] at (2,-0.8) {\footnotesize $1$};
            \node[left] at (0,-0.8) {\footnotesize $-c$};
            \node[above] at (0,1.2) {\footnotesize index};
            \node[right] at (2,-0.8) {\footnotesize value};
    
            \foreach \x in {0.3333,1.0000,1.6667}{
                \foreach \y in {-0.4667,0.0000,0.1200,0.2000,0.3333,0.5600,0.7600,0.8667}{
                    \fill[black] (\x,\y) circle (0.025);
                }
            }
        \end{tikzpicture}
    \end{minipage}
        \caption{
        Illustration of the two-level discretization in \Cref{thm:fptas}.
        The left panel shows the payoff-grid construction for the interim utility $u(\kappa)$.
        The right panel shows the resulting finite linear program supported on the value-index grids after discretizations.
        The black points denote the supports of the discrete solution under the finite LP.
        }
        \label{fig:approximation}
\end{figure}

\begin{algorithm}[ht]
\caption{An additive FPTAS}
\label{alg:continuous-fptas}
\DontPrintSemicolon
\SetAlgoLined
\SetKwInOut{Input}{Input}
\SetKwInOut{Output}{Output}

\Input{
Prior $F\in\Delta([0,1])$;
Inspection cost $c>0$;
$L$-Lipschitz interim utility $u(\kappa)$;
Target error $\bar\epsilon\in(0,1)$.
}

\Output{
A feasible strategy $\widetilde G\in\mathcal G(F,c)$ with payoff at least $\mathsf{OPT}-\bar\epsilon$.
}

Set the discretization parameters $\epsilon\gets \bar\epsilon/2$ and $\delta\gets \min\left\{1/2,\bar\epsilon/4L\right\}$\;

Construct the discretized prior $F_\delta$ according to \eqref{eq:discrete-prior} under parameter $\delta$\;

Construct the index grids $\Gamma$ according to \eqref{eq:index-grid} under parameter $\epsilon$\;

Construct and solve the finite linear program \eqref{eq:best-response-delta-epsilon}. Let $G_{\delta,\epsilon}^*$ be an optimal solution\;

Lift $G_{\delta,\epsilon}^*$ back to the original prior $F$ by spreading each mass at $v_i$ over $I_i$, and obtain $\widetilde G$\;

\Return{$\widetilde G$}\;

\end{algorithm}

\bibliographystyle{plainnat}
\bibliography{reference}

\newpage
\appendix
\section{Proofs in \Cref{sec:prelim}}
\label{apx:preli}

\subsection{\Cref{lem:convex-compact} and Its Proof}
\label{pf:convex-compact}
\begin{lemma}
\label{lem:convex-compact}
For any prior $F$ and cost $c$, the action space $\mathcal{G}(F,c)$ is convex and weakly compact.
\end{lemma}

\begin{proof}[Proof of \Cref{lem:convex-compact}]
We divide the proof into two steps.

First, we show that the strategy space $\mathcal{G}(F,c)$ is convex given any prior $F$ and cost $c$.
Consider two feasible strategies $G\in\mathcal{G}(F,c)$ and $H\in\mathcal{G}(F,c)$ (with densities $g$ and $h$).
We want to prove that the convex combination $K(x)=\lambda\cdot G(x)+(1-\lambda)\cdot H(x)$ for any $x\in[0,1]$ (with density $\lambda\cdot g+(1-\lambda)\cdot h$) also belongs to space $\mathcal{G}(F,c)$ for any $\lambda\in[0,1]$.

To see this, first, for all indices $\theta\in[-c,1-c]$, it holds
\begin{equation*}
    \lambda\cdot \mathbb{E}_{v\sim G_{\cdot|\theta}}[\max\{v-\theta,0\}]+(1-\lambda)\cdot\mathbb{E}_{v\sim H_{\cdot|\theta}}[\max\{v-\theta,0\}]=c~.
\end{equation*}
Besides, for each value $v\in[0,1]$, it holds
\begin{equation*}
    \lambda\cdot\int_{\theta=-c}^{1-c}g_{\cdot|v}(\theta)\dd v+(1-\lambda)\cdot\int_{\theta=-c}^{1-c}h_{\cdot|v}(\theta)\dd \theta=f(v)~.
\end{equation*}
Therefore, we have proved that any convex combination of the two feasible strategies also satisfies all the constraints of space $\mathcal{G}(F,c)$, which makes $\mathcal{G}(F,c)$ convex.

Second, we show that the strategy space $\mathcal{G}(F,c)$ is also compact.
The boundedness of space $\mathcal{G}(F,c)$ can be directly obtained by the fact that space $\mathcal{G}(F,c)$ is indeed a measure space over $[0,1]\times[-c,1-c]$.
Since $[0,1]\times [-c,1-c]$ is a compact metric space, the space of probability measures $\Delta([0,1]\times [-c,1-c])$ is compact in the weak topology by Prokhorov's theorem. 

Next, we prove that the strategy space $\mathcal{G}(F_i,c_i)$ is weak closed.
We assume that there exists a sequence $\{G^m\}_{m\in\mathbb{Z}}$ with each $G^m\in \mathcal{G}(F,c)$, weakly converges to a certain 2-D distribution $G$, we shall prove the limit point $G$ also belongs to set $\mathcal{G}(F,c)$. 

To see this, we first check the marginal constraint (Bayesian plausibility). 
Since each $G^m\in \mathcal{G}(F,c)$, we have 
\[
\int_{\theta=-c}^{1-c} g^m(v,\theta) \dd\theta =f(v)~, \quad \forall v\in[0,1]~.
\]
It is equivalent to the following condition: for all bounded and continuous function $\phi$ defined on $[0,1]$, we have
\[
\int_{(v,\theta)}\phi(v)\dd G^m(v,\theta)=\int_{v=0}^1 \phi(v)\dd F(v)~.
\]
According to the definition of weak convergence, we have
\[
\lim\limits_{m\to \infty}\int_{(v,\theta)}\phi(v)\dd G^m(v,\theta)=\int_{(v,\theta)}\phi(v)\dd G(v,\theta)=\int_{v=0}^1 \phi(v)\dd F(v)~,
\]
which is further equivalent to 
\[
\int_{\theta=-c}^{1-c} g(v,\theta) \dd\theta =f(v)~, \quad \forall v\in[0,1]~.
\]

We then check the index constraint, which is 
\[
\int^1_{v=0}\left((v-\theta)_+-c\right)g^m(v,\theta)\dd v =0~,\quad\forall \theta\in[-c,1-c]~.
\]
It is equivalent to the following condition: for all bounded and continuous function $\psi$ defined on $[-c,1-c]$, we have
\begin{align*}
&\int_{\theta=-c}^{1-c}\psi(\theta)\left(\int^1_{v=0} ((v-\theta)_+-c)\dd G^m(v,\theta)\right)\dd\theta\\
=&\int_{(v,\theta)}\psi(\theta)((v-\theta)_+-c)\dd G^m(v,\theta)=0~.  
\end{align*}
Both functions $\psi$ and $(v-\theta)_+-c$ are bounded and continuous in $\theta$ and $(v,\theta)$, respectively. 
According to the definition of weak convergence, we have
\[
\lim\limits_{m\to\infty}\int_{(v,\theta)}\psi(\theta)((v-\theta)_+-c)\dd G^m(v,\theta)=\int_{(v,\theta)}\psi(\theta)((v-\theta)_+-c)\dd G(v,\theta)=0~.
\]
This fact is equivalent to 
\[
\int^1_{v=0}((v-\theta)_+-c)g(v,\theta)\dd v =0~, \quad \forall \theta\in[-c,1-c]~.
\]

Thus we have proved that the limit point 
$G$ must belong to space $\mathcal{G}(F,c)$, which makes the space closed.
To sum up, we have achieved the boundedness and closeness of space $\mathcal{G}(F,c)$ under weak convergence, and we can directly make space $\mathcal{G}(F,c)$ weak compact through the Heine-Borel Theorem.
\end{proof}

\section{Proofs in \Cref{sec:verification}}
\label{apx:verification}

\subsection{Proof of \Cref{thm:optimal-solution-exists-in-dual-new}}
\label{pf:dual-existence}

Before proving \Cref{thm:optimal-solution-exists-in-dual-new}, we have to introduce the following \Cref{lem:lambda bigger than phi,lem:mu less than zero,lem:lambda increasing} with respect to the properties of dual variables $\lambda(v)$ and $\mu(\theta)$.
\begin{lemma}
\label{lem:lambda bigger than phi}
For any dual solution $(\lambda,\mu)$, it holds that $\lambda(v)\ge u(v-c)\ge 0$ for each $v\in[0,1]$.
\end{lemma}
\begin{proof}[Proof of \Cref{lem:lambda bigger than phi}]
For each $v\in[0,1]$, by the fact that tuple $(v,v-c)$ satisfies Condition (\ref{eq:dual-con1}), we directly have $\lambda(v)\ge u(v-c)\ge 0$.
\end{proof}

\begin{lemma}
\label{lem:mu less than zero}
Assume there exists a constant $L>0$ such that the interim utility function $u$ is $L$-Lipschitz continuous over $[-c,1-c]$.
If there exists an optimal solution $(\lambda^*,\mu^*)$ to the dual problem \ref{eq:dual-obj}, we can without loss assume that $\mu^*(\theta)\in [-L,0]$ for each $\theta\in[-c,1-c]$.
\end{lemma}
\begin{proof}[Proof of \Cref{lem:mu less than zero}]
First, we prove $\mu^*(\theta)\le 0$ for each $\theta\in[-c,1-c]$ by contradiction.
We assume that there exists an optimal dual solution $(\lambda^*,\mu^*)$ with
$\mu^*(\overline{\theta})>0$ for some $\overline{\theta}\in[-c,1-c]$.
If $\lambda^*(v)> p(v,\overline{\theta})-\mu^*(\overline{\theta})q(v,\overline{\theta})$ holds for all $v\in[\overline{\theta},1]$, then assuming $\mu^*(\overline{\theta})=0$ is without loss of generality, since this assumption does not affect the dual objective while keeping the constraints hold.
If there exists a value $\overline{v}>\overline{\theta}$ with $\lambda^*(\overline{v})= p(\overline{v},\overline{\theta})-\mu^*(\overline{\theta})q(\overline{v},\overline{\theta})$, 
then for any $\theta\ge \overline{v}$, by Constraint (\ref{eq:dual-con1}), it holds that $\lambda^*(\overline{v})\ge u(\overline{v})$.
For the tuple $(\overline{v},\overline{\theta})$, also by Constraint (\ref{eq:dual-con1}) and the monotonicity of $u$, we know that $\lambda^*(\overline{v})=u(\overline{\theta})-\mu^*(\overline{\theta})(\overline{v}-\theta)<u(\overline{\theta})<u(\overline{v})$, which forms a contradiction.
To sum up, we have proved that, for any optimal solution $(\lambda^*,\mu^*)$ to the dual problem, we can without loss assume that $\mu^*(\theta)\le 0$ for each $\theta\in[-c,1-c]$.

Second, we prove that for each $\theta\in[-c,1-c]$, $\mu^*(\theta)$ is bounded by the constant $L$.
The Lipschitz continuity of function $u$ implies that $|u(x)-u(y)|\le L\cdot |x-y|$ for any $x,y\in[-c,1-c]$.
We assume that there exists a $\theta_0\in(-c,1-c)$ (the case of $\theta_0=-c$ or $1-c$ is trivial) with $\mu(\theta_0)<-L$.
Then we know that $u(v-c)\ge -\mu(\theta_0)q(v,\theta_0)+p(v,\theta_0)$ for any $v\in[0,\theta_0+c]$ and $u(v-c)\le -\mu(\theta_0)q(v,\theta_0)+p(v,\theta_0)$ for any $v\in[\theta_0+c,1]$.
So setting $\mu(\theta)=-L$ will weakly further decrease the objective function.
To sum up, we have proved that $-\mu(\theta)\le L$, that is function $\mu$ is bounded by the constant $L$. 
Directly by Constraint (\ref{eq:dual-con1}), we know that the function $\lambda$ corresponding to the bounded function $\mu$, is also bounded.
\end{proof}

\begin{lemma}
\label{lem:lambda increasing}
Assume there exists a constant $L>0$ such that the interim utility function $u$ is $L$-Lipschitz continuous over $[-c,1-c]$.
If there exists an optimal solution $(\lambda^*,\mu^*)$ to the dual problem \ref{eq:dual-obj}, then the function $\lambda^*(\cdot)$ is non-decreasing and continuous over $[0,1]$.
\end{lemma}
\begin{proof}[Proof of \Cref{lem:lambda increasing}]
Suppose $(\lambda^*,\mu^*)$ is optimal to the dual problem \ref{eq:dual-obj}.
By Constraint (\ref{eq:dual-con1}) and the objective of the dual problem, we have
\begin{equation*}
    \lambda^*(v)=\max_{\theta\in[-c,1-c]}-\mu^*(\theta)q(v,\theta)+p(v,\theta)~,\quad \forall v\in[0,1]~.
\end{equation*}
Given a certain $\theta\in[-c,1-c]$, for all $v\in[0,1]$, we define
\begin{equation*}
    \hat{L}_\theta(v)\triangleq-\mu^*(\theta)q(v,\theta)+p(v,\theta)=
    \begin{cases}
        -\mu^*(\theta)(v-\theta-c)+u(\theta) \quad & \text{if }v\ge \theta~, \\
        u(v)+c\mu^*(\theta) \quad & \text{if } v<\theta~.
    \end{cases}
\end{equation*}
We observe that function $\hat{L}_\theta(\cdot)$ is linear over $[\theta,1]$ and forms a shifted version of function $u$ over $[0,\theta]$.
Besides, $\lambda^*(v)=\max_{\theta\in[-c,1-c]}\hat{L}_\theta(v)$ for each $v\in[0,1]$.
By \Cref{lem:mu less than zero}, it is obvious that function $\hat{L}_\theta$ is non-decreasing over $[0,1]$ for each $\theta\in[-c,1-c]$.
By the fact that $\lambda^*(v)=\max_{\theta\in[-c,1-c]}\hat{L}_\theta(v)$ for each $v\in[0,1]$, we have the function $\lambda^*(\cdot)$ is also non-decreasing over $[0,1]$.
Since the dual variable $\mu^*$ is a bounded function (\Cref{lem:mu less than zero}) and the function $u$ is continuous over $[-c,1-c]$, we know that each $\hat{L}_\theta(\cdot)$ is continuous over $[0,1]$ for any $\theta\in[-c,1-c]$.
This further implies that the function $\lambda^*(\cdot)$ is also continuous over $[0,1]$.
\end{proof}

With these lemmas above, now we can prove \Cref{thm:optimal-solution-exists-in-dual-new}.

\begin{proof}[Proof of \Cref{thm:optimal-solution-exists-in-dual-new}]
Based on the necessary conditions of the optimal solution to the dual problem (\Cref{lem:mu less than zero,lem:lambda increasing}), the dual problem \ref{eq:dual-obj} is equivalent to the following convex optimization problem \ref{eq:new-dual-obj}:
\begin{align}
\label{eq:new-dual-obj}
\tag{$\mathcal{D}_{\cc{NEW}}$}
\min_{\lambda,\ \mu} \quad & \int_0^1 \lambda(v)f(v)\ \mathrm{d}v\\
\text{subject to} \quad & \label{eq:new-dual-con1}\lambda(v)=\max_{\theta\in[-c,1-c]}\left\{p(v,\theta)-\mu(\theta)q(v,\theta)\right\}~,\quad \forall v\in[0,1]\\
&\label{eq:new-dual-con2}\mu(\theta)\le 0,\quad \forall \theta\in[-c,1-c]
\end{align}
Any optimal solution $(\lambda^*,\mu^*)$ to the problem \ref{eq:new-dual-obj} also forms an optimal solution to the dual problem \ref{eq:dual-obj}, and vice versa.
Thus, it suffices to prove that \ref{eq:new-dual-obj} has an optimal solution, and for any optimal solution $(\lambda^*,\mu^*)$ it holds: (i) $\mu^*(\theta)\in [-L,0]$ for any $\theta\in[-c,1-c]$; and (ii) $\lambda^*$ is non-decreasing and continuous over $[-c,1-c]$.
We divide the whole proof into two steps.

\xhdr{Step-1: There exists an optimal solution $(\lambda^*,\mu^*)$ to the problem \ref{eq:new-dual-obj}}
Let 
$$
\text{OPT}=\inf_{\lambda,\mu}\int_0^1\lambda(v)f(v)\dd v
$$ denote the optimal value of the dual problem.
Let $\text{OBJ}(\lambda,\mu)$ denote the objective value of any pair of feasible solutions $(\lambda,\mu)$.
We aim to show that there exists a feasible solution $(\lambda^*,\mu^*)$ with $\text{OBJ}(\lambda^*,\mu^*)=\text{OPT}$.
Whether the dual problem has an optimal solution or not, 
there always exists a feasible solution sequence 
$\{(\lambda^m,\mu^m)\}_{m\in \Zplus}$ such that the corresponding objective sequence $\{\text{OBJ}(\lambda^m,\mu^m)\}_{m\in\Zplus}$ converges to the optimal value $\text{OPT}^*$.

By \Cref{lem:lambda increasing}, we know that function $\lambda$ is non-decreasing.
Besides, without loss of generality, we can assume that function $\lambda$ is also bounded.
Combining these two facts, the sequence $\{\lambda^m\}_{m\in\Zplus}$ also weakly converges to certain function $\lambda^*$ that satisfies $\int_0^1\lambda^*(v)f(v)\dd v=\text{OPT}^*$.
Next, we show that there exists a function $\mu^*$ that forms a feasible solution with the function $\lambda^*$.
For any $m\in\Zplus$, each pair $(\lambda^m,\mu^m)$ is feasible.
So by Constraint (\ref{eq:dual-con1}), for any $\theta\in[-c,1-c]$, we have that
\begin{equation*}
    \max_{v\in[\theta+c,1]}\frac{u(\theta)-\lambda^m(v)}{v-\theta-c}\le \mu^m(\theta)\le\min\left\{\min_{v\in[0,\theta)}\frac{\lambda^m(v)-u(v)}{c},\min_{v\in[\theta,\theta+c)}\frac{u(\theta)-\lambda^m(v)}{v-\theta-c}\right\}~.
\end{equation*}
Since the pair $(\lambda^m,\mu^m)$ is feasible, we have 
\begin{equation*}
    \max_{v\in[\theta+c,1]}\frac{u(\theta)-\lambda^m(v)}{v-\theta-c}\le \min\left\{\min_{v\in[0,\theta)}\frac{\lambda^m(v)-u(v)}{c},\min_{v\in[\theta,\theta+c)}\frac{u(\theta)-\lambda^m(v)}{v-\theta-c}\right\}~,
\end{equation*}
and such $\mu^m$ exists.
When $m\rightarrow\infty$, we have
\begin{equation*}
    \lim\limits_{m\rightarrow \infty}\max_{v\in[\theta+c,1]}\frac{u(\theta)-\lambda^m(v)}{v-\theta-c}=\max_{v\in[\theta+c,1]}\frac{u(\theta)-\lambda^*(v)}{v-\theta-c}~,
\end{equation*}
and we have
\begin{align*}
    &\lim\limits_{m\rightarrow \infty}\min\left\{\min_{v\in[0,\theta)}\frac{\lambda^m(v)-u(v)}{c},\min_{v\in[\theta,\theta+c)}\frac{u(\theta)-\lambda^m(v)}{v-\theta-c}\right\}\\
    =&\min\left\{\min_{v\in[0,\theta)}\frac{\lambda^*(v)-u(v)}{c},\min_{v\in[\theta,\theta+c)}\frac{u(\theta)-\lambda^*(v)}{v-\theta-c}\right\}~.
\end{align*}
Thus, we know that there also exists such $\mu^*$ that forms a feasible solution along with function $\lambda^*$.
To sum up, we have proved that there exists an optimal solution $(\lambda^*,\mu^*)$ to the problem \ref{eq:new-dual-obj}.
By the equivalence of \ref{eq:new-dual-obj} and \ref{eq:dual-obj}, there also exists an optimal solution to the dual problem \ref{eq:dual-obj}.

\xhdr{Step-2: For any optimal solution $(\lambda^*,\mu^*)$ to the dual problem \ref{eq:dual-obj}, $\mu^*(\theta)\in [-L,0]$ for any $\theta\in[-c,1-c]$, and $\lambda^*$ is non-decreasing and continuous over $[-c,1-c]$}
These conditions hold directly by \Cref{lem:mu less than zero,lem:lambda increasing}.
\end{proof}

\subsection{Proof of \Cref{thm:primal-dual-optimal}}
\label{pf:primal-dual-optimal}

\begin{proof}[Proof of \Cref{thm:primal-dual-optimal}]
First, we prove the sufficiency of \Cref{thm:primal-dual-optimal}.
The dual variables $\lambda(v)$ and $\mu(\theta)$ are both bounded measurable functions, so the order of integration can be interchanged for any double integral over $[0,1]\times[-c,1-c]$ under Fubini's Theorem.
It is obvious that both the primal and the dual have a feasible solution.
We consider any pair of feasible solutions $(g,\lambda,\mu)$.
First, by the value-marginal constraint, we have that 
\begin{equation}
    \label{eq:proof-br-sufficient-1}\int_{v=0}^1\int_{\theta=-c}^{1-c}\lambda(v)g(v,\theta)\dd v\dd \theta=\int_0^1\lambda(v)f(v)\dd v~.
\end{equation}
By the row-wise index constraint, we have that 
\begin{equation}
    \label{eq:proof-br-sufficient-2}\int_{v=0}^1\int_{\theta=-c}^{1-c}\mu(\theta)q(v,\theta)g(v,\theta)\dd v\dd\theta=\int_{\theta=-c}^{1-c}\int_{v=0}^1\mu(\theta)q(v,\theta)g(v,\theta)\dd v\dd\theta=0~.
\end{equation}
After combining Equation (\ref{eq:proof-br-sufficient-1}) and Equation (\ref{eq:proof-br-sufficient-2}), we achieve that
\begin{equation}
    \label{eq:proof-br-sufficient-3}\int_{v=0}^1\int_{\theta=-c}^{1-c}[\lambda(v)+\mu(\theta)q(v,\theta)]\ g(v,\theta)\dd v\dd\theta=\int_0^1\lambda(v)f(v)\dd v~.
\end{equation}
By Constraint (\ref{eq:dual-con1}), we have that
\begin{equation}
    \label{eq:proof-br-sufficient-4}\int_{v=0}^1\int_{\theta=-c}^{1-c}[\lambda(v)+\mu(\theta)q(v,\theta)]g(v,\theta)\dd v\dd\theta\ge \int_{v=0}^1\int_{\theta=-c}^{1-c}p(v,\theta)g(v,\theta)\dd v\dd\theta~.
\end{equation}
Next, we assume that there exists a pair of feasible solutions $(g^*,\lambda^*,\mu^*)$ that satisfies Condition (\ref{eq:primal-dual-optimal}), that is
\begin{equation}
    \label{eq:proof-br-sufficient-5}\int_{v=0}^1\int_{\theta=-c}^{1-c}[\lambda^*(v)+\mu^*(\theta)q(v,\theta)]g^*(v,\theta)\dd v\dd\theta=\int_{v=0}^1\int_{\theta=-c}^{1-c}p(v,\theta)g^*(v,\theta)\dd v\dd\theta~.
\end{equation}
For this special pair of feasible solutions, after combining Equation (\ref{eq:proof-br-sufficient-3}) and Equation (\ref{eq:proof-br-sufficient-5}), then we have
\begin{equation}
    \label{eq:proof-br-sufficient-6}\int_0^1\lambda^*(v)f(v)\dd v=\int_{v=0}^1\int_{\theta=-c}^{1-c}p(v,\theta)g^*(v,\theta)\dd v\dd\theta~.
\end{equation}
For any feasible solution $g$ to the primal problem, after combining Equation (\ref{eq:proof-br-sufficient-3}) and Equation (\ref{eq:proof-br-sufficient-4}), we obtain that
\begin{equation}
    \label{eq:proof-br-sufficient-7}\int_0^1\lambda^*(v)f(v)\dd v\ge\int_{v=0}^1\int_{\theta=-c}^{1-c}p(v,\theta)g(v,\theta)\dd v\dd\theta~.
\end{equation}
By Equation (\ref{eq:proof-br-sufficient-6}) and Inequality (\ref{eq:proof-br-sufficient-7}), we have
\begin{equation*}
    \int_{v=0}^1\int_{\theta=-c}^{1-c}p(v,\theta)g^*(v,\theta)\dd v\dd\theta\ge \int_{v=0}^1\int_{\theta=-c}^{1-c}p(v,\theta)g(v,\theta)\dd v\dd\theta~,
\end{equation*}
which means the special solution $g^*$ is optimal to the primal problem \eqref{eq:best-response}.
In the same manner, we can also prove that $(\lambda^*,\mu^*)$ is the optimal solution to the dual problem \eqref{eq:dual-obj}.
Besides, Equation \eqref{eq:proof-br-sufficient-6} implies that the strong duality holds between this primal problem and the dual problem.

Then we prove the necessity of \Cref{thm:primal-dual-optimal}.
We assume that $G^*$ with density $g^*$ is an optimal solution to \eqref{eq:best-response} and $(\lambda^*,\mu^*)$ is the optimal solution to \ref{eq:dual-obj}.
Since the strong duality holds, we know that
\begin{equation}
\label{eq:proof-br-necessary-1}
    \int_0^1\lambda^*(v)f(v)\dd v=\int_{v=0}^1\int_{\theta=-c}^{1-c}p(v,\theta)\ g^*(v,\theta)\dd v\dd \theta~.
\end{equation}
By the value-marginal constraint, we have that 
\begin{equation}
    \label{eq:proof-br-necesity-1}\int_{v=0}^1\int_{\theta=-c}^{1-c}\lambda^*(v)g^*(v,\theta)\dd v\dd \theta=\int_0^1\lambda^*(v)f(v)\dd v~.
\end{equation}
By the row-wise index constraint, we have
\begin{equation}
    \label{eq:proof-br-necesity-2}\int_{v=0}^1\int_{\theta=-c}^{1-c}\mu^*(\theta)q(v,\theta)\ g^*(v,\theta)\dd v\dd \theta=0~.
\end{equation}
Combining Equation (\ref{eq:proof-br-necesity-1}) and Equation (\ref{eq:proof-br-necesity-2}), we have achieve that
\begin{equation}
    \label{eq:proof-br-necesity-3}\int_{v=0}^1\int_{\theta=-c}^{1-c}[\lambda^*(v)+\mu^*(\theta)q(v,\theta)]\ g^*(v,\theta)\dd v\dd \theta=\int_0^1\lambda^*(v)f(v)\dd v~.
\end{equation}
Combining Equation (\ref{eq:proof-br-necessary-1}) and Equation (\ref{eq:proof-br-sufficient-3}), we achieve that
\begin{equation*}
    \int_{v=0}^1\int_{\theta=-c}^{1-c}[\lambda^*(v)+\mu^*(\theta)q(v,\theta)-p(v,\theta)]\ g^*(v,\theta)\dd v\dd \theta=0~,
\end{equation*}
which means that the optimal solution to the primal problem $G^*$ and the optimal solution to the dual problem $(\lambda^*,\mu^*)$ satisfy Condition (\ref{eq:primal-dual-optimal}).
\end{proof}

\subsection{Proof of \Cref{lem:no-threshold-value}}
\label{pf:no-threshold-value}

\begin{proof}[Proof of \Cref{lem:no-threshold-value}]
For any feasible 2-D distribution $G$, there being no such threshold value $\underline{v}$ implies that, there exists a tuple $(0,\theta)\in\supp(G)$ such that $\theta>0$.
In any best response $G^*$, without loss of generality, we can assume that such $(0,\theta)\in\supp(G)$ with $\theta>0$ does not exist since pooling value $0$ up to a positive index brings no utility increase.
\end{proof}

\subsection{\Cref{alg:mu} for Construction of Dual Variable $\mu(\theta)$}
\label{pf:alg-2}
Here is the algorithm to construct the dual variable $\mu(\theta)$ for any $\theta\in[-c,1-c]$.

\begin{algorithm}[h]
\caption{MU ($\theta$): Constructions of the dual variable $\mu$}
\label{alg:mu}
\SetKwInOut{Require}{Require}
\SetKwInOut{Ensure}{Ensure}

\Require{Strategy $G$; Interim utility $u$; Cost $c$; Threshold value $\underline{v}$; All $\lambda(v)$ for $v \in [0,1]$.}
\Ensure{Whether the construction is valid (True/False).}

\BlankLine

\For{$\theta \in [-c, \underline{v}-c]$}{
    $\mu(\theta) \gets 0$ \;
}

\For{$\theta \in \mathrm{supp}(G_\theta) \cup (\underline{v}-c, 1-c]$}{
    \If{$\mathrm{supp}(G_{\cdot|\theta}) = \{\theta+c\}$}{
        \If{$\displaystyle \max_{v \in (\theta+c, 1]} \frac{u(\theta)-\lambda(v)}{v-\theta-c} \le \min \left\{ \min_{v \in [0, \theta)} \frac{\lambda(v)-u(v)}{c}, \min_{v \in [\theta, \theta+c)} \frac{u(\theta)-\lambda(v)}{v-\theta-c} \right\}$}{
            \Return \textbf{False} \;
        }
        $\mu(\theta) \gets \min \left\{ \min_{v \in [0, \theta)} \frac{\lambda(v)-u(v)}{c}, \min_{v \in [\theta, \theta+c)} \frac{u(\theta)-\lambda(v)}{v-\theta-c} \right\}$ \;
    }
    \Else{
        $\mu_{\text{ref}} \gets \text{None}$ \;
        \For{$v \in \mathrm{supp}(G_{\cdot|\theta}) \setminus \{\theta+c\}$}{
            $\mu_{\text{cur}} \gets \frac{p(v,\theta) - \lambda(v)}{q(v,\theta)}$ \;
            \If{$\mu_{\text{ref}} = \text{None}$}{
                $\mu_{\text{ref}} \gets \mu_{\text{cur}}$ \;
            }
            \ElseIf{$\mu_{\text{cur}} \neq \mu_{\text{ref}}$}{
                \Return \textbf{False} \;
            }
        }
        $\mu(\theta) \gets \mu_{\text{ref}}$ \;
    }
}

\For{$\theta \notin \mathrm{supp}(G_\theta)$}{
    \If{$\displaystyle \max_{v \in (\theta+c, 1]} \frac{u(\theta)-\lambda(v)}{v-\theta-c} \le \min \left\{ \min_{v \in [0, \theta)} \frac{\lambda(v)-u(v)}{c}, \min_{v \in [\theta, \theta+c)} \frac{u(\theta)-\lambda(v)}{v-\theta-c} \right\}$}{
        \Return \textbf{False} \;
    }
    $\mu(\theta) \gets \min \left\{ \min_{v \in [0, \theta)} \frac{\lambda(v)-u(v)}{c}, \min_{v \in [\theta, \theta+c)} \frac{u(\theta)-\lambda(v)}{v-\theta-c} \right\}$ \;
}

\Return \textbf{True} \;
\end{algorithm}

\subsection{Proof of \Cref{thm:constructed-lambda-mu-optimal}}
\label{pf:verification}

Before proving \Cref{thm:constructed-lambda-mu-optimal}, we have to introduce the following lemmas.

\begin{lemma}[Case 1 in \Cref{alg:lambda}]
\label{lem:primal-dula-1}
Given $G^*$ and $(\lambda^*,\mu^*)$ are both optimal to the primal and dual problems,
then we without loss of generality assume that $\lambda^*(v)=0$ for each $v\in[0,\underline{v}]$ and $\mu^*(\theta)=0$ for each $\theta\in[-c,\underline{v}-c]$.
\end{lemma}

\begin{proof}[Proof of \Cref{lem:primal-dula-1}]
Since strategy $G^*$ is optimal to the primal problem, the threshold value $\underline{v}$ is above zero.
Assuming no pooling is made to values $[0,\underline{v})$ is without loss of generality, thus we have $\lambda^*(v)=0$ for any $v\in[0,\underline{v}]$.
Besides, we can assume that $\mu^*(\theta)=0$ for each $\theta\in[-c,\underline{v}-c]$ without loss of generality.
Since assuming $\mu^*(\theta) = 0$ for each $\theta \in [-c, \underline{v}-c]$ does not affect the objective of the dual while keeping all constraints held.
\end{proof}

\begin{lemma}[Case 2 in \Cref{alg:lambda}]
\label{lem:primal-dual-2}
Given $G^*$ and $(\lambda^*,\mu^*)$ are both optimal to the primal and dual problems,
if there exists $v\in[\underline{v},1]$ such that $(v,v-c)\in\supp(G)$, then $\lambda^*(v)=u(v-c)$.
\end{lemma}
\begin{proof}[Proof of \Cref{lem:primal-dual-2}]
By \Cref{thm:primal-dual-optimal}, it directly holds that $\lambda^*(v)=-\mu^* (v-c)q(v,v-c)+p(v,v-c)=p(v,v-c)=u(v-c)$ since $q(v,v-c)=0$.
\end{proof}

\begin{lemma}[Case 3.1 in \Cref{alg:lambda}]
\label{lem:limit point with value=index+c}
Given $G^*$ and $(\lambda^*,\mu^*)$ are both optimal to the primal and dual problems, 
for any $(v,\theta)$ with $v=\theta+c$,
if there is a sequence $\{(v^m,\theta^m)\}_{m\in\Zplus}$ converging to $(v,\theta)$ such that : (i) $(v^m,\theta^m)\in\supp(G)$ for any $m\in\Zplus$; 
(ii) $\{v^m\}_{m\in\Zplus}$ is strictly increasing;
(iii) $\{\theta^m\}_{m\in\Zplus}$ is monotone;
and (iv) $v^m\in(\theta^m,v)$ for any $m\in\Zplus$, then $\lambda^*(v)=u(v-c)$.
\end{lemma}
\begin{proof}[Proof of \Cref{lem:limit point with value=index+c}]
For any $m\in\Zplus$, since $(v^m,\theta^m)\in\supp(G^*)$, \Cref{thm:primal-dual-optimal} implies that
\begin{equation*}
    \lambda^*(v^m)=-\mu^*(\theta^m)q(v^m,\theta^m)+p(v^m,\theta^m)=-\mu^*(\theta^m)(v^m-\theta^m-c)+u(\theta^m)~.
\end{equation*}
Besides, since the function $u$ is Lipschitz continuous everywhere, we have already proved the boundedness of the function $u$ in \Cref{thm:optimal-solution-exists-in-dual-new}.
Thus, we know that
\begin{equation*}
    \lambda^*(v)=\lim \limits_{m\rightarrow \infty}\lambda^*(v^m)=\lim\limits_{m\rightarrow \infty}p(v^m,\theta^m)=u(v-c)~.
\end{equation*}
This completes the proof.
\end{proof}

\begin{lemma}[Case 3.2 in \Cref{alg:lambda} for the limit point $(v,\theta)$ with $v<\theta$]
\label{lem:limit point with value<index}
Given $G^*$ and $(\lambda^*,\mu^*)$ are both optimal to the primal and dual problems, 
for any $(v,\theta)$ with $v\in(\underline{v},\theta)$,
if there is a sequence $\{(v^m,\theta^m)\}_{m\in\Zplus}$ converging to $(v,\theta)$ such that : 
(i) $(v^m,\theta^m)\in\supp(G)$ for any $m\in\Zplus$; 
(ii) $\{v^m\}_{m\in\Zplus}$ is strictly increasing;
(iii) $\{\theta^m\}_{m\in\Zplus}$ is monotone; and
(iv) $\theta^m>v$ for any $m\in\Zplus$, 
then $-\mu^*(\theta^{m_1})=u'(\theta^{m_1})=u'(\theta^{m_2})=-\mu^*(\theta^{m_2})$ for any $m_1,m_2\in\Zplus$.
Besides, $\lambda^*(v)=u'(\theta^1)q(v,\theta)+p(v,\theta)$.
\end{lemma}
\begin{proof}[Proof of \Cref{lem:limit point with value<index}]
Without loss of generality, we can assume that $\{\theta^m\}_{m\in\Zplus}$ is also strictly increasing.
For any $m\in\Zplus$, 
by \Cref{thm:primal-dual-optimal} and the fact that function $\lambda^*$ is continuous, we know that for any $m_1,m_2\in\Zplus$ with $m_1> m_2$, 
\begin{equation*}
    \lambda^*(v^{m_1})=(-c)\max_{\theta'\in[v,1-c]}-\mu^*(\theta')+u(v^{m_1})=-\mu^*(\theta^{m_1})(-c)+u(v^{m_1})~,
\end{equation*}
and
\begin{equation*}
    \lambda^*(v^{m_2})=(-c)\max_{\theta'\in[v,1-c]}-\mu^*(\theta')+u(v^{m_2})=-\mu^*(\theta^{m_2})(-c)+u(v^{m_2})~.
\end{equation*}
Combining these two equations, we achieve that
\begin{equation*}
    \max_{\theta'\in[v,1-c]}-\mu^*(\theta')=-\mu^*(\theta^{m_1})=-\mu^*(\theta^{m_2})~.
\end{equation*}
Next, we aim to prove that $\mu^*(\theta^{m_1})=-\mu^*(\theta^{m_2})=u'(\theta^{m_1})=u'(\theta^{m_2})$. 
To form feasible indices $\theta^{m_1}$ and $\theta^{m_2}$, there must exist $(\hat{v}^{m_1},\theta^{m_1})\in \supp(G^*)$ with $\hat{v}^{m_1}>\theta^{m_1}+c$, and $(\hat{v}^{m_2},\theta^{m_2})\in \supp(G^*)$ with $\hat{v}^{m_2}>\theta^{m_2}+c$.
By \Cref{thm:primal-dual-optimal}, it holds that
\begin{equation*}
    \lambda^*(\hat{v}^{m_1})=-\mu^*(\theta^{m_1})(\hat{v}^{m_1}-\theta^{m_1}-c)+u(\theta^{m_1})\ge -\mu^*(\theta^{m_2})(\hat{v}^{m_1}-\theta^{m_2}-c)+u(\theta^{m_2})~,
\end{equation*}
and 
\begin{equation*}
    \lambda^*(\hat{v}^{m_2})=-\mu^*(\theta^{m_2})(\hat{v}^{m_2}-\theta^{m_2}-c)+u(\theta^{m_2})\ge -\mu^*(\theta^{m_1})(\hat{v}^{m_2}-\theta^{m_1}-c)+u(\theta^{m_1})~.
\end{equation*}
Since 
(i) the function $-\mu^*(\theta^{m_1})q(\ \cdot\ ,\theta^{m_1})+u(\theta^{m_1})$ is linear over $[\theta_1,1-c]$ and crosses $(\theta^{m_1}+c,u(\theta^{m_1}))$;
(ii) the function $-\mu^*(\theta^{m_2})q(\ \cdot\ ,\theta^{m_2})+u(\theta^{m_2})$ is linear over $[\theta_2,1-c]$ and crosses $(\theta^{m_2}+c,u(\theta^{m_2}))$; and
(iii) $-\mu^*(\theta^{m_1})=-\mu^*(\theta^{m_2})$, 
it must hold that $\mu^*(\theta^{m_1})=-\mu^*(\theta^{m_2})=u'(\theta^{m_1})=u'(\theta^{m_2})$ and these two functions overlap over $[\theta_2+c,1-c]$.

Recall that $\lambda^*,p,q$ are all continuous; we complete our proof by showing that
\begin{equation*}
    \lambda^*(v)=\lim\limits_{m\rightarrow\infty}\lambda^*(v^m)=\lim\limits_{m\rightarrow\infty}-\mu^*(\theta^m)q(v^m,\theta^m)+p(v^m,\theta^m)=u'(\theta^1)q(v,\theta)+p(v,\theta)~.
\end{equation*}
This completes the proof.
\end{proof}

\begin{lemma}[Case 3.2 in \Cref{alg:lambda} for the limit point $(v,\theta)$ with $v\in(\theta,\theta+c)$]
\label{lem:limit point with value between theta and theta+c}
Given $G^*$ and $(\lambda^*,\mu^*)$ are both optimal to the primal and dual problems, 
for any $(v,\theta)$ with $v\in(\max\{\underline{v},\theta\},\theta+c)$,
if there is a sequence $\{(v^m,\theta^m)\}_{m\in\Zplus}$ converging to $(v,\theta)$ such that : (i) $(v^m,\theta^m)\in\supp(G)$ for any $m\in\Zplus$; 
(ii) $\{v^m\}_{m\in\Zplus}$ is strictly increasing;
(iii) $\{\theta^m\}_{m\in\Zplus}$ is monotone;
(iv) let $\theta'=\frac{v+\theta}{2}$, $\theta^m\in(v-c,\theta')$ and $v^m\in(\theta',v)$ for any $m\in\Zplus$, then $-\mu^*(\theta^{m_1})=u'(\theta^{m_1})=u'(\theta^{m_2})=-\mu^*(\theta^{m_2})$ for any $m_1,m_2\in\Zplus$.
Besides, $\lambda^*(v)=u'(\theta^1)q(v,\theta)+p(v,\theta)$.
\end{lemma}
\begin{proof}[Proof of \Cref{lem:limit point with value between theta and theta+c}]
Without loss of generality, we can assume that $\{\theta^m\}_{m\in\Zplus}$ is also strictly increasing.
We can consider two tuples $(v^3,\theta^3),\ (v^4,\theta^4)$ with $\theta'\le v_3<v_4<v$ and $\theta_4>\theta_3$.
There exist tuples $(\bar{v}^3,\theta^3),(\bar{v}^4,\theta^4)\in\supp(G)$ such that $\bar{v}^3>\theta^3+c$ and $\bar{v}^4>\theta^4+c$.
Based on \Cref{thm:primal-dual-optimal}, we achieve the following inequalities:
\begin{align*}
    -\mu^*(\theta^3)q(v^3,\theta^3)+p(v^3,\theta^3)&\ge -\mu^*(\theta^4)q(v^3,\theta^4)+p(v^3,\theta^4)~,\\
    -\mu^*(\theta^3)q(\bar{v}^3,\theta^3)+p(\bar{v}^3,\theta^3)&\ge -\mu^*(\theta^4)q(\bar{v}^3,\theta^4)+p(\bar{v}^3,\theta^4)~,\\
    -\mu^*(\theta^4)q(v^4,\theta^4)+p(v^4,\theta^4)&\ge -\mu^*(\theta^3)q(v^4,\theta^3)+p(v^4,\theta^3)~,\\
    -\mu^*(\theta^4)q(\bar{v}^4,\theta^4)+p(\bar{v}^4,\theta^4)&\ge -\mu^*(\theta^3)q(\bar{v}^4,\theta^3)+p(\bar{v}^4,\theta^3)~.
\end{align*}
Since (i) the function $-\mu^*(\theta^3)q(\ \cdot\ ,\theta^3)+p(\ \cdot\ ,\theta^3)$ is linear over $(\theta',v)$;
(ii) the function $-\mu^*(\theta^3)q(\ \cdot\ ,\theta^3)+p(\ \cdot\ ,\theta^3)$ crosses $(\theta^3+c,u(\theta^3))$;
(iii) the function $-\mu^*(\theta^4)q(\ \cdot\ ,\theta^4)+p(\ \cdot\ ,\theta^4)$ is linear over $(\theta',v)$; and
(iv) the function $-\mu^*(\theta^4)q(\ \cdot\ ,\theta^4)+p(\ \cdot\ ,\theta^4)$ crosses $(\theta^4+c,u(\theta^4))$, the above four inequalities imply that $\mu^*(\theta^3)=\mu^*(\theta^4)=u'(\theta^3)=u'(\theta^4
)$ and these two functions overlap over $(\theta',1]$.
The above analysis holds for any $m_1,m_2\in\Zplus$ with $m_1<m_2$.
Thus we know that $\mu^*(\theta^{m_1})=\mu^*(\theta^{m_2})=u'(\theta^{m_1})=u'(\theta^{m_2})$ for any $m_1,m_2\in\Zplus$.
For any $m\in\Zplus$, it holds
\begin{equation*}
    \lambda^*(v^m)=p(v^m,\theta^m)-\mu^*(\theta^m)q(v^m,\theta^m)~.
\end{equation*}
Passing $m$ to infinity, we achieve that
\begin{equation*}
    \lambda^*(v)=p(v,\theta)-\mu^*(\theta^1)q(v,\theta)=u'(\theta^1)q(v,\theta)+p(v,\theta)~.
\end{equation*}
This completes the proof.
\end{proof}

\begin{lemma}[Case 3.2 in \Cref{alg:lambda} for the limit point $(v,\theta)$ with $v>\theta+c$]
\label{lem:limit point with value > theta+c}
Given $G^*$ and $(\lambda^*,\mu^*)$ are both optimal to the primal and dual problems, 
for any $(v,\theta)$ with $v>\theta+c$,
if there is a sequence $\{(v^m,\theta^m)\}_{m\in\Zplus}$ converging to $(v,\theta)$ such that : (i) $(v^m,\theta^m)\in\supp(G)$ for any $m\in\Zplus$; 
(ii) $\{v^m\}_{m\in\Zplus}$ is strictly increasing;
(iii) $\{\theta^m\}_{m\in\Zplus}$ is monotone;
(iv) let $\theta'=\frac{v-c+\theta}{2}\in(\theta,v-c)$, $\theta^m\in[0,\theta')$ and $v^m\in(\theta'+c,v)$ for any $m\in\Zplus$, then $-\mu^*(\theta^{m_1})=u'(\theta^{m_1})=u'(\theta^{m_2})=-\mu^*(\theta^{m_2})$ for any $m_1,m_2\in\Zplus$.
Besides, $\lambda^*(v)=u'(\theta^1)q(v,\theta)+p(v,\theta)$.
\end{lemma}
\begin{proof}[Proof of \Cref{lem:limit point with value > theta+c}]
We proceed with the proof in two steps.

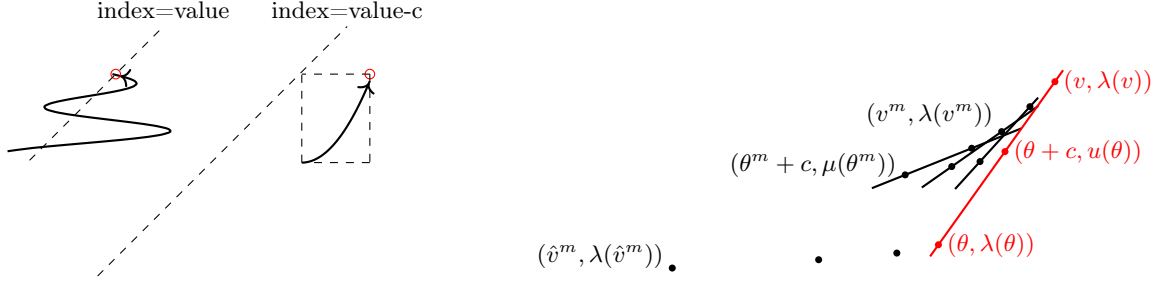
\begin{figure}[htbp]
    \centering
    \begin{minipage}[t]{0.38\textwidth}
        \centering
        \begin{tikzpicture}[scale=3.0,decoration={markings, mark=at position 0.98 with {\arrow{>}}}] 

        \draw[thick, postaction={decorate},domain=-2.5:0.9,smooth,variable=\t] plot ({0.2*exp(-0.3*\t)*sin(3*deg(\t)) + 0.6}, {0.1*\t + 1});
        \draw[red] (0.68,1.09) circle (0.6pt);
        
        \draw[dashed] (0.6,0.2) -- (1.7,1.3) node[above] {\small index=value-c};
        \draw[dashed] (0.3,0.71) -- (0.89,1.3) node[above] {\small index=value};
        
        \draw[red] (1.8,1.09) circle (0.6pt);
        \draw[dashed] (1.8,1.09) -- (1.5,1.09);
        \draw[dashed] (1.8,1.09) -- (1.8,0.7);
        \draw[dashed] (1.8,0.7) -- (1.5,0.7);
        \draw[dashed] (1.5,0.7) -- (1.5,1.09);

        \draw[thick, postaction={decorate}, domain=1.5:1.8, samples=50] plot (\x, {4*(\x - 1.5)^2+0.7});
        
    \end{tikzpicture}
    \end{minipage}
    \hfill
    \begin{minipage}[t]{0.58\textwidth}
        \centering
        \begin{tikzpicture}[scale=2.2] 
        
        \draw[thick, domain=0.9:1.8, samples=50] plot (\x, {0.4*(\x - 1.5)+1});
        \fill (1.5,1) circle (0.6pt);
        \draw (1.1,0.9) node[left] {\small $(\theta^m+c,\mu(\theta^m))$};
        \fill (1.1,0.84) circle (0.6pt);
        \draw (-0.3,0.35) node[left] {\small $(\hat{v}^m,\lambda(\hat{v}^m))$};
        \fill (-0.3,0.28) circle (0.6pt);
        
        \draw[thick, domain=1.2:1.9, samples=50] plot (\x, {0.7*(\x - 1.68)+1.1});
        \fill (1.68,1.1) circle (0.6pt);
        \draw (1.68,1.2) node[left] {\small $(v^m,\lambda(v^m))$};
        \fill (1.38,0.89) circle (0.6pt);
        \fill (0.58,0.33) circle (0.6pt);

        \draw[thick, domain=1.4:1.9, samples=50] plot (\x, {1.1*(\x - 1.85)+1.25});
        \fill (1.85,1.25) circle (0.6pt);
        \fill (1.55,0.92) circle (0.6pt);
        \fill (1.05,0.37) circle (0.6pt);

        \draw[thick, red, domain=1.25:2.05, samples=50] plot (\x, {1.4*(\x - 2)+1.4});
        \fill[red] (2,1.4) circle (0.6pt) node[right] {\small $(v,\lambda(v))$};
        \fill[red] (1.7,0.98) circle (0.6pt) node[right] {\small $(\theta+c,u(\theta))$};
        \fill[red] (1.3,0.42) circle (0.6pt) node[right] {\small $(\theta,\lambda(\theta))$};
        
    \end{tikzpicture}
    \end{minipage}
    \caption{Graph illustration of Lemma \ref{lem:limit point with value > theta+c}}
    \label{}
\end{figure}

\xhdr{Step-1: Prove that $-\mu^*(\theta^{m_1})=-\mu^*(\theta^{m_2})=u'(\theta^{m_1})=u'(\theta^{m_2})$ for any $m_1,m_2\in\Zplus$}
Let $v_l^m\triangleq \inf\supp(G^*_{\cdot|\theta^m})$ for any $m\in\Zplus$.
Since the sequence $\{(v^m,\theta^m)\}_{m\in\Zplus}$ converges to  $(v,\theta)$, the sequence $\{(v_l^m,\theta^m)\}_{m\in\Zplus}$ also converges to some $(v_l,\theta)$ with $v_l<\theta+c$.
We divide the proof into the following three cases.

\begin{enumerate}
    \item[(i)] \textbf{Case 1: $v_l<\theta$.}
    We can prove this result using a similar method as \Cref{lem:limit point with value<index}.

    \item[(ii)] \textbf{Case 2: $v_l\in(\theta,\theta+c)$.}
    We can prove this result using a similar method as \Cref{lem:limit point with value between theta and theta+c}.

    \item[(iii)] \textbf{Case 3: $v_l=\theta$.}
    If the sequence $\{\theta^m\}$ is also decreasing, then we know that there exists a subsequence $\{m(k)\}_{k\in\Zplus}$ such that $\{(v^{m(k)}_l,\theta^{m(k)})\}_{k\in{\Zplus}}$ converges to $(v_l,\theta)$, and is decreasing.
    In this case, we can prove the target result using a similar method as \Cref{lem:limit point with value<index}.
    
    Thus, without loss, we can assume that the sequence $\{\theta^m\}$ is increasing, and there exists a subsequence $\{m(k)\}_{k\in\Zplus}$ such that $\{(v^{m(k)}_l,\theta^{m(k)})\}_{k\in{\Zplus}}$ converges to $(v_l,\theta)$, and is also increasing.
    Abusing the notations a little, we rename the sequence $\{m(k)\}_{k\in\Zplus}$ as $\{m\}_{m\in\Zplus}$.
    Since the sequences $\{v^m\}_{m\in\Zplus}$ and $\{\theta^m\}_{m\in\Zplus}$ are both increasing, by \Cref{thm:primal-dual-optimal}, we know that $\{-\mu^*(\theta^m)\}_{m\in\Zplus}$ is also increasing.
    We have already prove the boundedness of $\mu^*$, thus $\{-\mu^*(\theta^m)\}_{m\in\Zplus}$ also converges to some real number, which is denoted as $-\hat{\mu}$, that is $$
    \lim\limits_{m\rightarrow\infty}-\mu^*(\theta^m)=-\hat{\mu}~.
    $$
    
    Next, we aim to prove that there exists an $M>0$ such that for any $m>M$, $-\mu^*(\theta^m)=-\hat{\mu}$.
    We prove it through contradiction.
    We assume that for any $m<\infty$, it holds that $-\mu^*(\theta^m)<-\hat{\mu}$.
    We consider three sequences:
    (i) $\{(v_l^m,\lambda(v_l^m))\}_{m\in\Zplus}$, 
    (ii) $\{(\theta^m+c,u(\theta^m))\}_{m\in\Zplus}$, and
    (iii) $\{(v^m,\lambda(v^m))\}_{m\in\Zplus}$.
    Let $(x_a,y_a)$ denote the limit point of the sequence $\{(v_l^m,\lambda(v_l^m))\}_{m\in\Zplus}$.
    Let $(x_b,y_b)$ denote the limit point of the sequence $\{(\theta^m+c,u(\theta^m))\}_{m\in\Zplus}$.
    Let $(x_c,y_c)$ denote the limit point of the sequence $\{(v^m,\lambda(v^m))\}_{m\in\Zplus}$.
    Based on the graph illustrations, we can observe that $u'(\theta)\le -\hat{\mu}$.
    Next we aim to show that these three points $(x_a,y_a),(x_b,y_b),(x_c,y_c)$ lie in a common linear function.
    First, we show that $\frac{y_c-y_b}{x_c-x_b}=-\hat{\mu}$:
    \begin{align*}
        \frac{y_c-y_b}{x_c-x_b}&=\lim\limits_{m\rightarrow\infty} \frac{\lambda^*(v^m)-u(\theta^m)}{v^m-(\theta^m+c)}\\
        &=\lim\limits_{m\rightarrow\infty}\frac{-\mu^*(\theta^m)[v^m-(\theta^m+c)]}{v^m-(\theta^m+c)}\\
        &=\lim\limits_{m\rightarrow\infty}-\mu^*(\theta^m)=-\hat{\mu}~.
    \end{align*}
    Next, we prove that $\frac{y_b-y_a}{x_b-x_a}=-\hat{\mu}$.
    For any $m\in\Zplus$, we have that
    \begin{align*}
        \frac{y_b-y_a}{x_b-x_a}&=\lim\limits_{m\rightarrow\infty}\frac{u(\theta^m)-\lambda(v_l^m)}{\theta^m+c-v_l^m}\\
        &=\lim\limits_{m\rightarrow\infty}\frac{u(\theta^m)-(-\mu^*(\theta^m)q(v_l^m,\theta^m)+p(v_l^m,\theta^m))}{\theta^m+c-(v_l^m-\theta^m)-\theta^m}=-\hat{\mu}~.
    \end{align*}
    Combining these two facts, we know that $(x_a,y_a),(x_b,y_b),(x_c,y_c)$ lie in a common linear function.
    Under this fact, we know that for sufficiently small $\epsilon>0$,
    \begin{equation*}
        -\mu^*(\theta-\epsilon)q(v_l,\theta-\epsilon)+p(v_l,\theta-\epsilon)>-\mu^*(\theta)q(v_l,\theta)+p(v_l,\theta)~,
    \end{equation*}
    which forms a contradiction with the fact that the sequence $\{v_l^m,\theta^m\}_{m\in\Zplus}$ converges to $(v_l,\theta)$ with $v_l=\theta$.
    Thus, we have proved that the assumption in the beginning is invalid.
    To sum up, we have proved
    that there exists a sufficiently large $M>0$ such that $-\mu^*(\theta^{m_1})=-\mu^*(\theta^{m_2})$ for any $m_1,m_2>M$.
    Using a similar idea of \Cref{lem:limit point with value<index}, we can also prove that $-\mu^*(\theta^{m_1})=-\mu^*(\theta^{m_2})=u'(\theta^{m_1})=u'(\theta^{m_2})$ for any $m_1,m_2>M$.
\end{enumerate}

\xhdr{Step-2: Prove that $\lambda^*(v)=u'(\theta^1)q(v,\theta)+p(v,\theta)$}
For any $m\in\Zplus$, \Cref{thm:primal-dual-optimal} implies that
\begin{equation*}
    \lambda^*(v^m)=p(v^m,\theta^m)-\mu^*(\theta^m)q(v^m,\theta^m)~.
\end{equation*}
Passing $m$ to infinity, 
since functions $\lambda^*,p,q$ are both continuous,
we achieve that
\begin{equation*}
    \lambda^*(v)=p(v,\theta)-\mu^*(\theta^1)q(v,\theta)=u'(\theta^1)q(v,\theta)+p(v,\theta)~.
\end{equation*}
This completes the proof.
\end{proof}

\begin{lemma}[Case 3.2 in \Cref{alg:lambda} for the limit point $(v,\theta)$ with $v=\theta$]
\label{lem:limit point with value=theta}
Given $G^*$ and $(\lambda^*,\mu^*)$ are both optimal to the primal and dual problems, 
for any $(v,\theta)$ with $v=\theta+c$,
if there is a sequence $\{(v^m,\theta^m)\}_{m\in\Zplus}$ converging to $(v,\theta)$ such that : (i) $(v^m,\theta^m)\in\supp(G)$ for any $m\in\Zplus$; 
(ii) $\{v^m\}_{m\in\Zplus}$ is strictly increasing; and
(iii) $\{\theta^m\}_{m\in\Zplus}$ is monotone, 
then there exists a sufficiently large $M>0$ such that $-\mu^*(\theta^{m})=u'(\theta)$ for any $m>M$.
Besides, $\lambda^*(v)=u'(\theta)q(v,\theta)+p(v,\theta)$.
\end{lemma}

\begin{proof}[Proof of \Cref{lem:limit point with value=theta}]
Based on the properties of the sequence $\{(v^m,\theta^m)\}_{m\in\Zplus}$, we can divide the proof into three cases:

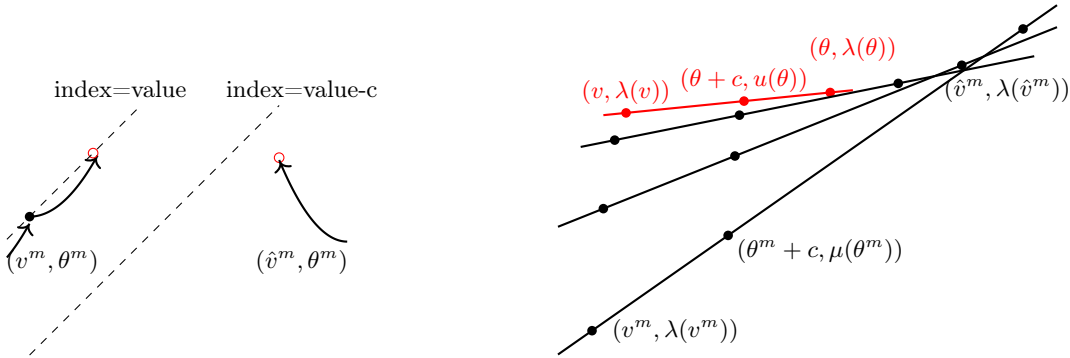
\begin{figure}[htbp]
    \centering
    \begin{minipage}[t]{0.38\textwidth}
        \centering
        \begin{tikzpicture}[scale=3.0,decoration={markings, mark=at position 0.98 with {\arrow{>}}}] 

        \draw[thick, postaction={decorate}, domain=0.4:0.7, samples=50] plot (\x, {3*(\x - 0.4)^2+0.81});
        \draw[thick, postaction={decorate}, domain=0.3:0.4, samples=50] plot (\x, {3*(\x - 0.1)^2+0.51});
        \draw (0.5,0.7) node[below] {\small $(v^m,\theta^m)$};
        \fill (0.4,0.81) circle (0.6pt);
        \draw[red] (0.68,1.09) circle (0.6pt);
        
        \draw[dashed] (0.4,0.2) -- (1.5,1.3);
        \draw (1.6,1.3) node[above] {\small index=value-c};
        \draw[dashed] (0.3,0.71) -- (0.89,1.3);
        \draw (0.8,1.3) node[above] {\small index=value};
        
        \draw[red] (1.5,1.07) circle (0.6pt);

        \draw[thick, postaction={decorate}, domain=1.8:1.5, samples=50] plot (\x, {4*(\x - 1.8)^2+0.7});
        \draw (1.6,0.7) node[below] {\small $(\hat{v}^m,\theta^m)$};
        
    \end{tikzpicture}
    \end{minipage}
    \hfill
    \begin{minipage}[t]{0.58\textwidth}
        \centering
        \begin{tikzpicture}[scale=3.0] 

        \draw[thick, red, domain=-0.1:1, samples=50] plot (\x, {0.1*(\x - 1.4)+0.95});
        \fill[red] (0.9,0.9) circle (0.6pt);
        \draw[red] (1,1) node[above] {\small $(\theta,\lambda(\theta))$};
        \fill[red] (0.52,0.862) circle (0.6pt) node[above] {\small $(\theta+c,u(\theta))$};
        \fill[red] (0,0.81) circle (0.6pt) node[above] {\small $(v,\lambda(v))$};
        
        \draw[thick, domain=-0.2:1.8, samples=50] plot (\x, {0.2*(\x - 1.5)+1});
        \fill (1.2,0.94) circle (0.6pt);
        \fill (0.5,0.8) circle (0.6pt);
        \fill (-0.05,0.69) circle (0.6pt);
        
        \draw[thick, domain=-0.3:1.9, samples=50] plot (\x, {0.4*(\x - 1.68)+1.1});
        \fill (1.48,1.02) circle (0.6pt);
        \draw (1.68,1) node[below] {\small $(\hat{v}^m,\lambda(\hat{v}^m))$};
        \fill (0.48,0.62) circle (0.6pt);
        \fill (-0.1,0.388) circle (0.6pt);

        \draw[thick, domain=-0.3:1.9, samples=50] plot (\x, {0.7*(\x - 1.85)+1.25});
        \fill (1.75,1.18) circle (0.6pt);
        \fill (0.45,0.27) circle (0.6pt);
        \draw (0.45,0.2) node[right] {\small $(\theta^m+c,\mu(\theta^m))$};
        \fill (-0.15,-0.15) circle (0.6pt);
        \draw (-0.1,-0.15) node[right] {\small $(v^m,\lambda(v^m))$};
        
    \end{tikzpicture}
    \end{minipage}
    \caption{Graph illustration of Lemma \ref{lem:limit point with value=theta}}
    \label{}
\end{figure}

\xhdr{Case 1: The sequence $\{\theta^m\}_{m\in\Zplus}$ is decreasing}
This lemma can be proved in a method similar to that of \Cref{lem:limit point with value<index}.

\xhdr{Case 2:
The sequence $\{\theta^m\}_{m\in\Zplus}$ is increasing, and there exists $M>0$ such that $v^m\ge \theta^m$ for any $m>M$}
For any $m\in\Zplus$, there exists $(\hat{v}^m,\theta^m)\in\supp(G^*)$ with $\hat{v}^m>\theta^m+c$.
By \Cref{thm:primal-dual-optimal}, we have that
\begin{equation*}
    p(\hat{v}^m,\theta^m)-\mu^*(\theta^m)q(\hat{v}^m,\theta^m)\ge p(\hat{v}^m,\theta^{m+1})-\mu^*(\theta^{m+1})q(\hat{v}^m,\theta^{m+1})~.
\end{equation*}
Since the functions $p(\cdot,\theta^m)-\mu^*(\theta^m)q(\cdot,\theta^m)$ and $p(\cdot,\theta^{m+1})-\mu^*(\theta^{m+1})q(\cdot,\theta^{m+1})$ are both linear over $[\hat{v}^m,1]$, it follows that $-\mu^*(\theta^m)\ge-\mu^*(\theta^{m+1})$ and $\hat{v}^m\ge \hat{v}^{m+1}$.
Since the function $\mu^*$ is bounded (\Cref{thm:optimal-solution-exists-in-dual-new}), we know that $\lim\limits_{m\rightarrow \infty}-\mu^*(\theta^m)=-\underline{\mu}$.

We consider three sequences:
(i) $\{(v^m,\lambda(v^m))\}_{m\in\Zplus}$, 
(ii) $\{(\theta^m+c,u(\theta^m))\}_{m\in\Zplus}$, and
(iii) $\{(\hat{v}^m,\lambda(\hat{v}^m))\}_{m\in\Zplus}$.
Notice that for any $m\in\Zplus$, three points $(v^m,\lambda(v^m)),(\theta^m+c,u(\theta^m)),(\hat{v}^m,\lambda(\hat{v}^m))$ lie in a common linear function $-\mu^*(\theta^m)q(\cdot,\theta^m)+p(\cdot,\theta^m)$ over $[\theta^m,1-c]$.
Let $(x_a,y_a)$ denote the limit point of the sequence $\{(v^m,\lambda(v^m))\}_{m\in\Zplus}$.
Let $(x_b,y_b)$ denote the limit point of the sequence $\{(\theta^m+c,u(\theta^m))\}_{m\in\Zplus}$.
Let $(x_c,y_c)$ denote the limit point of the sequence $\{(\hat{v}^m,\lambda(\hat{v}^m))\}_{m\in\Zplus}$.
We can observe that the function $-\mu^*(\theta)q(\cdot,\theta)+p(\cdot,\theta)$ crosses the point $(x_b,y_b)$.
It follows that $u'(\theta)\ge -\underline{\mu}$.

Next we aim to prove $u'(\theta)= -\underline{\mu}$ through contradiction.
We assume that $u'(\theta)> -\underline{\mu}$.
We divide the proof into two cases.

\sxhdr{Subcase 2.1: Points $(x_b,y_b)$ and $(x_c,y_c)$ do not overlap}
We show that $\frac{y_c-y_b}{x_c-x_b}=-\underline{\mu}$.
\begin{align*}
    \frac{y_c-y_b}{x_c-x_b}&=\lim\limits_{m\rightarrow\infty} \frac{\lambda^*(\hat{v}^m)-u(\theta^m)}{\hat{v}^m-(\theta^m+c)}\\
    &=\lim\limits_{m\rightarrow\infty}\frac{-\mu^*(\theta^m)[\hat{v}^m-(\theta^m+c)]}{\hat{v}^m-(\theta^m+c)}\\
    &=\lim\limits_{m\rightarrow\infty}-\mu^*(\theta^m)=-\underline{\mu}~.
\end{align*}
In a similar manner, we can also prove that $\frac{y_b-y_a}{x_b-x_a}=-\underline{\mu}$.
Thus, we have proved that points $(x_a,y_a),(x_b,y_b),(x_c,y_c)$ lie in a common linear function.
Next we consider an index $\theta+\epsilon$ for sufficiently small $\epsilon>0$.
Since the function $\lambda^*$ is continuous, we know that
\begin{equation*}
    -\mu^*(\theta)q(x_c,\theta)+p(x_c,\theta)\ge -\mu^*(\theta+\epsilon)q(x_c,\theta+\epsilon)+p(x_c,\theta+\epsilon)~,
\end{equation*}
which further implies that $-\mu^*(\theta+\epsilon)<-\underline{\mu}$.
Finally, we know that
\begin{align*}
    -\mu^*(\theta+\epsilon)q(v,\theta+\epsilon)+p(v,\theta+\epsilon)=u(v)+\mu^*(\theta+\epsilon)\cdot c>u(v)+\underline{\mu}\cdot c=\lambda(v)~,
\end{align*}
which forms a contradiction to \Cref{thm:primal-dual-optimal}.
Thus, we have proved that $u'(\theta)=-\underline{\mu}$.

\sxhdr{Subcase 2.2: Points $(x_b,y_b)$ and $(x_c,y_c)$ overlap}
For any $m\in\Zplus$, it holds $\lambda(\hat{v}^m)\ge u(\hat{v}^m-c)$.
The fact that points $(x_b,y_b)$ and $(x_c,y_c)$ overlap implies that $\lambda^*(\theta+c)=u(\theta)$.
Further,
\begin{equation*}
    \frac{\lambda(\hat{v}^m)-\lambda(\theta+c)}{\hat{v}^m-\theta-c}\ge \frac{u(\hat{v}^m-c)-u(\theta)}{\hat{v}^m-\theta-c}~.
\end{equation*}
Based on the graph illustrations of these functions, we know $-\mu^*(\theta^m)\ge \frac{\lambda(\hat{v}^m)-\lambda(\theta+c)}{\hat{v}^m-\theta-c}$.
Thus, we have that
\begin{equation*}
    -\mu^*(\theta^m)\ge \frac{u(\hat{v}^m-c)-u(\theta)}{\hat{v}^m-\theta-c}~.
\end{equation*}
Passing $m$ to infinity, we achieve that
\begin{equation*}
    -\underline{\mu}=\lim\limits_{m\rightarrow\infty}-\mu^*(\theta^m)\ge u'(\theta)~,
\end{equation*}
which forms a contradiction with the former assumption.

By combining these two cases, we have proved that $-\underline{\mu}=u'(\theta)$.
We know that $\frac{y_b-y_a}{x_b-x_a}=-\underline{\mu}$.
Thus, we have that
\begin{equation*}
    u'(\theta)=-\underline{\mu}=\frac{y_b-y_a}{x_b-x_a}=\frac{u(\theta)-\lambda^*(v)}{\theta+c-v}~,
\end{equation*}
which implies that $\lambda^*(v)=u(\theta)+u'(\theta)(v-\theta-c)$.

\sxhdr{Subcase 2.3: Otherwise}
For any $\epsilon>0$, we define two sets
\begin{align*}
    A\triangleq\{(v',\theta')|v'\ge \theta',v'\in[v-\epsilon,v]\}~,\quad B\triangleq\{(v',\theta')|\theta'\ge v,v'\in[v-\epsilon,v]\}~.
\end{align*}
In this case, there exists an $\epsilon>0$ such that regions $A$ and $B$ are both empty.
In other words, for any value $v'\in(v-\epsilon,v)$, there exists $\theta(v')\in(v',v)$ such that $(v',g(v'))\in\supp(G^*)$.
Next, we define function $h(v')=-\mu^*(\theta(v'))$ for any $v'\in(v-\epsilon,v)$.
For any $v_1,v_2\in(v-\epsilon,v)$ with $v_1<v_2$, it follows that
\begin{equation*}
    h(v_1)=\min_{\theta'\in[v_1,1-c]}-\mu^*(\theta'),\ 
    h(v_2)=\min_{\theta'\in[v_2,1-c]}-\mu^*(\theta')~.
\end{equation*}
Thus we know that $h(v_2)\ge h(v_1)$, which implies that the function $h$ is weakly increasing over $(v-\epsilon,v)$.  
Next we are going to prove that the function $h$ is also continuous over $(v-\epsilon,v)$.

First, we show the function $h$ is right-continuous over $(v-\epsilon,v)$.
For any $w\in(v',g(v'))$, it holds that $h(w)=h(v')$.
Since the prior has positive density everywhere, we know that the function $h$ is right-continuous over $(v-\epsilon,v)$.

Then, we show the function $h$ is left-continuous over $(v-\epsilon,v)$.
We prove it through contradiction.
We assume that the function $h$ is not left-continuous everywhere over $(v-\epsilon,v)$.
For any tuple $(v',g(v'))\in\supp(G^*)$,we define sets
\begin{align*}
    R\triangleq\{(v'',\theta'')|\theta''\in(v',v),v''\in[0,v')\}~,\quad S\triangleq\{(v'',\theta'')|v''\in(v-\epsilon,v'),\theta''\in(v',v)\}~.
\end{align*}
Under this case, the region $R$ must be empty.
Besides, for any value $v''\in(v-\epsilon,v')$, if $(v'',\theta'')\in\supp(G^*)$, then it must hold that $(v'',\theta'')\in S$.
Since $h$ is increasing, we know that
$\lim\limits_{w\rightarrow v'^-}h(w)=\hat{h}$.
Since $\lim_{w\rightarrow v'^-}\lambda^*(w)=\lambda^*(v')$, we know that $u(v')-c\hat{h}=u(v')-ch(v')$.
Finally, we achieve that $\hat{h}=h(v')$, which forms a contradiction with the assumption.
Till now, we have proved that the function $h$ is continuous over $(v-\epsilon,v)$.

Next we aim to prove that $h$ is flat over $(v-\epsilon,v)$.
We assume that the image of the function $h$ includes some open interval $(y_1,y_2)$, thus the image is uncountable.
Then, we define a set
\begin{equation*}
    T(y)\triangleq\{x|h(x)=y\}~.
\end{equation*}
It is obvious that $T(y)$ is non-empty, and we let $z\in T(y)$.
We know that  $T(y)$ includes $(z,g(z))$.
Therefore, we know that for any given $y$, its corresponding set $T(y)$ will consist of mutually disjoint intervals $(z,g(z))$. If we arbitrarily select a rational number from each such interval, then we would have obtained uncountably many rational numbers within the interval $(v-\epsilon,v)$. However, we know that the number of rational numbers within any interval is countable. This leads to a contradiction.

To sum up, we have proved that, under this case, there exists an $\delta>0$, such that $-\mu^*(\theta')=u'(\theta)$ for any $\theta'\in(\theta-\delta,\theta]$.
Thus, we directly have that
\begin{equation*}
    \lambda^*(v)=u'(\theta)q(v,\theta)+p(v,\theta)~.
\end{equation*}
This completes the proof.
\end{proof}

\begin{lemma}[Case 3.3 in \Cref{alg:lambda}]
\label{lem:primal-dual-3}
Given $G^*$ and $(\lambda^*,\mu^*)$ are both optimal to the primal and dual problems,
if there exists $(v,\theta)$ such that $v\in[\underline{v},1]$, $v\neq \theta+c$, and $v>v'=\inf\supp(G_{\cdot|\theta})$, then we have 
$$\lambda^*(v)=\frac{p(v,\theta)q(v',\theta)-p(v',\theta)q(v,\theta)+\lambda^*(v')q(v,\theta)}{q(v',\theta)}~.$$
\end{lemma}
\begin{proof}[Proof of \Cref{lem:primal-dual-3}]
Since $(v,\theta)\in\supp(G)$, by \Cref{thm:primal-dual-optimal}, we know that 
$
    \lambda^*(v)=-\mu^*(\theta)q(v,\theta)+p(v,\theta)
$.
Since the functions $\lambda^*$, $p$, and $q$ are all continuous (\Cref{lem:lambda increasing}), it holds that
$
    \lambda^*(v')=-\mu^*(\theta)q(v',\theta)+p(v',\theta)
$.
By combining these two equations, we have
$$
\lambda^*(v)=\frac{p(v,\theta)q(v',\theta)-p(v',\theta)q(v,\theta)+\lambda^*(v')q(v,\theta)}{q(v',\theta)}~.
$$
This completes the proof.
\end{proof}

\begin{lemma}[\Cref{alg:mu}]
\label{lem:primal-dual-9}
Given $G^*$ and $(\lambda^*,\mu^*)$ are both optimal to the primal and dual problems,
for any $\theta\in\supp(G^*_\theta)$ and any two distinct $v_1,v_2\neq \theta+c$ such that $(v_1,\theta),(v_2,\theta)\in\supp(G^*)$, we have that \begin{equation*}
    \mu^*(\theta)=\frac{p(v_1,\theta)-\lambda^*(v_1)}{q(v_1,\theta)}=\frac{p(v_2,\theta)-\lambda^*(v_2)}{q(v_2,\theta)}~.
\end{equation*}
\end{lemma}

\begin{proof}[Proof of \Cref{lem:primal-dual-9}]
It holds directly by \Cref{thm:primal-dual-optimal}.
\end{proof}

\begin{lemma}[\Cref{alg:mu}]
\label{lem:primal-dual-10}
Given $G^*$ and $(\lambda^*,\mu^*)$ are both optimal to the primal and dual problems,
for any $\theta\in[-c,1-c]$, it holds that
\begin{equation*}
    \max_{v\in(\theta+c,1]}\frac{u(\theta)-\lambda^*(v)}{v-\theta-c}\le \mu^*(\theta)\le\min\left\{\min_{v\in[0,\theta)}\frac{\lambda^*(v)-u(v)}{c},\min_{v\in[\theta,\theta+c)}\frac{u(\theta)-\lambda^*(v)}{v-\theta-c}\right\}~.
\end{equation*}
\end{lemma}

\begin{proof}[Proof of \Cref{lem:primal-dual-10}]
By Constraint (\ref{eq:dual-con1}) and \Cref{thm:primal-dual-optimal}, we directly have that
\begin{equation*}
    \max_{v\in(\theta+c,1]}\frac{u(\theta)-\lambda^*(v)}{v-\theta-c}\le \mu^*(\theta)\le\min\left\{\min_{v\in[0,\theta)}\frac{\lambda^*(v)-u(v)}{c},\min_{v\in[\theta,\theta+c)}\frac{u(\theta)-\lambda^*(v)}{v-\theta-c}\right\}~.
\end{equation*}
This completes the proof.
\end{proof}

With the above lemmas, finally, we can provide the whole proof of \Cref{thm:constructed-lambda-mu-optimal}.

\begin{proof}[Proof of \Cref{thm:constructed-lambda-mu-optimal}]
This theorem directly holds by \Cref{lem:no-threshold-value,lem:primal-dula-1,lem:primal-dual-2,lem:limit point with value=index+c,lem:limit point with value=index+c,lem:limit point with value > theta+c,lem:limit point with value<index,lem:limit point with value between theta and theta+c,lem:limit point with value between theta and theta+c,lem:limit point with value=theta,lem:primal-dual-3,lem:primal-dual-9,lem:primal-dual-10}.
\end{proof}

\subsection{Equilibrium Characterization for Convex Priors with Low Cost}
\label{subsec:low-cost-equilibrium}
\begin{theorem}
\label{thm:convex-prior-unique-equilibrium}
    If all senders share the same prior $F$ that is (weakly) convex, with all values in its support at least~$c$, 
    then each sender fully revealing his value is the unique symmetric equilibrium.
\end{theorem}

\begin{proof}[Proof of \Cref{thm:convex-prior-unique-equilibrium}]
When all senders adopt the full-revelation strategy, it is easy to check that the corresponding strategy profile induced is a symmetric equilibrium.
We now show that this is the unique symmetric equilibrium.

Since we consider symmetric equilibria, we omit the subscripts for senders.
For the sake of contradiction, suppose $G^*$ is another symmetric equilibrium strategy that does not fully reveal the value.
Let $K$ denote the distribution of the amortized value induced by the 2-D distribution~$G^*$.
By the MPC condition, we know that $K$ shifted by~$c$ is a MPC of~$F$, i.e., 
\begin{equation*}
    \int_{-c}^tK(x)\dd x\le \int_0^{t+c}F(x)\dd x~,\quad \forall t\in[-c,1-c]~.
\end{equation*}
Besides, we have $K(-c)=F(0)=0$, $K(1-c)=F(1)=1$, and
\begin{equation*}
    \int_{-c}^{1-c}K(x)\dd x= \int_0^{1}F(x)\dd x~.
\end{equation*}
Let $u(v)$ denote the contribution of value $v\in[c,1]$ to the sender's expected utility.
Then, we can rewrite the expected utility of the sender as $\int_{c}^{1}u(v)f(v)\dd v$.
Under the fact that strategy $G^*$ forms a symmetric equilibrium, function $K$ must be continuous over $[-c,1-c]$; otherwise, each sender has a profitable deviation by spreading the index where the mass is. 
Under this fact, we have the corresponding function $\phi(x)=K^{N-1}(x)$ for any $x\in[-c,1-c]$.
Each sender can always choose to send no information, thereby obtaining an expected utility of \( \int_c^1 K^{N -1}(v-c)f(v)\dd v \). 
Therefore, each sender is guaranteed to achieve at least an expected utility of \(  \int_c^1 K^{N-1}(v-c)f(v)\dd v \), that is
\begin{align*}
    \int_{c}^{1}u(v)f(v)\dd v&\ge \int_c^1 K^{N-1}(v-c)f(v)\dd v=\int_c^1 f(v)\dd \left(\int_{-c}^{v-c}K^{N-1}(t)\dd t\right)\\
    &=\left[f(v)\int_{-c}^{v-c}K^{N-1}(t)\dd t\right]\Bigg|_{v=c}^{1}-\int_c^1\int_{-c}^{v-c}K^{N-1}(t)\dd t\dd f(v)\\
    &=\left[f(v)\int_{c}^{v}F^{N-1}(t)\dd t\right]\Bigg|_{v=c}^{1}-\int_c^1\int_{-c}^{v-c}K^{N-1}(t)\dd t\dd f(v)~.
\end{align*}
Since $\int_{-c}^{v-c}K^{N-1}(t)\dd t\le \int_c^{v}F^{N-1}(t)\dd t$ for any $v\in[c,1]$ and the density of the prior is non-decreasing over $[c,1]$, we have 
\begin{align*}
    \int_{c}^{1}u(v)f(v)\dd v&\ge\left[f(v)\int_{c}^{v}F^{N-1}(t)\dd t\right]\Bigg|_{v=c}^{1}-\int_c^1\int_c^{v}F^{N-1}(t)\dd t\dd f(v)\\
    &=\int_c^1F^{N-1}(v)\dd F(v)=\frac{1}{N}~.
\end{align*}
In any symmetric equilibrium of the $N$-senders game, each sender can obtain at most an expected utility of \( \sfrac{1}{N} \).
Combining these two facts, we obtain that $\int_{c}^{1}u(v)f(v)\dd v=\sfrac{1}{N}$.
Therefore, all the inequalities above must hold with equality, where we have 
\begin{equation*}
    \int_{-c}^{t}K(x)\dd x=\int_0^{t+c} F(x)\dd x,\quad \forall t\in[-c,1-c]~.
\end{equation*}
This fact is equivalent to the fact that the symmetric equilibrium strategy $G^*$ is indeed induced by the full-revelation strategy, which forms a contradiction with the assumption.
Thus, we have proved that the symmetric equilibrium induced by the full-revelation strategy is the unique symmetric equilibrium.
\end{proof}

For an intuition that full revelation constitutes an equilibrium, observe that, for each sender, the interim utility $u(\kappa) = F^{N-1}(\kappa + c)$ is convex in this scenario; any pooling of values decreases the expected utility by Jensen's inequality.
Intuitively, fully revealing the values is a best response.

The uniqueness of this symmetric equilibrium is nontrivial.
As a starting point, note that the senders' game is constant-sum  for such a small $c$:
since all values are at least~$c$, the index of each sender is always non-negative, regardless of the signaling; therefore, the buyer never quits without selecting a seller.
Thus, in any symmetric equilibrium, each sender must share a common utility of $\sfrac{1}{N}$. 
We show that if all senders except $i$ use a common strategy which does not always fully reveal the values,
sender $i$ can get a utility strictly higher than $\sfrac{1}{N}$ by fully revealing $v_i$.
The uniqueness of the symmetric equilibrium follows.

Note that the market achieves maximal efficiency under this equilibrium: the receiver fully learns the relevant information from the senders' signaling and therefore always selects the best box, attaining first-best revenue.

\subsection{Proof of \Cref{thm:symmetric-equilibrium-convex-prior-big-c}}

\begin{proof}[Proof of \Cref{thm:symmetric-equilibrium-convex-prior-big-c}]
Denote the shifted prior by $\tilde{F}(x) := F(x+c)$ for $x\in[-c,1-c]$ (with density $\tilde{f}$).
By the MPC condition, we know that the amortized value distribution of any feasible strategy forms an MPC of the function $\tilde{F}$ over $[-c,1-c]$.
We divide this proof into three steps.
Here we provide the proof of the non-degenerate case where $\theta_1'=\theta_2<1-c$.

\xhdr{Step 1: There uniquely exist $\theta_1\in(-c,0)$ and $\theta_2\in(0,1-c]$ such that $\frac{\tilde{F}(\theta_1)}{-\theta_1}=\frac{\tilde{F}(\theta_2)-\tilde{F}(\theta_1)}{\theta_2}$, and $K\in\MPC(\tilde{F})$ over $[-c,1-c]$}
For any $\theta_1\in(-c,0)$, we can construct the following linear function:
\begin{equation*}
    L_{\theta_1}(\theta)=\frac{\tilde{F}(\theta_1)}{-\theta_1}(\theta-\theta_1)~,\quad \forall \theta\in[\theta_1,1-c]~.
\end{equation*}
There exists an $\epsilon_1\in(0,c)$ such that for any $\theta_1\in(-c,-c+\epsilon)$, it holds that $L_{\theta_1}(\theta)\le \tilde{F}(\theta)$ for any $\theta\in[\theta_1,1-c]$. 
There exists an $\epsilon_2\in(0,c)$ such that for any $\theta_1\in(-\epsilon_2,0)$, the function $L_{\theta_1}$ has at least one intersection point with function $\tilde{F}$.
Since $\tilde{F}(\theta_1)$ and $\sfrac{\tilde{F}(\theta_1)}{-\theta_1}$ are both strictly increasing in $\theta_1$ over $(-c,0)$, we know that there uniquely exists a $\underline{\theta}_1\in(-c,0)$ such that $L_{\underline{\theta}_1}(\theta)\le \tilde{F}(\theta)$ for any $\theta\in[\underline{\theta}_1,1-c]$, and function $L_{\underline{\theta}_1}$ and function $\tilde{F}$ are tangent at some points.
Furthermore, there uniquely exists a $\overline{\theta}_1\in(\underline{\theta}_1,0)$ such that $L_{\underline{\theta}_1}(1-c)= \tilde{F}(1-c)$.
Thus for any $\theta_1\in(\underline{\theta}_1,0)$, function $L_{\theta_1}$ and function $\tilde{F}$ have at least two intersection points within the interval $[\theta_1,1-c]$.
For any $\theta_1\in(\underline{\theta}_1,0)$, we define $\theta_2=\max\{\theta\in[0,1-c]|L_{\theta_1}(\theta)=\tilde{F}(\theta)\}$ and $\theta_3=\min\{\theta\in[0,1-c]|L_{\theta_1}(\theta)=\tilde{F}(\theta)\}$.
Using the three new parameters, we can split the integral difference from $\theta_1$ to $\theta_2$ into two parts:
\begin{align*}
    &\int_{\theta_1}^{\theta_2}\tilde{F}(\theta)\dd \theta-  \int_{\theta_1}^{\theta_2}L_{\theta_1}(\theta)\dd \theta\\
    =&\left(\int_{\theta_3}^{\theta_2}\tilde{F}(\theta)\dd \theta-  \int_{\theta_3}^{\theta_2}L_{\theta_1}(\theta)\dd \theta \right)-\left(\int_{\theta_1}^{\theta_3}L_{\theta_1}(\theta)\dd \theta-\int_{\theta_1}^{\theta_3}\tilde{F}(\theta)\dd \theta \right)~.
\end{align*}
There exists an $\epsilon_3\in(0,\overline{\theta}_1-\underline{\theta}_1)$ such that $\int_{\theta_1}^{\theta_2}\tilde{F}(\theta)\dd \theta-  \int_{\theta_1}^{\theta_2}L_{\theta_1}(\theta)\dd \theta<0$ holds for any $\theta_1\in(\underline{\theta}_1,\underline{\theta}_1+\epsilon_3)$.
There exists an $\epsilon_4\in(0,\overline{\theta}_1-\underline{\theta}_1)$ such that $\int_{\theta_1}^{\theta_2}\tilde{F}(\theta)\dd \theta-  \int_{\theta_1}^{\theta_2}L_{\theta_1}(\theta)\dd \theta>0$ holds for any $\theta_1\in(\underline{\theta}_1,\underline{\theta}_1+\epsilon_4)$.
Besides, the first part $\int_{\theta_3}^{\theta_2}\tilde{F}(\theta)\dd \theta-  \int_{\theta_3}^{\theta_2}L_{\theta_1}(\theta)\dd \theta$ is decreasing at $\theta_1$ over $(\underline{\theta}_1,0)$, while the second part
$\int_{\theta_1}^{\theta_3}L_{\theta_1}(\theta)\dd \theta-\int_{\theta_1}^{\theta_3}\tilde{F}(\theta)\dd \theta$ is increasing at $\theta_1$ over $(\underline{\theta}_1,0)$.
So we know that there uniquely exists a $\theta_1\in(\underline{\theta}_1,\overline{\theta}_1)$ such that $\int_{\theta_1}^{\theta_2}\tilde{F}(\theta)\dd \theta-  \int_{\theta_1}^{\theta_2}L_{\theta_1}(\theta)\dd \theta=0$.
Finally, we prove that such a pair of parameters $(\theta_1,\theta_2)$ exists and is unique.

\xhdr{Step 2: Given $\theta_1$ and $\theta_2$ determined in Step~1, there exists a 2-D distribution that induces the hinge-shaped amortized value distribution}
We prove this by a construction.
We divide this part of the proof into two cases, based on the shape of $K$, the amortized value distribution.

\sxhdr{Case 1: $K(\theta_1') < 1$}
In this case, $\theta_2 = \theta_1' <1-c$.
Values in the $[0,\theta_1 + c]$ and $[\theta_2 - c,1]$ are fully revealed, so that $K$  and $\tilde F$ coincide on these intervals.
We are to pool values in $[\theta_1+c,\theta_2+c]$ to produce amortized values uniformly distributed over $[0, \theta_2]$.
Next, we show how to make such poolings.

Based on the structure of the function $K$, the slope over $[0,\theta_2]$ is denoted by $\rho=\frac{\tilde{F}(\theta_2)-\tilde{F}(\theta_1)}{\theta_2}$.
Since $K\in\MPC(\tilde{F})$, it holds that $\tilde{f}(0)<\rho$ and $\tilde{f}(\theta_2)>\rho$.
Besides, the density $\tilde{f}$ is increasing over $[0,\theta_2]$.
So there exists a row $\theta^*\in(0,\theta_2)$ such that $\tilde{f}(\theta^*)=K(\theta^*)=\rho$.
We need to ensure the index equation is preserved while pooling the values from $[\theta_1 + c, \theta_2 + c]$ into indices $[0, \theta_2]$, with the density of each row equal to $\rho$.
We divide the density of each value $v\in[\theta^*+c,\theta_2+c]$ into two parts: we keep a portion of size $\rho$ unpooled to individually form an index $v-c$, while the remaining density will be pooled together with lower values.
In details, first we let $\supp(G_{\cdot|\theta})=\{\theta+c\}$ and $g(\theta+c,\theta)=\rho$ for any $\theta\in(\theta^*,\theta_2)$.
Then, we show how to pool the values from $[\theta_1+c,\theta^*+c]$ and $[\theta^*+c,\theta_2+c]$ together to form indices over $[0,\theta^*]$.
We define functions $\alpha:[0,\theta^*]\rightarrow[0,1]$ and $\beta:[0,\theta^*]\rightarrow[0,1]$.
We assume that the support of the conditional distribution $G_\theta$ is $\{\alpha(\theta),\beta(\theta)\}$ for each $\theta\in[0,\theta^*]$.
We give the initial values of these two functions as $\alpha(0)=\theta_1+c$ and $\beta(0)=\theta_2+c$.
Given the initial conditions for the functions $\alpha$ and $\beta$, we now only need to compute their derivatives to fully characterize these two functions.
We consider the following two cases to show how to specifically compute the derivatives of these two functions.

\textbf{case-1.1: $\alpha(\theta)>\theta$.} 
In this case, we have these two equations:
\begin{equation*}
    \frac{f(\alpha(\theta))}{\alpha'(\theta)}+\frac{f(\beta(\theta))-\rho}{-\beta'(\theta)}=\rho~,\quad
    (\alpha(\theta)-\theta)\frac{f(\alpha(\theta))}{\alpha'(\theta)\rho}+(\beta(\theta)-\theta)\frac{f(\beta(\theta))-\rho}{-\beta'(\theta)\rho}=c~.
\end{equation*}
The first equation means that the sum of densities at points $(\alpha(\theta),\theta)$ and $(\beta(\theta),\theta)$ is equal to $\rho$.
The second equation is exactly the index equation of row $\theta$.
Both equations are ordinary differential equations. 
By combining these two equations, we achieve that:
\begin{equation*}
    \alpha'(\theta)=\frac{(\beta(\theta)-\alpha(\theta))f(\alpha(\theta))}{\rho(\beta(\theta)-c-\theta)},\quad 
    \beta'(\theta)=\frac{(\beta(\theta)-\alpha(\theta))(\rho-f(\beta(\theta)))}{\rho(\theta+c-\alpha(\theta))}~.
\end{equation*}

\textbf{case-1.2: $\alpha(\theta)\le\theta$.} 
In this case, we have these two equations:
\begin{equation*}
f(\alpha(\theta))+\frac{f(\beta(\theta))-\rho}{-\beta'(\theta)}=\rho~,\quad
    (\beta(\theta)-\theta)\frac{f(\beta(\theta))-\rho}{-\beta'(\theta)}=c\left(\frac{f(\alpha(\theta))}{\alpha'(\theta)}+\frac{f(\beta(\theta))-\rho}{-\beta'(\theta)}\right)~.
\end{equation*}
The first equation means that the sum of densities at points $(\alpha(\theta),\theta)$ and $(\beta(\theta),\theta)$ is equal to $\rho$.
The second equation is exactly the index equation of row $\theta$.
By combining these two equations, we achieve that:
\begin{equation*}
    \alpha'(\theta)=\frac{-cf(\alpha(\theta))}{(f(\alpha(\theta))-\rho)(\beta(\theta)-\theta-c)}~,\quad 
    \beta'(\theta)=\frac{f(\beta(\theta))-\rho}{f(\alpha(\theta))-\rho}~.
\end{equation*}

In both subcases, for each $\theta\in[0,\theta^*]$, since the value $\alpha(\theta)$ and the value $\beta(\theta)$ are pooled together to form the index $\theta$, we have $\alpha(\theta)\le \theta+c\le\beta(\theta)$.
Since the prior density $f$ is increasing, we achieve that $\alpha'(\theta)\ge 0$ if $\alpha(\theta)\le \theta^*+c$, and $\beta'(\theta)\le 0$ if $\beta(\theta)\ge \theta^*+c$.
In addition to the initial values of both functions, there always exists some $\theta'\in[0,\theta^*]$ such that the function $\alpha$ increases over the interval $[0,\theta']$, while the function $\beta$ decreases over the same interval. 
We define the maximal such $\theta'$ as $\tilde{\theta}$, i.e., $\tilde{\theta} \triangleq \sup \{\theta \in [0,\theta^*] \mid \alpha'(\theta) \geq 0 \ \& \ \beta'(\theta) \leq 0\}$.
We next show $\tilde{\theta}=\theta^*$ via contradiction.
We assume that $\tilde{\theta}<\theta^*$, then it must hold that $\beta(\tilde{\theta})=\theta^*+c$
and $\alpha(\tilde{\theta})=\tilde{\theta}+c$.
Since the function $K$ constructed in Step-1 forms an MPC of the function $\tilde{F}$ over $[\theta_1,\theta_2]$, we have 
\begin{equation*}
    F(\theta_2+c)-F(\theta_1+c)=\rho \cdot\theta_2~.
\end{equation*}
Besides, based on the construction above, we have
\begin{equation*}
    \rho\cdot(\theta_2-\theta^*+\tilde{\theta})=F(\theta_2+c)-F(\theta^*+c)+F(\tilde{\theta}+c)-F(\theta_1+c)~.
\end{equation*}
Combining these two equations, we achieve that
\begin{equation*}
    \rho\cdot (\theta^*-\tilde{\theta})=F(\theta^*+c)-F(\tilde{\theta}+c)~,
\end{equation*}
which forms a contradiction since $F(\theta^*+c)-F(\tilde{\theta}+c)<\rho\cdot (\theta^*-\tilde{\theta})$.
Thus, we have proved that $\tilde{\theta}=\theta^*$.
Under this fact, we complete the proof by showing that $\alpha(\theta^*)=\beta(\theta^*)=\theta^*+c$.
We also prove it through contradiction. 
Suppose $\alpha(\theta^*)=a<\theta^*+c<b=\beta(\theta^*)$.
Since the constructed function $K$ forms an MPC of the function $\tilde{F}$ over $[\theta_1,\theta_2]$, it holds that
\begin{equation*}
    F(\theta_2+c)-F(\theta_1+c)=\rho \cdot\theta_2~.
\end{equation*}
Besides, based on the construction above, we have
\begin{equation*}
    \rho\cdot\theta_2=F(\theta_2+c)-F(b)+F(a)-F(\theta_1+c)~.
\end{equation*}
These two equations form a contradiction, which makes the assumption invalid.
Till now, we have proved that the method we provided above indeed can construct such 2-D distribution that induces the amortized value distribution defined in \Cref{thm:symmetric-equilibrium-convex-prior-big-c}.

\sxhdr{Case 2: There exists $\epsilon>0$ such that the amortized value distribution $K$ is constant over $(1-c-\epsilon,1-c)$}
The construction and proof follow a similar (and much simpler) approach to Case 1, so we omit them here.

\xhdr{Step 3: Equilibrium verification} 
Based on the techniques in \Cref{sec:verification}, we can verify that the strategy constructed above indeed constitutes a best response to the interim utility function $u$.
Using the \Cref{alg:lambda,alg:mu}, we can construct such $\lambda$ and $\mu$:
\begin{equation*}
    \lambda(v)=
    \begin{cases}
        K(v-c) \quad &\text{if}\ v\in[0,\theta_1+c)\cup[c,1]~,\\
        \rho(v-c)+F(\theta_1+c) \quad &\text{if}\ v\in[\theta_1+c,c)~,
    \end{cases}
\end{equation*}
and $$\mu(\theta)=
\begin{cases}
    \lambda'(\theta+c)\quad &\text{if}\ \theta\in[0,\theta_1+c)\cup(\theta_1+c,\theta_2+c)\cup(\theta_2+c,1]~,\\
    \rho\quad &\text{if}\ \theta\in\{\theta_1+c,\theta_2+c\}~.
\end{cases}
$$
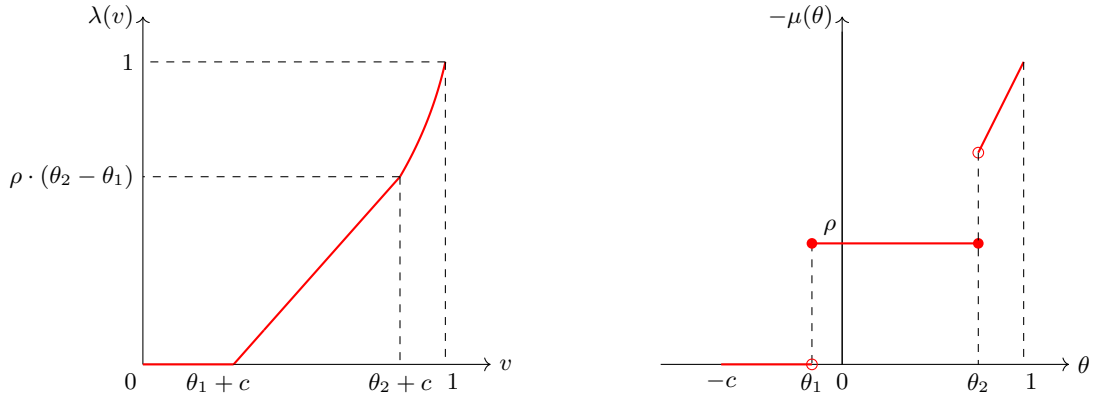
\begin{figure}[h]
    \centering
    \begin{minipage}[t]{0.48\textwidth}
        \centering
        \begin{tikzpicture}[scale=2.0]
            \draw[->] (-0.8,0) -- (1.5,0) node[right] {\footnotesize $v$};
            \draw[->] (-0.8,0) -- (-0.8,2.3) node[left] {\footnotesize $\lambda(v)$};

        
        
            
            \draw[red, thick] (0,0.226) -- (0.9, 1.24078);
            \draw[red, thick, domain=0.9:1.2, samples=50] plot (\x, {-(7-(\x+1.38113)^2)^0.5+2.5811});
            \draw[red, thick] (0,0.226) -- (-0.2,0);
            \draw[red, thick] (-0.8,0) -- (-0.2,0);
            
            \draw (-0.8,0) node[below] {\footnotesize $0\quad$};
            \draw[dashed] (-0.2,0) -- (-0.2,0) node[below] {\footnotesize $\theta_1+c\quad\;$};
            \draw[dashed] (0.9,1.24078) -- (0.9,0) node[below] {\footnotesize $\theta_2+c$};
            \draw[dashed] (0.9,1.24078) -- (-0.8,1.24078) node[left] {\footnotesize $\rho\cdot (\theta_2-\theta_1)$};
            \draw[dashed] (1.2,2) -- (1.2,0);
            \draw (1.25,0) node[below] {\footnotesize $1$};
            \draw[dashed] (1.2,2) -- (-0.8,2) node[left] {\footnotesize $1$};
        \end{tikzpicture}
    \end{minipage}
    \hspace{0.02\textwidth}
    \begin{minipage}[t]{0.48\textwidth}
        \centering
        \begin{tikzpicture}[scale=2.0]
            \draw[->] (-1.2,0) -- (1.5,0) node[right] {\footnotesize $\theta$};
            \draw[-] (0,0) -- (0,2.2) node[left] {};
            \draw[->] (0,0) -- (0,2.3) node[left] {\footnotesize $-\mu(\theta)$};

        
        
            
            \draw[red, thick] (-0.2,0) -- (-0.8,0);
            
            \draw (-0.8,0) node[below] {\footnotesize $-c$};
            \draw (0,0) node[below] {\footnotesize $0$};
            \draw[dashed] (-0.2,0.8) -- (-0.2,0) node[below] {\footnotesize $\theta_1$};
            \draw[red] (-0.2,0) circle (1pt);
            \fill[red] (-0.2,0.8) circle (1pt);
            \draw[red, thick] (-0.2,0.8) -- (0.9,0.8);
            \fill[red] (0.9,0.8) circle (1pt);
            \draw[red] (0.9,1.4) circle (1pt);
            \draw[dashed] (0.9,1.4) -- (0.9,0) node[below] {\footnotesize $\theta_2$};
            \draw[red, thick] (0.9,1.4) -- (1.2,2);
            \draw[dashed] (1.2,2) -- (1.2,0);
            \draw (1.25,0) node[below] {\footnotesize $1$};
            \draw (0,0.8) node[above] {\footnotesize $\rho\quad$};
        \end{tikzpicture}
    \end{minipage}
    \caption{The dual variables constructed by the Recursive Reduction Procedure, based on the symmetric equilibrium constructed in \Cref{thm:symmetric-equilibrium-convex-prior-big-c}}
    \label{fig:dual-constructed}
\end{figure}
Based on the formulas of $\lambda$ and $\mu$, we directly have that for any $(v,\theta)\in\supp(G)$:
\begin{equation*}
    \lambda(v)=p(v,\theta)-\mu(\theta)q(v,\theta)~,
\end{equation*}
which implies that $G$ and $(\lambda,\mu)$ satisfy the Complementary Slackness conditions in \Cref{thm:primal-dual-optimal}.
Besides, it also can be verified that such $\lambda$ and $\mu$ are feasible to the corresponding dual problem, that is:
\begin{equation*}
    \lambda(v)=\max_{\theta\in[-c,1-c]}p(v,\theta)-\mu(\theta)q(v,\theta)~,\quad\forall v\in[0,1]~,
\end{equation*}
which makes $(\lambda,\mu)$ feasible to the corresponding dual problem.
Combining these two results and by \Cref{thm:constructed-lambda-mu-optimal}, we know that such $G$ indeed constitutes a best response to the interim utility function $u$.
\end{proof}

\section{Proofs in \Cref{sec:approximation}}
\label{apx:approximation}
\begin{proof}[Proof of \Cref{thm:fptas}]
We proceed with the proof in the following three steps.

\xhdr{Step 1: Discretizing the original prior $F$}
Fix a value-grid size $\delta\in(0,1)$.
Let $M\triangleq\lceil 1/\delta\rceil$ be the number of intervals and partition $[0,1]$ into intervals
$$
I_r
\triangleq
\begin{cases}
\left[\frac{r-1}{M},\frac{r}{M}\right)\quad & \mathrm{if\ }r\in[M-1]~,\\[4pt]
\left[\frac{M-1}{M},1\right]\quad & \mathrm{if\ }r=M~.
\end{cases}
$$
For each interval $I_r$, let $\widehat v_r$ be its midpoint $\widehat v_r
\triangleq
\frac{2r-1}{2M}$ for each $r\in[M]$.
Then, we move all probability mass of the original prior $F$ inside each interval $I_r$ to the representative value $\widehat v_r$:\footnote{The specific representative point inside each interval is immaterial for the approximation argument; we use the midpoint only for convenience.}
$$
F_\delta
\triangleq
\sum_{r=1}^M F(I_r)\cdot \delta_{\widehat v_r}~,
$$
where $\delta_{\widehat v_r}$ denotes the Dirac measure at $\widehat v_r$.

Under the discretized prior $F_\delta$, we obtain the following auxiliary best-response problem:
\begin{equation}
\label{eq:best-response-delta}
    \max_{G_\delta\in\mathcal G(F_\delta,c)}
    U(G_\delta)=
    \mathbb E_{(v,\theta)\sim G_\delta}
    \left[
    u\left(\min\{v,\theta\}\right)
    \right]~.
    \tag{$\mathcal P^{\delta}_{\cc{BR}}$}
\end{equation}
Since the feasible space $\mathcal G(F_\delta,c)$ is convex and weakly compact, and $u(\kappa)$ is continuous, problem \eqref{eq:best-response-delta} still admits an optimal solution.
Let $G_\delta^*$ be an optimal solution to the problem \eqref{eq:best-response-delta}, and denote its optimal value by $\mathsf{OPT}_\delta
\triangleq
U(G_\delta^*)$.
Notice that \eqref{eq:best-response} and \eqref{eq:best-response-delta} are not nested optimization problems.
Indeed, their feasible sets are $\mathcal G(F,c)$ and $\mathcal G(F_\delta,c)$, respectively, and neither feasible set contains the other.
Therefore, there is no direct monotonicity relation between $\mathsf{OPT}$ and $\mathsf{OPT}_\delta$.
Instead, we prove a stability bound in two steps:
$$
|\mathsf{OPT}_\delta-\mathsf{OPT}|
\le
L\delta~.
$$

\sxhdr{Step 1.1: From the original prior to the discretized prior}
Starting from the optimal solution $G^*\in\mathcal G(F,c)$ of the original problem \eqref{eq:best-response}, we construct a feasible solution $G_\delta\in\mathcal G(F_\delta,c)$ for the discretized-prior problem \eqref{eq:best-response-delta}, and prove that the construction loss is at most $L\delta$.
For $G^*_\theta$-almost every row $\theta$, we construct a discretized posterior by moving all mass in each value interval $I_r$ to its midpoint $\widehat v_r$.
Specifically, write $m_r(\theta)\triangleq G^*_{v\mid\theta}(I_r)$ for all $r\in[M]$.
Then we define the discretized posterior on the value grid by
$$
\widehat G_{v\mid\theta}
\triangleq
\sum_{r=1}^M m_r(\theta)\cdot \delta_{\widehat v_r}~.
$$
Let $\theta_\delta(\theta)$ be the index of this discretized posterior, i.e., the unique solution to
$$
\int_0^1
\left((v-\theta_\delta(\theta))_+-c\right)
\dd \widehat G_{v\mid\theta}(v)
=
0~.
$$
Thus, each original row $\theta$ is transformed as follows: value mass in every interval $I_r$ is moved horizontally to $\widehat v_r$, and the whole row with newly moved masses at each midpoint is then moved vertically from $\theta$ to the recomputed index $\theta_\delta(\theta)$.
Formally, define $G_\delta$ by
$$
G_\delta(A\times B)
\triangleq
\int_{-c}^{1-c}
\mathbf 1\{\theta_\delta(\theta)\in B\}
\widehat G_{v\mid\theta}(A)
\dd G^*_\theta(\theta)
$$
for all measurable $A\subseteq[0,1]$ and $B\subseteq[-c,1-c]$.

We first show that the vertical and horizontal movements are both small for every probability mass.
To see this, for every $v\in I_r$, we have $|v-\widehat v_r|\le\delta$.
Hence, for every row $\theta$,
$$
(\widehat v_r-(\theta+\delta))_+
\le
(v-\theta)_+~,\quad
(\widehat v_r-(\theta-\delta))_+
\ge
(v-\theta)_+~,
\qquad
\forall v\in I_r~.
$$
Integrating over all intervals with respect to $G^*_{v\mid\theta}$ gives
$$
\int_0^1
(\widehat v-(\theta+\delta))_+
\dd\widehat G_{v\mid\theta}(\widehat v)
\le c~,\quad
\int_0^1
(\widehat v-(\theta-\delta))_+
\dd\widehat G_{v\mid\theta}(\widehat v)
\ge c~.
$$
By monotonicity of the index equation in the index variable,
$$
|\theta_\delta(\theta)-\theta|\le\delta
\quad
\text{for }G^*_\theta\text{-almost every }\theta~.
$$

It remains to compare payoffs.
Recall that, for each original row $\theta$, all mass in interval $I_r$ is moved to the representative value $\widehat v_r$, and the row index is recomputed as $\theta_\delta(\theta)$.
For every $v\in I_r$, we have $|v-\widehat v_r|\le \delta$,
and from the index-stability argument above, $|\theta-\theta_\delta(\theta)|\le \delta$ for $G^*_\theta$-almost every $\theta$.
Therefore, for every $v\in I_r$ and $G^*_\theta$-almost every $\theta$,
$$
\left|
\min\{\widehat v_r,\theta_\delta(\theta)\}
-
\min\{v,\theta\}
\right|
\le \delta~.
$$
Since the interim utility $u(\kappa)$ is $L$-Lipschitz continuous, this implies
$$
u\left(\min\{\widehat v_r,\theta_\delta(\theta)\}\right)
\ge
u\left(\min\{v,\theta\}\right)-L\delta~.
$$
Applying the pointwise payoff bound on each interval $I_r$, we obtain
$$
\begin{aligned}
U(G_\delta)
&=
\int_{-c}^{1-c}
\sum_{r=1}^M
u\left(\min\{\widehat v_r,\theta_\delta(\theta)\}\right)
G^*_{v\mid\theta}(I_r)
\dd G^*_\theta(\theta)\\
&\ge
\int_{-c}^{1-c}
\sum_{r=1}^M
\int_{I_r}
\left[
u\left(\min\{v,\theta\}\right)-L\delta
\right]
\dd G^*_{v\mid\theta}(v)
\dd G^*_\theta(\theta)\\
&=
U(G^*)-L\delta=
\mathsf{OPT}-L\delta~.
\end{aligned}
$$
Since $G_\delta^*$ is optimal for \eqref{eq:best-response-delta}, we conclude that
$$
\mathsf{OPT}_\delta
\ge
U(G_\delta)
\ge
\mathsf{OPT}-L\delta~.
$$

\sxhdr{Step 1.2: From the discretized prior back to the original prior}
Conversely, starting from the optimal solution $G_\delta^*\in\mathcal G(F_\delta,c)$ of the discretized-prior problem \eqref{eq:best-response-delta}, we construct a feasible solution $G\in\mathcal G(F,c)$ for the original problem \eqref{eq:best-response}.

For $G^*_{\delta,\theta}$-almost every row $\theta$, write $m_r(\theta)
\triangleq
G^*_{\delta,v\mid\theta}(\{\widehat v_r\})$ for all $\forall r\in[M]$.
We lift the discretized posterior by spreading the mass at each midpoint $\widehat v_r$ back to the interval $I_r$ according to the conditional distribution of $F$ on $I_r$.
Specifically, define a posterior distribution $\widetilde G_{v\mid\theta}$ on $[0,1]$ by
$$
\widetilde G_{v\mid\theta}(A)
\triangleq
\sum_{r=1}^M
m_r(\theta)
\frac{F(A\cap I_r)}{F(I_r)}
$$
for every measurable set $A\subseteq[0,1]$.
It holds $F(I_r)>0$ for each $I_r$.
Equivalently, each probability mass $m_r(\theta)$ originally concentrated at $\widehat v_r$ is redistributed over $I_r$ in the same proportion as the prior $F$.
Let $\widetilde\theta_\delta(\theta)$ be the index of this lifted posterior, i.e., the unique solution to
$$
\int_0^1
\left((v-\widetilde\theta_\delta(\theta))_+-c\right)
\dd \widetilde G_{v\mid\theta}(v)
=
0~.
$$
Thus, each row $\theta$ of $G_\delta^*$ is transformed as follows: mass at each midpoint $\widehat v_r$ is spread back to $I_r$, and the whole row is then moved vertically from $\theta$ to the recomputed index $\widetilde\theta_\delta(\theta)$.
Formally, define
$$
G(A\times B)
\triangleq
\int_{-c}^{1-c}
\mathbf 1\{\widetilde\theta_\delta(\theta)\in B\}
\widetilde G_{v\mid\theta}(A)
\dd G^*_{\delta,\theta}(\theta)
$$
for all measurable $A\subseteq[0,1]$ and $B\subseteq[-c,1-c]$.
It can be easily verified that $G\in\mathcal G(F,c)$.

We next show that the vertical and horizontal movements are both small.
For every $v\in I_r$, we have $|v-\widehat v_r|\le\delta$.
Since $G_\delta^*$ is feasible, for $G^*_{\delta,\theta}$-almost every row $\theta$, $\sum_{r=1}^M
m_r(\theta)
\left((\widehat v_r-\theta)_+-c\right)
=
0$.
For every $v\in I_r$, we have
$$
(v-(\theta+\delta))_+
\le
(\widehat v_r-\theta)_+~,\quad
(v-(\theta-\delta))_+
\ge
(\widehat v_r-\theta)_+~.
$$
Integrating with respect to the lifted posterior $\widetilde G_{v\mid\theta}$ gives
$$
\int_0^1
(v-(\theta+\delta))_+
\dd \widetilde G_{v\mid\theta}(v)
\le c~,\quad
\int_0^1
(v-(\theta-\delta))_+
\dd \widetilde G_{v\mid\theta}(v)
\ge c~.
$$
By monotonicity of the index equation in the index variable,
$$
|\widetilde\theta_\delta(\theta)-\theta|\le\delta
\quad
\text{for }G^*_{\delta,\theta}\text{-almost every }\theta~.
$$

It remains to compare payoffs.
For every $v\in I_r$ and $G^*_{\delta,\theta}$-almost every $\theta$, we have
$$
\left|
\min\{v,\widetilde\theta_\delta(\theta)\}
-
\min\{\widehat v_r,\theta\}
\right|
\le\delta~.
$$
Since the interim utility function $u(\kappa)$ is $L$-Lipschitz continuous, we can bound the utility loss
$$
u\left(\min\{v,\widetilde\theta_\delta(\theta)\}\right)
\ge
u\left(\min\{\widehat v_r,\theta\}\right)-L\delta~.
$$
Applying this pointwise payoff bound on each interval $I_r$, we obtain
$$
\begin{aligned}
U(G)
&=
\int_{-c}^{1-c}
\sum_{r=1}^M
m_r(\theta)
\int_{I_r}
u\left(\min\{v,\widetilde\theta_\delta(\theta)\}\right)
\frac{\dd F(v)}{F(I_r)}
\dd G^*_{\delta,\theta}(\theta)\\
&\ge
\int_{-c}^{1-c}
\sum_{r=1}^M
m_r(\theta)
\left[
u\left(\min\{\widehat v_r,\theta\}\right)-L\delta
\right]
\dd G^*_{\delta,\theta}(\theta)\\
&=
U(G_\delta^*)-L\delta=
\mathsf{OPT}_\delta-L\delta~.
\end{aligned}
$$
Since $G^*$ is optimal for the original problem, we conclude that
$$
\mathsf{OPT}
\ge
U(G)
\ge
\mathsf{OPT}_\delta-L\delta~.
$$

Combining the two inequalities from steps 1.1 and 1.2 yields
$$
|\mathsf{OPT}_\delta-\mathsf{OPT}|
\le
L\delta~.
$$

\xhdr{Step 2: Discretizing the index support}
Consider the discrete prior $F_\delta$ constructed in Step~1.
Fix an index-discretization parameter $\epsilon\in(0,1)$.
Recall that the discretized prior constructed in the previous step is $F_\delta=\sum_{i=1}^M f_i^\delta\cdot \delta_{\widehat v_i}$ and $f_i^\delta\triangleq F(I_i)$.
For each grid value $\widehat v_i$ and each index $\theta\in[-c,1-c]$, define these two functions for convenience $q_i(\theta)\triangleq(\widehat v_i-\theta)_+-c$ and $p_i(\theta)\triangleq u(\min\{\widehat v_i,\theta\})$.
Let $\Theta_\delta\triangleq[\widehat v_1-c,\widehat v_M-c]$ be the effective index domain under the discretized prior.

\sxhdr{Step 2.1: Constructing the index grid}
We first discretize the vertical axis of the interim utility function $u(\kappa)$ at granularity $\epsilon$.
Let $\ell_k\triangleq k\epsilon$ for $k=0,1,\ldots,\lceil 1/\epsilon\rceil$.
For every level $\ell_k\le u(\widehat v_M-c)$, define the generalized inverse $\xi_k
\triangleq
\inf\{\theta\in\Theta_\delta:u(\theta)\ge \ell_k\}$.
Since the interim utility $u(\kappa)$ is continuous and nondecreasing, we have $u(\xi_k)\ge \ell_k$ for each $k$.

Then, the full index grid is constructed as
$$
\Gamma
\triangleq
\{\widehat v_i-c:i\in[M]\}
\cup
\left(\{\widehat v_i:i\in[M]\}\cap\Theta_\delta\right)
\cup
\{\xi_k:k=0,1,\ldots,\lceil 1/\epsilon\rceil\}~.
$$
After sorting and removing duplicates, write
$$
\Gamma=\{\theta_1,\ldots,\theta_R\}~,
\qquad
\widehat v_1-c=\theta_1<\cdots<\theta_R=\widehat v_M-c~.
$$
The grid size satisfies $|\Gamma|\triangleq R\le 2M+\lceil 1/\epsilon\rceil+3$, which is finite.

We illustrate why constructing such an index grid.
The points $\widehat v_i$ are added so that the map $\theta\mapsto\min\{\widehat v_i,\theta\}$ has no kink inside any grid interval.
The points $\widehat v_i-c$ are added so that the coefficient $q_i(\theta)=(\widehat v_i-\theta)_+-c$ does not change sign inside any open interval between two consecutive grid points.
The points $\xi_k$ ensure the one-sided payoff approximation used in the rounding step.
Specifically, if $\alpha<\beta$ are two consecutive points in $\Gamma$, then for every $\theta\in[\alpha,\beta)$,
$$
u(\theta)\le u(\alpha)+\epsilon~.
$$
To see this, if $u(\theta)>u(\alpha)+\epsilon$ for some $\theta\in[\alpha,\beta)$, then there is a level $\ell_k$ such that $u(\alpha)<\ell_k\le u(\theta)$.
By definition of $\xi_k$, we have $\xi_k\le \theta<\beta$.
If $\xi_k\le\alpha$, then monotonicity implies $u(\alpha)\ge u(\xi_k)\ge\ell_k$, a contradiction.
Hence $\xi_k\in(\alpha,\beta)$, contradicting the fact that $\alpha$ and $\beta$ are consecutive grid points.

\sxhdr{Step 2.2: Defining the finite LP}
We now restrict the index support to the finite grid $\Gamma$.
Let $\boldsymbol x=(x_{i r})_{i\in[M],r\in[R]}$ be the decision variable, where $x_{i r}\ge0$ denotes the probability mass assigned to value $\widehat v_i$ and index $\theta_r$.
The discretized problem is the following finite-dimensional LP:
\begin{equation}
\label{eq:best-response-delta-epsilon}
\tag{$\mathcal P^{\delta,\epsilon}_{\cc{BR}}$}
\begin{aligned}
\max_{\boldsymbol x\ge0}\quad
&\sum_{i=1}^M\sum_{r=1}^R
u(\min\{\widehat v_i,\theta_r\})x_{i r}\\ 
\mathrm{s.t.}\quad
&\sum_{r=1}^R x_{i r}=f_i^\delta~,
&&\forall i\in[M]~,\\
&\sum_{i=1}^M
\big((\widehat v_i-\theta_r)_+-c\big)x_{i r}=0~,
&&\forall r\in[R]~.
\end{aligned}
\end{equation}
The LP is feasible since the full-revelation strategy belongs to its feasible region: for each value $\widehat v_i$, assign mass $f_i^\delta$ to the index $\widehat v_i-c$.
Let $G_{\delta,\epsilon}^*$ be an optimal solution to \eqref{eq:best-response-delta-epsilon}, and let $\mathsf{OPT}_{\delta,\epsilon}
\triangleq
U(G_{\delta,\epsilon}^*)$ be its optimal value.
Every feasible solution of \eqref{eq:best-response-delta-epsilon} is feasible for \eqref{eq:best-response-delta}, because it is simply a feasible joint distribution under $F_\delta$ with finite index support.
Therefore,
$$
\mathsf{OPT}_{\delta,\epsilon}\le \mathsf{OPT}_\delta~.
$$

\sxhdr{Step 2.3: Rounding an optimal solution to the index grid}
Let $G_\delta^*\in\mathcal G(F_\delta,c)$ be an optimal solution to the problem \eqref{eq:best-response-delta} under discrete prior $F_\delta$. In this step, we round all mass of $G_\delta^*$ on non-grid index rows to the finite grid $\Gamma$, while preserving feasibility.

For each grid value $\widehat v_i$, let $\gamma_i^*$ be the finite Borel measure on $\Theta_\delta$ induced by $G_\delta^*$:
$$
\gamma_i^*(A)
\triangleq
G_\delta^*(\{\widehat v_i\}\times A)~,
\qquad
\forall A\in\mathcal B(\Theta_\delta)~.
$$
It suffices to describe the rounding operation on one open interval between two consecutive grid points. 
Fix $\alpha<\beta$ in $\Gamma$ with no grid point between them, and let $J\triangleq(\alpha,\beta)$. We will move all mass with indices in $J$ to the two endpoint rows $\alpha$ and $\beta$.

Let $M_J\triangleq\sum_{i=1}^M \gamma_i^*|_J$ be the total index-mass measure on $J$. If $M_J(J)=0$, there is nothing to move. Otherwise, for each $i\in[M]$, define the row composition at $\theta$
$$
z_i(\theta)
\triangleq
\frac{\dd \gamma_i^*|_J}{\dd M_J}(\theta)~.
$$
For $M_J$-almost every $\theta\in J$, the vector $z(\theta)=(z_1(\theta),\ldots,z_M(\theta))$ satisfies
$$
\sum_{i=1}^M z_i(\theta)=1~,
\qquad
\sum_{i=1}^M q_i(\theta)z_i(\theta)=0~.
$$

The key observation is that each such feasible row can be split into two feasible endpoint rows at $\alpha$ and $\beta$. 
Since $\Gamma$ contains all kink points $\widehat v_i$ and all zero-crossing points $\widehat v_i-c$, every $q_i$ is affine on $[\alpha,\beta]$ and has a fixed sign on $J$. Define
$$
P\triangleq\{i\in[M]:q_i(\theta)>0\}~,
\qquad
N\triangleq\{j\in[M]:q_j(\theta)<0\}~.
$$
Then these two sets, $P$ and $N$, form a partition $[M]$, and row feasibility (the index equation in \Cref{def:index}) implies that the positive and negative contributions balance:
$$
A(\theta)
\triangleq
\sum_{i\in P}q_i(\theta)z_i(\theta)
=
\sum_{j\in N}\bigl(-q_j(\theta)\bigr)z_j(\theta)>0~.
$$

For each pair $(i,j)\in P\times N$, define
$$
a_{ij}(\theta)\triangleq \frac{z_i(\theta)z_j(\theta)}{A(\theta)}~.
$$
Also define the two-value balanced package $r^{ij}(\theta)\in\mathbb R_+^M$ by
$$
r_k^{ij}(\theta)
\triangleq
\begin{cases}
-q_j(\theta)\quad & \mathrm{if}\ k=i~,\\
q_i(\theta)\quad & \mathrm{if}\ k=j~,\\
0\quad & \mathrm{otherwise}~.
\end{cases}
$$
This package is feasible at row $\theta$, because $\sum_k q_k(\theta)r_k^{ij}(\theta)=0$. Moreover, the row vector $z(\theta)$ can be decomposed as
$$
z(\theta)=\sum_{i\in P,\ j\in N}a_{ij}(\theta)r^{ij}(\theta)~.
$$
To see this, the $i$-th coordinate for $i\in P$ equals
$$
\sum_{j\in N}a_{ij}(\theta)(-q_j(\theta))
=
\frac{z_i(\theta)}{A(\theta)}
\sum_{j\in N}(-q_j(\theta))z_j(\theta)
=
z_i(\theta)~,
$$
and the verification for $j\in N$ is identical.

We now move each balanced package to the endpoints. Write
$$
\theta=\omega(\theta)\alpha+\bigl(1-\omega(\theta)\bigr)\beta~,
\qquad
\omega(\theta)\triangleq \frac{\beta-\theta}{\beta-\alpha}~.
$$
Since each $q_i$ is affine on $[\alpha,\beta]$, we have
$$
r^{ij}(\theta)
=
\omega(\theta)r^{ij}(\alpha)
+
\bigl(1-\omega(\theta)\bigr)r^{ij}(\beta)~.
$$
Therefore, define
$$
z^\alpha(\theta)
\triangleq
\omega(\theta)
\sum_{i\in P,\ j\in N}a_{ij}(\theta)r^{ij}(\alpha)~,
\qquad
z^\beta(\theta)
\triangleq
\bigl(1-\omega(\theta)\bigr)
\sum_{i\in P,\ j\in N}a_{ij}(\theta)r^{ij}(\beta)~.
$$
Then $z(\theta)=z^\alpha(\theta)+z^\beta(\theta)$. Moreover, $z^\alpha(\theta)$ is feasible at row $\alpha$, and $z^\beta(\theta)$ is feasible at row $\beta$, because each package $r^{ij}(\alpha)$ and $r^{ij}(\beta)$ is feasible at the corresponding endpoint row.

We now apply this local splitting operation to every open interval between consecutive grid points. Let $J_r\triangleq(\theta_r,\theta_{r+1})$ for $r\in[R-1]$. For each $J_r$, let $M_r$ be the total index-mass measure on $J_r$, and let $z^{r,L}$ and $z^{r,R}$ be the endpoint split vectors sent to $\theta_r$ and $\theta_{r+1}$, respectively. Define the rounded measures $\widehat\gamma_i$ by
$$
\widehat\gamma_i
\triangleq
\gamma_i^*|_\Gamma
+
\sum_{r=1}^{R-1}
\left(
\int_{J_r}z_i^{r,L}(\theta)\dd M_r(\theta)
\right)\delta_{\theta_r}
+
\sum_{r=1}^{R-1}
\left(
\int_{J_r}z_i^{r,R}(\theta)\dd M_r(\theta)
\right)\delta_{\theta_{r+1}}~.
$$
Finally, define the rounded joint distribution $\widehat G_{\delta,\epsilon}$ by
$$
\widehat G_{\delta,\epsilon}(\{\widehat v_i\}\times A)
\triangleq
\widehat\gamma_i(A)~,
\qquad
\forall A\in\mathcal B(\Theta_\delta)~.
$$
By construction, $\widehat G_{\delta,\epsilon}$ is supported on the grid points $\{\widehat v_1,\ldots,\widehat v_M\}\times\Gamma$.

We verify feasibility. First, the value marginal is preserved. For every $i\in[M]$,
$$
\begin{aligned}
\widehat\gamma_i(\Theta_\delta)
&=
\gamma_i^*(\Gamma)
+
\sum_{r=1}^{R-1}
\int_{J_r}
\left(z_i^{r,L}(\theta)+z_i^{r,R}(\theta)\right)
\dd M_r(\theta)\\
&=
\gamma_i^*(\Gamma)
+
\sum_{r=1}^{R-1}
\int_{J_r}z_i(\theta)\dd M_r(\theta)
=
\gamma_i^*(\Theta_\delta)
=
f_i^\delta~.
\end{aligned}
$$

Second, the row constraints are preserved. Since $\widehat G_{\delta,\epsilon}$ is supported only on $\Gamma$, it suffices to check the constraint at each grid row $\theta_r$. The mass originally placed on $\theta_r$ is feasible because $G_\delta^*$ is feasible. Every piece of mass moved into $\theta_r$ from a neighboring interval is feasible at $\theta_r$ by the endpoint-splitting construction. Hence, writing $\widehat x_{ir}\triangleq\widehat\gamma_i(\{\theta_r\})$, we have
$$
\sum_{i=1}^M
\big((\widehat v_i-\theta_r)_+-c\big)\widehat x_{ir}
=
0~,
\qquad
\forall r\in[R]~.
$$
Therefore, $\widehat G_{\delta,\epsilon}$ is feasible for the finite LP \eqref{eq:best-response-delta-epsilon}.

\sxhdr{Step 2.4: Bounding the rounding loss}
We now bound the payoff loss caused by the rounding operation for every probability mass.
Fix a grid interval $(\alpha,\beta)$ and a nonempty row $\theta\in(\alpha,\beta)$ with row vector $z(\theta)$.
This row is replaced by a value vector $z^\alpha(\theta)$ at row $\alpha$ and a value vector $z^\beta(\theta)$ at row $\beta$.
By construction of the grid, for value $\widehat v_i$,
$$
p_i(\alpha)\ge p_i(\theta)-\epsilon~,
\qquad
p_i(\beta)\ge p_i(\theta)~.
$$
Indeed, the payoff grid ensures that moving to the left endpoint loses at most $\epsilon$.
Therefore, the payoff after rounding this row is at least
$$
\begin{aligned}
\sum_{i=1}^M p_i(\alpha)z_i^\alpha(\theta)
+
\sum_{i=1}^M p_i(\beta)z_i^\beta(\theta)
&\ge
\sum_{i=1}^M
\bigl(p_i(\theta)-\epsilon\bigr)z_i^\alpha(\theta)
+
\sum_{i=1}^M
p_i(\theta)z_i^\beta(\theta)\\
&=
\sum_{i=1}^M
p_i(\theta)
\left(z_i^\alpha(\theta)+z_i^\beta(\theta)\right)
-
\epsilon\sum_{i=1}^M z_i^\alpha(\theta)\\
&\ge
\sum_{i=1}^M p_i(\theta)z_i(\theta)
-
\epsilon\sum_{i=1}^M z_i(\theta)~.
\end{aligned}
$$
Thus the loss on this row is at most $\epsilon\sum_{i=1}^M z_i(\theta)$.
Summing over all rows and all grid intervals, and using that the total probability mass is one, the aggregate payoff loss is at most $\epsilon=\epsilon$.
Hence
$$
\mathsf{OPT}_{\delta,\epsilon}
\ge
U(\widehat G_{\delta,\epsilon})
\ge
U(G_\delta^*)-\epsilon=\mathsf{OPT}_\delta-\epsilon~.
$$

Combining this with $\mathsf{OPT}_{\delta,\epsilon}\le\mathsf{OPT}_\delta$, we conclude
$$
\mathsf{OPT}_\delta-\epsilon
\le
\mathsf{OPT}_{\delta,\epsilon}
\le
\mathsf{OPT}_\delta~.
$$

\xhdr{Step 3: Lifting the finite-LP solution back to the original prior}
We now construct a feasible solution under the original prior $F$ from the finite-LP solution $G_{\delta,\epsilon}^*$.
Let $\boldsymbol x^*=(x^*_{ir})_{i\in[M],r\in[R]}$ be the optimal solution to the finite LP \eqref{eq:best-response-delta-epsilon}, so that $G_{\delta,\epsilon}^*$ assigns mass $x^*_{ir}$ to the value-index pair $(\widehat v_i,\theta_r)$ for each $i\in[M]$ and $r\in[R]$.

For each grid row $\theta_r$, write $Y_r\triangleq \sum_{i=1}^M x^*_{ir}$.
If $Y_r=0$, this row carries no mass and can be ignored.
If $Y_r>0$, we lift the posterior on row $\theta_r$ by spreading the mass at each midpoint $\widehat v_i$ back to the interval $I_i$ according to the conditional distribution of $F$ on $I_i$.
Specifically, define a posterior distribution $\widetilde G_{v\mid r}$ by, for every measurable set $A\subseteq[0,1]$,
$$
\widetilde G_{v\mid r}(A)
\triangleq
\sum_{i=1}^M
\frac{x^*_{ir}}{Y_r}
\frac{F(A\cap I_i)}{F(I_i)}~.
$$
Since $F$ is strictly increasing, $F(I_i)>0$ for every $i\in[M]$.
Let $\widetilde\theta_r$ be the index of this lifted posterior, i.e., the unique solution to the index equation
$$
\int_0^1
\left((v-\widetilde\theta_r)_+-c\right)
\dd \widetilde G_{v\mid r}(v)
=
0~.
$$
Thus, each grid row $\theta_r$ is transformed as follows: every mass $x^*_{ir}$ originally placed at the value midpoint $\widehat v_i$ is spread back over $I_i$, and the entire row is moved vertically from the original index $\theta_r$ to the newly recomputed index $\widetilde\theta_r$.
Formally, define $\widetilde G$ by
$$
\widetilde G(A\times B)
\triangleq
\sum_{r:Y_r>0}
\mathbf 1\{\widetilde\theta_r\in B\}
\sum_{i=1}^M
x^*_{ir}
\frac{F(A\cap I_i)}{F(I_i)}
$$
for all measurable $A\subseteq[0,1]$ and $B\subseteq[-c,1-c]$.

We first verify that $\widetilde G\in\mathcal G(F,c)$.
For the value marginal, for every measurable set $A\subseteq[0,1]$,
$$
\begin{aligned}
\widetilde G(A\times[-c,1-c])
=
\sum_{r=1}^R
\sum_{i=1}^M
x^*_{ir}
\frac{F(A\cap I_i)}{F(I_i)}
=
\sum_{i=1}^M
\frac{F(A\cap I_i)}{F(I_i)}
\sum_{r=1}^R x^*_{ir}
=
\sum_{i=1}^M F(A\cap I_i)
=
F(A)~.
\end{aligned}
$$
Hence the value marginal of $\widetilde G$ is $F$.
The row-wise index constraint also holds by construction.
Indeed, for any bounded measurable function $\varphi:[-c,1-c]\to\mathbb R$,
$$
\begin{aligned}
\int_{[0,1]\times[-c,1-c]}
\varphi(\epsilon)\left((v-\epsilon)_+-c\right)
\dd \widetilde G(v,\epsilon)=
\sum_{r:Y_r>0}
\varphi(\widetilde\theta_r)Y_r
\int_0^1
\left((v-\widetilde\theta_r)_+-c\right)
\dd \widetilde G_{v\mid r}(v)
=
0~.
\end{aligned}
$$
Therefore, the constructed strategy $\widetilde G$ is feasible for the original problem \eqref{eq:best-response}.

We next show that the vertical and horizontal movements are both small.
For every $v\in I_i$, we have $|v-\widehat v_i|\le\delta$.
Since $G_{\delta,\epsilon}^*$ is feasible for the finite LP, for every grid row $\theta_r$ with $Y_r>0$,
$$
\sum_{i=1}^M
\frac{x^*_{ir}}{Y_r}
(\widehat v_i-\theta_r)_+
=
c~.
$$
For every $v\in I_i$, we have
$$
(v-(\theta_r+\delta))_+
\le
(\widehat v_i-\theta_r)_+~,\quad
(v-(\theta_r-\delta))_+
\ge
(\widehat v_i-\theta_r)_+~.
$$
Integrating with respect to the lifted posterior $\widetilde G_{v\mid r}$ gives
$$
\int_0^1
(v-(\theta_r+\delta))_+
\dd \widetilde G_{v\mid r}(v)
\le c~,\quad
\int_0^1
(v-(\theta_r-\delta))_+
\dd \widetilde G_{v\mid r}(v)
\ge c~.
$$
By monotonicity of the index equation in the index variable,
$$
|\widetilde\theta_r-\theta_r|\le\delta
\quad
\text{for every }r\text{ with }Y_r>0~.
$$

It remains to compare payoffs.
For every $v\in I_i$ and every row $r$ with $Y_r>0$,
$$
\left|
\min\{v,\widetilde\theta_r\}
-
\min\{\widehat v_i,\theta_r\}
\right|
\le\delta~.
$$
Since the interim utility function $u(\kappa)$ is $L$-Lipschitz continuous,
$$
u\left(\min\{v,\widetilde\theta_r\}\right)
\ge
u\left(\min\{\widehat v_i,\theta_r\}\right)-L\delta~.
$$
Using the definition of $\widetilde G$, we obtain the following utility bound
$$
\begin{aligned}
U(\widetilde G)
&=
\sum_{r:Y_r>0}
\sum_{i=1}^M
x^*_{ir}
\int_{I_i}
u\left(\min\{v,\widetilde\theta_r\}\right)
\frac{\dd F(v)}{F(I_i)}\\
&\ge
\sum_{r:Y_r>0}
\sum_{i=1}^M
x^*_{ir}
\left[
u\left(\min\{\widehat v_i,\theta_r\}\right)-L\delta
\right]=
U(G_{\delta,\epsilon}^*)-L\delta=
\mathsf{OPT}_{\delta,\epsilon}-L\delta~.
\end{aligned}
$$

Combining this inequality with the bounds from Steps 1 and 2 gives
$$
U(\widetilde G)
\ge
\mathsf{OPT}-\epsilon-2L\delta~.
$$
Thus, the finite-LP solution $G^*_{\delta,\epsilon}$ can be lifted back to the original continuous-prior problem \eqref{eq:best-response} with only an additional loss of $L\delta$.

To obtain a final additive error of at most $\bar\epsilon>0$, choose the index-discretization parameter and the value-grid size so that $\epsilon=\bar\epsilon/2$ and $\delta\le \bar\epsilon/4L$.
With these choices, we have $\epsilon+2L\delta\le \bar\epsilon$,
and hence the constructed strategy $\widetilde G$ in above steps satisfies $U(\widetilde G)\ge \mathsf{OPT}-\bar\epsilon$.

The number of value grid points is $M=O(1/\delta)$, and the number of index grid points satisfies $R\le 2M+\left\lceil{1/\epsilon}\right\rceil+3$.
Therefore, the finite LP has $MR$ variables and $M+R$ equality constraints, both polynomial in $L$ and $1/\bar\epsilon$.
Assuming that the interval masses $F(I_i)$ can be computed in polynomial time, the whole procedure runs in polynomial time.

Therefore, for every target accuracy $\bar\epsilon>0$, the procedure outputs in polynomial time a feasible solution $\widetilde G\in\mathcal G(F,c)$ such that $U(\widetilde G)\ge \mathsf{OPT}-\bar\epsilon$.
This proves that the sender's persuasion problem admits a fully polynomial additive approximation scheme.
\end{proof}

\section{Equilibrium Existence of the Competitive Information Design Game}
\label{apx:existence}

In this section, we establish the existence of an equilibrium in the game among the senders, under some mild assumptions.
\begin{theorem}[Equilibrium Existence]
\label{thm:existence}
    There exists an equilibrium in the game among the senders if each sender $i$'s prior distribution $F_i$ is atomless over $[0,1]$.
\end{theorem}

We proceed with the proof in the following three steps.
Briefly speaking, in Step $1$, we construct a finite game called $m$-discrete approximation game for each granularity $m\in\Zplus$, and directly apply Nash's Theorem to show the existence of a discrete equilibrium.
In Step $2$, we show that the sequence of discrete equilibria has a subsequence that weakly converges to some distribution profile, where any of the distributions forms a feasible action for each sender.
In Step $3$, we prove that the limit profile indeed constitutes an equilibrium of the original game through establishing the convergence of utilities along with the convergence of equilibrium strategies.
For ease of presentation, we assume all senders share a common inspection cost $c$, that is $c_1=\cdots=c_N=c$.

\xhdr{Step~1: Constructing $m$-discrete approximation games}
By \Cref{lem:convex-compact}, each sender $i$'s strategy space $\mathcal{G}(F_i,c)$ is a compact and convex set containing all feasible 2-D distributions.
By Krein-Milman Theorem, any convex and compact set is the closed convex hull of its extreme points;
moreover, each element within the set can be represented as a convex combination of its extreme points.
Thus, each sender's strategy space can be precisely characterized by its extreme points.
Recall from \Cref{def:joint-distribution} that the space $\mathcal{G}(F_i,c)$ is subject to an infinite number of constraints; the convex space $\mathcal{G}(F_i,c)$ may therefore have an infinite number of extreme points.
Therefore, even if one takes the set of these extreme points as the action set, Nash's theorem does not apply.
To address this problem, we construct finite games that not only discretize the support of the value and the index, but also modify constraints in \Cref{def:joint-distribution}.
Each finite game is parameterized by an integer $m \in \mathbb Z^+$.

\begin{definition}[$m$-Discrete Approximation Game]
\label{defn:m discrete game}
    Fix any $m\in\mathbb{Z}^+$,
    let $V^m \triangleq \{0, \frac{1}{2^m}, \cdots, 1\}$ and $\Theta^m \triangleq \{-c, \frac{1}{2^m}-c, \ldots, 1-c\}$
    be the discretized support for the value and the index, respectively.
    The \emph{$m$-discrete approximation game} is as follows:
    \begin{itemize}
        \item \textbf{Strategy space: }For each sender $i$,
        the strategy space $\mathcal{S}_i^m$ is a subset of distributions on $V^m \times \Theta^m$.
        Each $\mathbf p \in \mathcal S_i^m$, with $p_{i,j}$ denoting the probability on the point $(\frac{i}{2^m},\frac{j}{2^m}-c)$, satisfies
        \begin{align}
            \sum_{i=0}^{2^m}p_{i,j}\cdot\left[\frac{i}{2^m}-\frac{j}{2^m}+c\right]_+ = c\cdot \sum_{i=0}^{2^m}p_{i,j}~,\quad &\forall j\in\{0\}\cup[2^m]~; \label{eq:discrete-feasible-constraint-1}\\
            \sum_{j=0}^{2^m}p_{i,j}=F_i\left(\frac{i}{2^m} \right)-F_i\left(\left[\frac{i-1}{2^m}\right]_+\right)~,\quad &\forall i\in\{0\}\cup[2^m]~. \label{eq:discrete-feasible-constraint-2}
        \end{align}
        \item \textbf{Action space: }Each sender $i$'s action space, denoted by $\mathcal{A}_i^m$, is defined as the set of all extreme points of $\mathcal{S}_i^m$.  
        \item \textbf{Utility: }
        Given a strategy profile $(G_1^m,...,G_N^m)$ where each $G_i^m\in\mathcal{A}_i^m$, the receiver receives $(v_i^m, \theta_i^m)\sim G_i^m$ for each~$i$, and chooses a sender whose amortized value $\kappa_i^m \triangleq \min(v_i^m, \theta_i^m)$ is non-negative and the largest.  
        If all amortized values are negative, the receiver chooses no one.
        If there is a tie among multiple senders, break the tie uniformly at random.\footnote{Here we adopt the uniformly random tie-breaking rule for convenience of presentation.  As we remarked above, the proof goes through for a host of other tie-breaking rules.}
        A sender gets utility $1$ if chosen by the receiver, and $0$ otherwise.
    \end{itemize}
\end{definition}

\begin{lemma}
\label{lem:discrete-game-strategy-action-space}
For any $m\in\mathbb{Z}^+$, each sender $i$'s strategy space $\mathcal{S}_i^m$ is non-empty, convex, and compact.
Each sender $i$'s action space $\mathcal{A}_i^m$ is finite. 
\end{lemma}

\begin{proof}[Proof of \Cref{lem:discrete-game-strategy-action-space}]
First, we prove that the space $\mathcal{S}_i^m$ is non-empty.
For any $m\in\mathbb{Z}^+$ and each sender $i$, the full-revelation strategy belongs to the space $\mathcal{S}_i^m$, which makes the space $\mathcal{S}_i^m$ non-empty.

Second, we prove the space $\mathcal{S}_i^m$ is convex.
For any $m\in\mathbb{Z}^+$, consider any pair of feasible strategies $G^m,\ H^m\in \mathcal{S}_i^m$ (with probability mass $g_{i,j}$ and $h_{i,j}$ on the point $(\frac{i}{2^m},\frac{j}{2^m}-c)$ respectively) and any $\lambda\in[0,1]$, 
we construct the convex combination $T^m$ where each mass $t_{i,j}=\lambda\cdot g_{i,j}+(1-\lambda)\cdot h_{i,j}$.
It is obvious that for each $j\in\{0\}\cup[2^m]$,
\begin{align*}
    &\sum_{i=0}^{2^m}t_{i,j}\cdot\max\left\{\frac{i}{2^m}-\frac{j}{2^m}+c,0\right\}\\
    =\,&\lambda\cdot \sum_{i=0}^{2^m}g_{i,j}\cdot\max\left\{\frac{i}{2^m}-\frac{j}{2^m}+c,0\right\}+(1-\lambda)\cdot\sum_{i=0}^{2^m}h_{i,j}\cdot\max\left\{\frac{i}{2^m}-\frac{j}{2^m}+c,0\right\}\\
    =\,&\lambda \cdot c\cdot \sum_{i=0}^{2^m}g_{i,j}+(1-\lambda)\cdot c\cdot \sum_{i=0}^{2^m}h_{i,j}\,=\,c\cdot \sum_{i=0}^{2^m}t_{i,j}~,
\end{align*}
and for each $i\in\{0\}\cup[2^m]$, 
\begin{align*}
    &\sum_{i=0}^{2^m}t_{i,j}
    =\lambda\cdot\sum_{j=0}^{2^m}g_{i,j}+(1-\lambda)\cdot\sum_{j=0}^{2^m}h_{i,j}=F\left(\frac{i}{2^m}\right)-F\left(\max\left\{0,\frac{i-1}{2^m}\right\}\right)~.
\end{align*}
So we know that the convex combination $T^m$ also belongs to the space $\mathcal{S}_i^m$, which makes the space $\mathcal{S}_i^m$ convex.

Third, we prove that the space $\mathcal{S}_i^m$ is compact.
Since the space $\mathcal{S}_i^m$ is indeed a measure space over $V^m\times \Theta^m$, the space $\mathcal{S}_i^m$ is bounded.
For any $m\in\mathbb{Z}^+$ and any sender $i$, we assume that sequence $\{G^k\}_{k\in\mathbb{Z}^+}$ converges to some discrete distribution $G$ where each $G^k\in \mathcal{S}_i^m$.
So the sequence $\{p_{i,j}^k\}_{k\in\mathbb{Z}^+}$ converges to $p_{i,j}$ where $p_{i,j}^k$ denotes the probability measure of strategy $G^k$ on the point $(\frac{i}{2^m},\frac{j}{2^m}-c)$ for any $i,j\in\{0\}\cup [2^m]$, and $p_{i,j}$ denotes the probability measure of strategy $G$ on the same point.
This implies that there exists $K>0$ such that for any $k>K$ and any $\epsilon>0$, we have $|p_{i,j}^k-p_{i,j}|<\epsilon$ for any $i,j\in\{0\}\cup [2^m]$.
If we assume that there exist some $m$ and $j\in\{0\}\cup[2^m]$ such that 
$
    \sum_{i=0}^{2^m}p_{i,j}\cdot\max\left\{\frac{i}{2^m}-\frac{j}{2^m}+c,0\right\}=d\cdot \sum_{i=0}^{2^m}p_{i,j}\neq c\cdot \sum_{i=0}^{2^m}p_{i,j}
$.
There exists $K_1>0$ such that for any $k>K_1$, we have
\begin{equation*}
    \Bigg|\sum_{i=0}^{2^m}p_{i,j}^k\cdot\max\left\{\frac{i}{2^m}-\frac{j}{2^m}+c,0\right\}-\sum_{i=0}^{2^m}p_{i,j}\cdot\max\left\{\frac{i}{2^m}-\frac{j}{2^m}+c,0\right\}\Bigg|<\frac{1}{3}\Bigg|\,d\cdot \sum_{i=0}^{2^m}p_{i,j}- c\cdot \sum_{i=0}^{2^m}p_{i,j}\,\Bigg|~.
\end{equation*}
There exists $K_2>0$ such that for any $k>K_2$, we have
\begin{equation*}
    \Bigg|\,c\cdot \sum_{i=0}^{2^m}p_{i,j}^k-c\cdot \sum_{i=0}^{2^m}p_{i,j}\,\Bigg|<\frac{1}{3}\Bigg|\,d\cdot \sum_{i=0}^{2^m}p_{i,j}- c\cdot \sum_{i=0}^{2^m}p_{i,j}\,\Bigg|~.
\end{equation*}
Thus, we know for any $k>\max\{K_1,K_1\}$, we have
\begin{equation*}
    \Bigg|\sum_{i=0}^{2^m}p_{i,j}^k\cdot\max\left\{\frac{i}{2^m}-\frac{j}{2^m}+c,0\right\}- c\cdot \sum_{i=0}^{2^m}p_{i,j}^k\,\Bigg|>\frac{1}{3}\Bigg|\,d\cdot \sum_{i=0}^{2^m}p_{i,j}- c\cdot \sum_{i=0}^{2^m}p_{i,j}\,\Bigg|~,
\end{equation*}
which forms a contradiction with the fact that $G^k\in \mathcal{S}_i^m$.
Furthermore, if we assume that there exist some $m$ and $j\in\{0\}\cup[2^m]$ such that 
$
    \sum_{j=0}^{2^m}p_{i,j}\neq F\left(\frac{i}{2^m}\right)-F\left(\max\left\{0,\frac{i-1}{2^m}\right\}\right)
$,
then we can achieve a contradiction in the same manner.
So we have proved that the space $\mathcal{S}_i^m$ is closed.
Then we can directly make space $\mathcal{S}_i^m$ compact through the Heine-Borel Theorem.

Last, we prove that space $\mathcal{A}_i^m$ is finite.
By \Cref{defn:m discrete game}, set $\mathcal{S}_i^m$ is formed by $O(2^m)$ linear constraints, 
which implies that there are a finite number of extreme points of set $\mathcal{S}_i^m$. 
Thus, the action space $\mathcal{A}_i^m$ is finite.
\end{proof}

\xhdr{Step 2: Showing the limit of equilibrium strategies is feasible}
By \Cref{lem:discrete-game-strategy-action-space}, each $m$-discrete approximation game is finite, there being a finite number of senders, and each sender $i$ having a finite action space.
Thus, Nash's Theorem applies, 
and there is an equilibrium $(\tilde{G}_1^m,...,\tilde{G}_N^m)$,
where for each sender $i$, $\tilde{G}_i^m\in \mathcal{S}_i^m$ is a mixed equilibrium strategy.
For each sender $i$, these equilibrium strategies form a sequence of 2-D distributions: $\{\tilde{G}_i^m\}_{m\in \mathbb{Z}^+}$.
We want to establish the convergence of this sequence, so we first need to introduce the following Helly's Selection Theorem.
\begin{lemma}[Helly's Selection Theorem]
\label{lem:helly}
Let $\{G^m\}_{m\in \mathbb{Z}^+}$ be a sequence of CDFs which is tight,\footnote{Take the one-dimensional case as an example, a sequence of functions $\{G^m\}_{m\in \mathbb{Z}^+}$ is tight, if and only if $\forall \epsilon>0$ there exists an interval $[a,b]$ such that for each $m\in\mathbb{Z}^+$ we have $G^m(b)-G^m(a)>1-\epsilon$.} 
then there exists a subsequence $\{m(k)\}_{k\in \mathbb{Z}^+}\subseteq \mathbb{Z}^+$ such that $\{G^{m(k)}\}_{k\in \mathbb{Z}^+}$ weakly converges to a certain CDF $G$.\footnote{Take the one-dimensional case as an example, sequence $\{G^m\}_{m\in\mathbb{Z}^+}$ weakly converges to CDF $G$ if and only if $\lim\limits_{m\rightarrow \infty}G^{m}(x)=G(x)$ for each point $x$ at which $G$ is continuous.}
\end{lemma}
By Helly's Selection Theorem, there is a subsequence that weakly converges to a 2-D distribution; we further show that this limit is a feasible for each sender~$i$.

\begin{lemma}
\label{lem:converge-to-feasible-strategy}
There exists a subsequence $\{m(k)\}_{k\in\mathbb{Z}^+}\subseteq\mathbb{Z}^+$ such that the sequence $\{(\tilde{G}_i^{m(k)})\}_{k\in\mathbb{Z}^+}$ weakly converges to a certain 2-D distribution $\tilde{G}_i$.
Furthermore, $\tilde G_i$ is a feasible strategy for sender $i$ in the original game.
\end{lemma}

\begin{proof}[Proof of \Cref{lem:converge-to-feasible-strategy}]
For any $m\in\mathbb{Z}^+$, space $\mathcal{S}_i^m$ is a measure space over $[0,1]\times[-c,1-c]$, which makes space $\mathcal{S}_i^m$ tight.
By \Cref{lem:helly}, for each sender $i\in[N]$, there exists a subsequence $\{m_i(k)\}_{k\in\mathbb{Z}^+}\subseteq\mathbb{Z}^+$ such that $\{\tilde{G}_i^{m_i(k)}\}_{k\in\mathbb{Z}^+}$ weakly converges to a certain 2-d CDF $\tilde{G}_i$.
This means that, for any continuous point $(\hat{v},\hat{\theta})$ of distribution $\tilde{G}_i$, the sequence $\{\tilde{G}_i^{m_i(k)}(\hat{v},\hat{\theta})\}_{k\in\Zplus}$ converges to $\tilde{G}_i(\hat{v},\hat{\theta})$.
Next, it can be shown that, there exists a common subsequence $\{m(k)\}_{k\in \mathbb{Z}^+}$ such that for each sender $i$, the sequence $\{\tilde{G}^{m(k)}_{i}\}_{k\in \mathbb{Z}^+}$ weakly converges to a certain CDF $\tilde{G}_i$.  
To see this, we first find a subsequence $\{m_1(k)\}_{k\in \Zplus}$ of $\Zplus$ such that $\{\tilde{G}_1^{m_1(k)}\}_{k\in \Zplus}$ weakly converges to $\tilde{G}_1$. We then find a subsequence $\{m_2(k)\}_{k\in \Zplus}$ of $\{m_1(k)\}_{k\in \Zplus}$ such that $\{\tilde{G}_2^{m_2(k)}\}_{k\in \Zplus}$ weakly converges to $\tilde{G}_2$.  It is known that if a sequence converges, any subsequence of it also converges to the same limit. Hence, $\{\tilde{G}_1^{m_2(k)}\}_{k\in \Zplus}$ also weakly converges to $\tilde{G}_1$. In the same manner, we obtain a common subsequence $\{m(k)\}_{k\in \Zplus}$ of $\Zplus$ such that for each sender $i$, the sequence $\{\tilde{G}^{m(k)}_{i}\}_{k\in \Zplus}$ weakly converges to a certain CDF $\tilde{G}_i$. 

Then, we are going to show that the limit 2-d distribution $\tilde{G}_i$ is a feasible strategy of sender $i$ in the original game.
For convenience, for each sender $i$, we let $\{\tilde{G}^{k}\}_{k\in\Zplus}$ denote the sequence of discrete equilibrium strategies, and let $\tilde{G}$ denote the limit.
We divide this proof into two steps.

First, we prove that $\mathbb{E}_{v\sim \tilde{G}_{\cdot|\theta}}[\max\{v-\theta,0\}]=c$ for any $\theta\in[-c,1-c]$.
Let's consider any index $\theta\in(-c,1-c)$ since the cases of index $-c$ and index $1-c$ are trivial.
For any $\hat{\theta}\in\supp(\tilde{G}_\theta)$, 
we assume that 
\begin{equation*}
    \int_{v=0}^1 |c-\max\{v-\theta,0\}|\ \tilde{g}(v,\hat{\theta})\dd v=d>0~.
\end{equation*}
By the weak convergence of sequence $\{\tilde{G}^{k}(\hat{v},\hat{\theta})\}_{k\in\Zplus}$ to distribution $\tilde{G}$, there exists $K_1>0$ such that when $k>K_1$, we have
\begin{equation*}
    \Bigg|\,\int_{v=0}^1 |c-\max\{v-\theta,0\}|\ \tilde{g}(v,\hat{\theta})\dd v-\int_{v=0}^1 |c-\max\{v-\theta,0\}|\ \tilde{g}^k(v,\hat{\theta})\dd v\,\Bigg|<\frac{d}{2}~.
\end{equation*}
This implies that for any $k>K_1$, we have
\begin{equation*}
    \int_{v=0}^1 |c-\max\{v-\theta,0\}|\ \tilde{g}^k(v,\hat{\theta})\dd v\in\left(\frac{d}{2},\frac{3d}{2}\right)~,
\end{equation*}
which forms a contradiction with the fact $\int_{v=0}^1 |c-\max\{v-\theta,0\}|\ \tilde{g}^k(v,\hat{\theta})\dd v=0$, and makes this assumption invalid.

Second, we prove that $\int_{\theta=-c}^{1-c}\dd \tilde{G}(v,\theta)=f(v)$ for any $v\in[0,1]$.
Let's consider any value $v\in(0,1]$, since the case of value $0$ is trivial.
For any $\hat{v}=\alpha\cdot 2^{-\beta}$ for some $\beta\in\mathbb{Z}^+$ and some $\alpha\in[2^\beta]$, then for any $k>\beta$,
there exists $t\in\{0\}\cup[2^k]$ such that $\hat{v}=\frac{t}{2^k}$ and $G^k(\frac{t}{2^k},1-c)-G^k(\frac{t-1}{2^k},1-c)=F(\frac{t}{2^k})-F(\frac{t-1}{2^k})$.
Dividing both sides of the equation by $2^{-k}$, we get
\begin{equation*}
    \frac{G^k(\frac{t}{2^k},1-c)-G^k(\frac{t-1}{2^k},1-c)}{2^{-k}}=\frac{F(\frac{t}{2^k})-F(\frac{t-1}{2^k})}{2^{-k}}~.
\end{equation*}
When $k\rightarrow\infty$, we know that $\int_{\theta=-c}^{1-c}\dd G(\hat{v},\theta)=f(\hat{v})$.
For any $\hat{v}\neq\alpha\cdot 2^{-\beta}$ for any $\beta\in\mathbb{Z}^+$ and any $\alpha\in[2^\beta]$, then for any $k\in\mathbb{Z}^+$ there exist $t\in[2^m]$ such that $\hat{v}\in(\frac{t-1}{2^k},\frac{t}{2^k})$ and $G^k(\frac{t}{2^k},1-c)-G^k(\hat{v},1-c)=F(\frac{t}{2^k})-F(\frac{t-1}{2^k})$.
Dividing both sides of the equation by $2^{-k}$,
\begin{equation*}
    \frac{G^k(\frac{t}{2^k},1-c)-G^k(\hat{v},1-c)}{2^{-k}}=\frac{F(\frac{t}{2^k})-F(\frac{t-1}{2^k})}{2^{-k}}~.
\end{equation*}
When $k\rightarrow\infty$, we know that $\int_{-c}^{1-c}\dd G(\hat{v},\theta)=f(\hat{v})$.

Combining these two steps, we have proved that each sender $i$'s strategy in the limit, $\tilde{G}_i$, is indeed a feasible strategy for the original game.
\end{proof}

\xhdr{Step 3: Proving the limit strategy profile $(\tilde{G}_1,...,\tilde{G}_N)$ is an equilibrium in the original game} 
To this end, we need to show that each sender's utility converges in the subsequence along with the convergence of discrete equilibrium strategies.
An obstacle to this is that the senders' utilities may be discontinuous in their strategies if there are ties in $(\tilde{G}_1,...,\tilde{G}_N)$ with strictly positive probability.
To rule out this case, we first 
prove that, in the subsequence of approximation games, the probability of ties in any neighborhood diminishes.
We only need to show this for amortized values that actually affect the senders' utilities (above the smallest winning amortized value defined below).  

\begin{definition}[Smallest Winning Amortized Value]
\label{defn:smallest-winning-kappa}
For any profile of amortized value distributions $(K_1,...,K_N)$, the \emph{smallest winning amortized value} is
$\underline{\kappa}\triangleq\max_{i\in[N]}\inf\supp(K_{i})$.
\end{definition}

Each sender has zero utility for realizing an amortized value below the smallest winning amortized value $\underline{\kappa}$.
If sender~$i$'s strategy has a mass in $K_i$ below $\underline \kappa$, it does not cause discontinuity in anyone's utility.
Therefore we need only to focus on the part of the distribution at or above $\underline{\kappa}$.
{For the limit strategy profile $(\tilde{G}_1,...,\tilde{G}_N)$, denote each sender $i$'s distribution of amortized value as $\tilde{K}_i$, and the smallest winning amortized value as $\underline{\kappa}$.}
\begin{lemma}
\label{lem:no-tie-in-limit}
There exist no amortized value $\hat{\kappa}\in[\underline{\kappa},1-c]$ and two distinct senders $i,j\in[N]$ such that $\tilde{K}_i,\tilde{K}_j$ both assign positive probabilities at $\hat{\kappa}$.
\end{lemma}

\begin{proof}[Proof of \Cref{lem:no-tie-in-limit}]
By the assumption of the prior, there exists no tie at index $1-c$ in any feasible strategy profile.
We assume that there exists an index $\hat{\theta}\in[\underline{\theta},1-c)$ such that sender 1 and 2's strategies, $\tilde{G}_1,\tilde{G}_2$,
simultaneously assign a positive mass at index $\hat{\theta}$.
Let $\phi_1(\hat{\theta})-\phi_1(\hat{\theta}^-)=p_1>0$ and $\phi_2(\hat{\theta})-\phi_2(\hat{\theta}^-)=p_2>0$.
We want to prove that discrete strategies, $\tilde{G}_1^m$ and $\tilde{G}_2^m$, also assign a big enough probability at index $\hat{\theta}$ for sufficiently large $m$, which violates the equilibrium conditions. 
By \Cref{lem:converge-to-feasible-strategy}, strategy $\tilde{G}_1$ must be a multi-value row at row $\hat{\theta}$, that is $\inf\supp(\tilde{G}_{1,\cdot|\hat{\theta}})<\sup\supp(\tilde{G}_{1,\cdot|\hat{\theta}})$. 
In the same manner, we have $\inf\supp(\tilde{G}_{2,\cdot|\hat{\theta}})<\sup\supp(\tilde{G}_{2,\cdot|\hat{\theta}})$.
First, based on the definition of $u$, for each sender $i\in[N]$, we define the discrete version of the compressed interim utility as
\begin{equation*}
    \phi_i^m(x)=\prod_{j\neq i}\left(\tilde{G}_j^m(1,x) + \tilde{G}_j^m(x,1-c)-\tilde{G}_j^m(x,x)\right)~,\quad\forall x\in[-c,1-c]~.
\end{equation*}
Based on whether the index $\hat{\theta}$ belongs to the discrete support of index $\Theta^m$ for some $m\in\Zplus$, we divide this part of proof into two cases.

\xhdr{Case 1: There exist $\beta\in\mathbb{Z}^+$ and $\alpha\in[2^\beta]$ such that $\hat{\theta}=\alpha\cdot 2^{-\beta}-c$}
We take sender $1$ as an example, and the proof of sender $2$ follows the same idea.
There exists $m_1>0$ such that when $m>m_1$, $(\hat{\theta}-8\cdot 2^{-m},\hat{\theta}+8\cdot 2^{-m})\subset(\underline{\theta},1-c)$ and spreading the row $\hat{\theta}$ of the discrete equilibrium strategy to any row $\theta'\in(\hat{\theta},\hat{\theta}+8\cdot 2^{-m})$ and any row $\theta''\in(\hat{\theta}-8\cdot 2^{-m},\hat{\theta})$ is possible, and will not cause any value of row $\theta$ exceed the line $v=\theta$.
We define $\epsilon_1=\frac{p_1}{100}$,
then there exists $m_2>0$ such that when $m>m_2$, it holds that $|\phi_1^m(\hat{\theta}^+)-\phi_1^m(\hat{\theta}^-)-p_1|<\epsilon_1$.
There exists $m_3>0$ such that when $m>m_3$, it holds that $\hat{\theta}-7\cdot 2^{-m}>\underline{\theta}$ and $\phi_1^m(\hat{\theta}^-)-\phi_1^m(\hat{\theta}-7\cdot 2^{-m})<\epsilon_2$ where $\epsilon_2=\frac{1}{100}p_1$.
There exists $m_4>0$ such that when $m>m_4$, it holds that 
$\hat{\theta}+2^{-m}<1-c$ and $\phi_1^m(\hat{\theta}+2^{-m})-\phi_1^m(\hat{\theta}^+)<\epsilon_3$ where $\epsilon_3=\frac{1}{100}p_1$.
Based on these inequalities, we achieve that $3\epsilon_1+\epsilon_2<3p_1$, which implies that
\begin{equation*}
    \frac{7}{8}(p_1-\epsilon_1+\epsilon_2+\epsilon_3)>\epsilon_2+\frac{p_1-\epsilon_1}{2}~,\quad
    \frac{7}{8}(p_1+\epsilon_1+\epsilon_2+\epsilon_3)>\epsilon_2+\frac{p_1+\epsilon_1}{2}~.
\end{equation*}
This implies that, for any $m>\max\{m_1,m_2,m_3,m_4\}$, spreading the row $\hat{\theta}$ of discrete equilibrium strategy $G_1^m$ to row $\hat{\theta}-7\cdot 2^{-m}$ and row $\hat{\theta}+2^{-m}$ can bring a strictly positive utility increase, which violates the equilibrium conditions.

\xhdr{Case 2: $\hat{\theta}\neq\alpha\cdot 2^{-\beta}-c$ for any $\beta\in\mathbb{Z}^+$ and any $\alpha\in[2^\beta]$}
We take sender $1$ as an example, and the proof of sender $2$ follows the same idea.
For any $m\in\mathbb{Z}^+$, there exists $k\in\{0\}\cup[2^m]$ such that $\hat{\theta}\in(\frac{k}{2^m},\frac{k+1}{2^m})$.
There exists $m_1>0$ such that when $m>m_1$, it holds that $(k\cdot 2^{-m}-8\cdot 2^{-m},(k+1)\cdot 2^{-m}+8\cdot 2^{-m})\subset(\underline{\theta},1-c)$,
and spreading the row $k\cdot 2^{-m}$ of the discrete equilibrium strategy to any row $\theta'\in((k+1)\cdot 2^{-m},(k+1)\cdot 2^{-m}+8\cdot 2^{-m})$ and any row $\theta''\in(k\cdot 2^{-m}-8\cdot 2^{-m},k\cdot 2^{-m})$ is possible, and will not cause any value of row $k\cdot 2^{-m}$ exceed the line $v=\theta$.
There exists $m_2>0$ such that when $m>m_2$, it holds that $|\phi_1^m(\left((k+1)\cdot 2^{-m}\right)^+)-\phi_1^m((k\cdot 2^{-m})^-)-p_1|<\epsilon_1$ where $\epsilon_1=\frac{p_1}{100}$.
There exists $m_3>0$ such that when $m>m_3$, it holds that $k\cdot 2^{-m}-5\cdot 2^{-m}>\underline{\theta}$ and $\phi_1^m((k\cdot 2^{-m})^-)-\phi_1^m(k\cdot 2^{-m}-5\cdot 2^{-m})<\epsilon_2$ where $\epsilon_2=\frac{p_1}{100}$.
There exists $m_4>0$ such that when $m>m_4$, it holds that 
$(k+1)\cdot 2^{-m}+2^{-m}<1-c$ and 
$\phi_1^m((k+1)\cdot 2^{-m}+2^{-m})-\phi_1^m(((k+1)\cdot 2^{-m})^+)<\epsilon_3$ where $\epsilon_3=\frac{p_1}{100}$.
Based on these inequalities, we have $\epsilon_1+3\epsilon_2<p_1$, which implies that
\begin{equation*}
    \frac{5}{8}(p_1-\epsilon_1+\epsilon_2+\epsilon_3)>\epsilon_2+\frac{p_1-\epsilon_1}{2}~,\quad\frac{5}{8}(p_1+\epsilon_1+\epsilon_2+\epsilon_3)>\epsilon_2+\frac{p_1+\epsilon_1}{2}~.
\end{equation*}
This implies that, for any $m>\max\{m_1,m_2,m_3,m_4\}$, spreading the row $k\cdot 2^{-m}$ of discrete equilibrium strategy $G_1^m$ to row $k\cdot 2^{-m}-5\cdot 2^{-m}$ and row $(k+1)\cdot 2^{-m}+2^{-m}$ can bring a strictly positive utility increase, which violates the equilibrium conditions.

In summary, given the assumption, the converging discrete strategy profile cannot form an equilibrium, which make this assumption invalid.
Thus, we have proved that there is no tie in the limit strategy profile $(\tilde{G}_1,...,\tilde{G}_N)$ at and above $\underline{\theta}$.
\end{proof}

With the obstacle of discontinuity cleared, we obtain the convergence of utility functions along with the convergence of discrete equilibrium strategies.
Here we comes the final step of the proof.

\begin{lemma}
\label{lem:limit-profile-is-equilibrium}
The limit profile $(\tilde{G}_1,...,\tilde{G}_N)$ is indeed an equilibrium in the original game.
\end{lemma}

\begin{proof}[Proof of \Cref{lem:limit-profile-is-equilibrium}]
Consider each sender $i\in[N]$ and any feasible strategy $G_i\in \mathcal{G}_i(F_i,c)$, 
we want to show that strategy $\tilde{G}_i$ achieves a higher expected utility than strategy $G_i$, given others strategies $\tilde{G}_{-i}$.
We divide this part of the proof into two cases.

\xhdr{Case 1: There exist ties over $[\underline{\theta},1-c]$ in strategy profile $(G_i,\tilde{G}_{-i})$}
We assume that there exists a tie at index $\hat{\theta}$ in strategy profile $(G_i,\tilde{G}_{-i})$.
By the assumption of the prior, strategy $G_i$ must be a multi-value row at row $\hat{\theta}$, that is $\inf\supp(G_{i,\cdot|\hat{\theta}})<\sup\supp(G_{i,\cdot|\hat{\theta}})$, which implies that strategy $G_i$ can spread the probability of index $\hat{\theta}$ to index $\hat{\theta}+\epsilon$ and index $\hat{\theta}-\epsilon$ for sufficiently small $\epsilon>0$.
By the fact that $\phi_i(\hat{\theta}^-)<\phi_i(\hat{\theta})<\phi_i(\hat{\theta}^+)$, we know spreading the probability of strategy 
$G_i$ at index $\hat{\theta}$ to index $\hat{\theta}+\epsilon$ and index $\hat{\theta}-\epsilon$ achieves a utility increase for sufficiently small $\epsilon>0$.
Therefore, we know that, given others' strategies, the best response strategy of sender $i$ will not form a tie with the strategies of others within the interval $[\underline{\theta},1-c]$. 
Thus, in Case 2, we only need to prove that for those strategies that do not form a tie with the others' strategies within $[\underline{\theta},1-c]$, strategy $\tilde{G}_i$ will achieve a higher expected utility. 
This implies that strategy $\tilde{G}_i$ is the best response to others' strategies.

\xhdr{Case 2: There is no tie over $[\underline{\theta},1-c]$ in strategy profile $(G_i,\tilde{G}_{-i})$} 
Let $\tilde{u}_i$ denote sender $i$'s interim utility function given others' strategies $\tilde{G}_{-i}$.
Let $\tilde{u}_i^m$ denote sender $i$'s interim utility function given others' strategies $\tilde{G}_{-i}^m$.
By \Cref{lem:no-tie-in-limit}, there is no tie over $[\underline{\theta},1-c]$ in the limit strategy profile 
$(\tilde{G}_1, \dots, \tilde{G}_N)$. 
By the definition of weak convergence, 
we have that
\begin{equation}
\label{eq:proof-of-existence-1}
    \lim\limits_{m\to \infty}\int_{[0,1]\times[-c,1-c]}\tilde{u}_i^m(v,\theta)\ \tilde{g}^m_i(v,\theta)\,\dd v\,\dd \theta=\int_{[0,1]\times[-c,1-c]}\tilde{u}_i(v,\theta)\ \tilde{g}_i(v,\theta)\,\dd v\,\dd \theta~.
\end{equation}
For each sender $i\in[N]$ and any feasible strategy $G_i\in \mathcal{G}_i(F_i,c)$ that do not form a tie with others' strategies within $[\underline{\theta},1-c]$, we construct a 2-dimensional CDFs sequence $\{G_i^m\}_{m\in \mathbb{Z}^+}$ where each $G_i^m \in \mathcal{S}_i^m$, and the sequence weakly converges to strategy $G_i$. 
Specifically, for $\forall m\in \mathbb{Z}^+$, we define $G_i^m$ as below.
\begin{equation*}
    G_i^m(v,\theta) =
    \begin{cases}
    G_i(v,\theta)~, & \text{if } (v,\theta) \in V^m\times \Theta^m~,\\
    G_i(\max\{t:t\in V^m,\ t\le v\},\max\{t:t\in \Theta^m,\ t\le \theta\})~, & \text{otherwise}~.\\
    \end{cases}
\end{equation*}

First, we show that the sequence $\{G_i^m\}_{m\in \mathbb{Z}^+}$ weakly converges to strategy $G_i$.
We consider any pair of $(v,\theta)\in[0,1]\times[-c,1-c]$.
If there exist $\alpha_1\in\Zplus$ and $\beta_1\in[2^{-\alpha_1}]$ such that $v=\beta_1\cdot 2^{-\alpha_1}$ and $\alpha_2\in\Zplus$ and $\beta_2\in[2^{-\alpha_2}]$ such that $\theta=\beta_2\cdot 2^{-\alpha_2}$, then when $m>\max\{\alpha_1,\alpha_2\}$, we have $G_i^m(v,\theta)=G_i(v,\theta)$ which implies that sequence $\{G_i^m(v,\theta)\}_{m\in \mathbb{Z}^+}$ converges to $G_i(v,\theta)$. 
If $v\neq \beta_1\cdot 2^{-\alpha_1}$ for any $\alpha_1\in\Zplus,\ \beta_1\in[2^{-\alpha_1}]$ or $\theta\neq \beta_2\cdot 2^{-\alpha_2}$ for any $\alpha_2\in\Zplus,\ \beta_2\in[2^{-\alpha_2}]$, then for any $m\in\Zplus$, we have $G_i^m(v,\theta)=G_i(\max\{t:t\in V^m,\ t\le v\},\max\{t:t\in \Theta^m,\ t\le \theta\})$.
When $m$ goes to infinity, sequence $\{\max\{t:t\in V^m,\ t\le v\}\}_{m\in\Zplus}$ converges to $v$, and sequence $\{\max\{t:t\in \Theta^m,\ t\le \theta\}\}_{m\in\Zplus}$ converges to $\theta$.
So we have sequence $\{G_i^m(v,\theta)\}_{m\in \mathbb{Z}^+}$ also converges to $G_i(v,\theta)$.
By the fact that sequence $\{G_i^m\}_{m\in \mathbb{Z}^+}$ weakly converges to strategy $G_i$, using a similar method of \Cref{lem:converge-to-feasible-strategy}, we can also prove that there exists $m_1>0$ such that when $m<m_1$, $G_i^m\in\mathcal{S}_i^m$.

By \Cref{eq:proof-of-existence-1}, we have
\begin{align*}
    &\int_{[0,1]\times[-c,1-c]}\tilde{u}_i(v,\theta)\ g_i(v,\theta)\,\dd v\,\dd \theta-\int_{[0,1]\times[-c,1-c]}\tilde{u}_i(v,\theta)\ \tilde{g}_i(v,\theta)\,\dd v\,\dd \theta\\
    =&\lim\limits_{m\rightarrow\infty}\left(\int_{[0,1]\times[-c,1-c]}\tilde{u}_i(v,\theta)\ g_i(v,\theta)\,\dd v\,\dd \theta-\int_{[0,1]\times[-c,1-c]}\tilde{u}_i^m(v,\theta)\ g^m_i(v,\theta)\,\dd v\,\dd \theta\right)\\
    +&\lim\limits_{m\rightarrow\infty}\left(\int_{[0,1]\times[-c,1-c]}\tilde{u}_i^m(v,\theta)\ g^m_i(v,\theta)\,\dd v\,\dd \theta-\int_{[0,1]\times[-c,1-c]}\tilde{u}_i^m(v,\theta)\ \tilde{g}^m_i(v,\theta)\,\dd v\,\dd \theta\right)\\
    =&\lim\limits_{m\rightarrow\infty}\left(\int_{[0,1]\times[-c,1-c]}G_i^m(v,\theta)\dd \tilde{u}_i^m(v,\theta)-\int_{[0,1]\times[-c,1-c]}G_i(v,\theta)\dd \tilde{u}_i(v,\theta)\right)\\
    +&\lim\limits_{m\rightarrow\infty}\left(\int_{[0,1]\times[-c,1-c]}\tilde{u}_i^m(v,\theta)\ g^m_i(v,\theta)\,\dd v\,\dd \theta-\int_{[0,1]\times[-c,1-c]}\tilde{u}_i^m(v,\theta)\ \tilde{g}^m_i(v,\theta)\,\dd v\,\dd \theta\right)~.
\end{align*}
Combining the facts that sequence $\{(\tilde{G}_1^m,...,\tilde{G}_N^m)\}_{m\in\Zplus}$ weakly converges to $(\tilde{G}_1,...,\tilde{G}_N)$, sequence $\{G_i^m\}_{m\in \mathbb{Z}^+}$ weakly converges to $G_i$, and there is no tie over $[\underline{\theta},1-c]$ in strategy profile $(G_i,\tilde{G}_{-i})$, we have that
\begin{equation}
\label{eq:proof-of-existence-2}
\lim\limits_{m\rightarrow\infty}\left(\int_{[0,1]\times[-c,1-c]}G_i(v,\theta)\dd u_i(v,\theta)-\int_{[0,1]\times[-c,1-c]}G_i^m(v,\theta)\dd u_i^m(v,\theta)\right)= 0~.
\end{equation}
In addition to the fact $(\tilde{G}_1^m,\dots,\tilde{G}_N^m)$ is an equilibrium in the $m$-th discrete approximation game, we have
\begin{equation}
\label{eq:proof-of-existence-3}
    \lim\limits_{m\rightarrow\infty}\left(\int_{[0,1]\times[-c,1-c]}u_i^m(v,\theta)\ g^m_i(v,\theta)\,\dd v\,\dd \theta-\int_{[0,1]\times[-c,1-c]}u_i^m(v,\theta)\ \tilde{g}^m_i(v,\theta)\,\dd v\,\dd \theta\right)\le0~.
\end{equation}
Combining Inequalities (\ref{eq:proof-of-existence-2}) and (\ref{eq:proof-of-existence-3}), we have 
\begin{equation*}
    \int_{[0,1]\times[-c,1-c]}u_i(v,\theta)\ g_i(v,\theta)\,\dd v\,\dd \theta-\int_{[0,1]\times[-c,1-c]}u_i(v,\theta)\ \tilde{g}_i(v,\theta)\,\dd v\,\dd \theta\le 0~,
\end{equation*}
which shows that for each sender $i$, strategy $\tilde{G}_i$ is a best response to others' strategies $\tilde{G}_{-i}$, and $(\tilde{G}_1,\dots,\tilde{G}_N)$ is indeed an equilibrium in our game.
\end{proof}

Putting all the pieces together, finally we can prove \Cref{thm:existence}.

\begin{proof}[Proof of \Cref{thm:existence}]
\Cref{thm:existence} holds directly by combining \Cref{lem:discrete-game-strategy-action-space,lem:helly,lem:converge-to-feasible-strategy,lem:no-tie-in-limit,lem:limit-profile-is-equilibrium}.
\end{proof}

\section{Symmetric Equilibrium under Concave Priors}
\label{apx:sym-equilibrium}

Here we show that the characterization of symmetric equilibrium of concave priors is much more involved than that of convex priors.
We consider a two-sender game where $\supp(F)=[c,1]$, and the prior $F$ is concave over its support.
Let $S=\int_0^1F(x)\dd x$, and we know $S\ge \frac{1}{2}$ since the concavity of the prior.

Different values of the cost $c$ lead to distinct symmetric equilibrium structures, making it challenging to unify all cases under a single amortized value distribution or a generalized 2-D distribution construction. 
While a complete equilibrium characterization remains an open problem, we present partial results below, specifically detailing the equilibrium structure for cost regimes within certain intervals.

When the cost $c$ is sufficiently large, we make the following construction of a simplest kind of symmetric equilibrium.

\begin{theorem}
\label{thm:symmetric-equilibrium-concave-1}
When the cost $c\ge \frac{2}{3}(1-S)$, there exists a symmetric equilibrium with the amortized value distribution $K(\cdot)$ defined as below:
\begin{equation*}
    K(\theta)=
    \begin{cases}
        F(\theta+c)\ &\emph{if}\ \theta\in[-c,0]~,\\
        \min \left\{\frac{1}{\overline{\theta}}\theta,1 \right\}\ &\emph{if}\ \theta\in(0,1-c]~,
    \end{cases}
\end{equation*}
where $\overline{\theta}$ is the unique solution to the equation $\int_{0}^{1-c}\min \left\{\frac{1}{\overline{\theta}}\theta,1 \right\}\dd \theta=\int_0^{1-c}F(\theta+c)\dd\theta$.
\end{theorem}

\noindent\begin{proof}[Proof of \Cref{thm:symmetric-equilibrium-concave-1}]
    This theorem follows a similar proof of \Cref{thm:symmetric-equilibrium-convex-prior-big-c}.
    The construction method can be seen in \Cref{fig:concave-1}.
\end{proof}

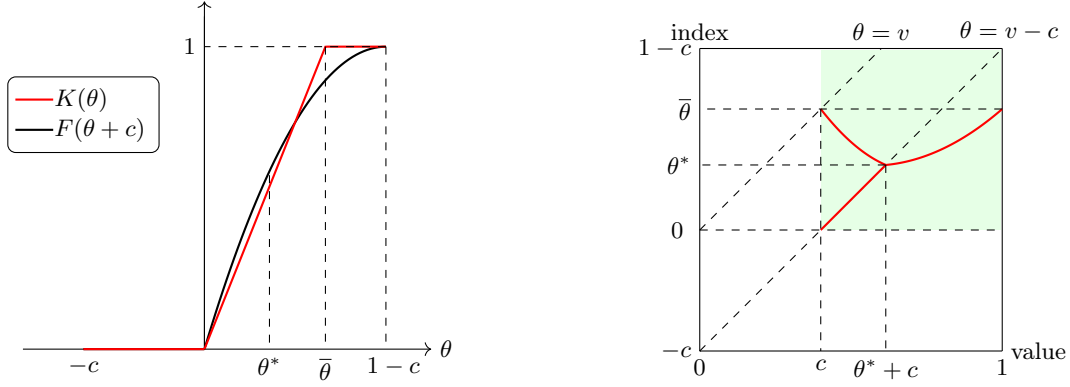
\begin{figure}[h]
    \centering
    \begin{minipage}[t]{0.48\textwidth}
        \centering
        \begin{tikzpicture}[scale=2.0]
            \draw[->] (-1.2,0) -- (1.5,0) node[right] {\footnotesize $\theta$};
            \draw[-] (0,0) -- (0,2.2) node[left] {};
            \draw[->] (0,0) -- (0,2.3) node[left] {};

            \draw[rounded corners=3pt, fill=white, draw=black] (-1.3,1.3) rectangle (-0.3,1.8);
        
            \draw[red, thick] (-1.25,1.65) -- (-1.0,1.65);
            \draw (-1.05,1.65) node[anchor=west, font=\footnotesize] {$K(\theta)$};
        
            \draw[black, thick] (-1.25,1.45) -- (-1.0,1.45);
            \draw (-1.05,1.45) node[anchor=west, font=\footnotesize] {$F(\theta+c)$};
            
            \draw[thick, domain=-0.8:0, samples=50] plot (\x, {0});
            \draw[thick, domain=0:1.2, samples=50] plot (\x, {-1.3889*(\x)^2+3.333*\x});

            \draw[red, thick, domain=-0.8:0, samples=50] plot (\x, {0});
            \draw[red, thick, domain=0:1.2, samples=100] plot (\x,{min(2.5*\x,2)});

            \draw (-0.8,0) node[below] {\footnotesize $-c$};
            \draw[dashed] (0.8,2) -- (0.8,0) node[below] {\footnotesize $\overline{\theta}$};
            \draw[dashed] (0.43,1.15) -- (0.43,0) node[below] {\footnotesize $\theta^*$};
            \draw[dashed] (1.2,2) -- (1.2,0);
            \draw (1.25,0) node[below] {\footnotesize $1-c$};
            \draw[dashed] (1.2,2) -- (0,2) node[left] {\footnotesize $1$};
        \end{tikzpicture}
    \end{minipage}
    \hspace{0.02\textwidth}
    \begin{minipage}[t]{0.48\textwidth}
        \centering
        \begin{tikzpicture}[scale=2.0]
            \fill[green!10] (0.8,0) rectangle (1.99,1.19); 
            \draw[dashed] (0,0) -- (2,0);
            \draw (-0.05,0) node[left] {\footnotesize $0$};
            \draw (0,-0.8) node[below] {\footnotesize $0$};
            \draw (0,1.2) node[left] {\footnotesize $1-c$};
            \draw (2,-0.8) node[below] {\footnotesize $1$};
            \draw (0,-0.8) -- (0,1.2) node[above] {\footnotesize index};
            \draw (2,-0.8) -- (2,1.2);
            \draw (0,-0.8) -- (2,-0.8) node[right] {\footnotesize value};
            \draw (2,1.2) -- (0,1.2);
            
            \draw[dashed] (0,0) -- (1.2,1.2) node[above] {\footnotesize $\theta=v$};
            \draw[dashed] (0,-0.8) -- (2,1.2) node[above] {\footnotesize $\;\; \theta=v-c$};
            
            \draw[red, thick] (1.23,0.43) -- (0.8,0);
            \draw[red, thick, domain=0.8:1.23, samples=100] plot (\x, {(\x)^2-2.89*\x+2.472});
            \draw[red, thick, domain=1.23:2, samples=100] plot (\x, {0.5*(\x)^2-1.1345*\x+1.06895});
            
            \draw[dashed] (0,-0.8) -- (0,-0.8) node[left] {\footnotesize $-c$};
            \draw[dashed] (2,0.8) -- (0,0.8) node[left] {\footnotesize $\overline{\theta}$};
            \draw[dashed] (0.8,0.8) -- (0.8,-0.8) node[below] {\footnotesize $c$};
            \draw[dashed] (1.23,0.43) -- (0, 0.43) node[left] {\footnotesize $\theta^*$};
            \draw[dashed] (1.23,0.43) -- (1.23, -0.8) node[below] {\footnotesize $\theta^*+c$};
        \end{tikzpicture}
    \end{minipage}
        \caption{An example of \Cref{thm:symmetric-equilibrium-concave-1}.
        In the left panel, the red curve represents the amortized value distribution $K(\cdot)$, and the black curve represents the shifted prior.
        Function $K(\cdot)$ forms an MPC of the shifted prior over $[-c,1-c]$.
        In the right panel, the red curves represent the support set of the corresponding 2-D distribution that forms an equilibrium.
        }
        \label{fig:concave-1}
\end{figure}

This type of symmetric equilibrium may fail to exist when $c<\frac{2}{3}(1-S)$, since in that case it holds $\overline{\theta}>c$, our construction of 2-D distribution may break down halfway.
Although this amortized value distribution forms a candidate for equilibrium, it cannot be guaranteed that there exists a 2-D distribution that induces this amortized value distribution.





\section{Setting Comparison of the Competitive Information Design Game}
\label{apx:comparison}

As mentioned in Related Works, 
\citet{ding2023competitive,hwang2025competitive} study a closely related model, where the sender's signal is seen by the buyer \emph{after} inspection, and the buyer sees only the signal instead of her value.
We call this setting one of \emph{information obfuscation}; in comparison, our setting is one of \emph{information revelation}.
In this section, we first compare the informational properties of these two settings and then compare the senders' and the receiver's utility.

\subsection{Setting Comparison}

\paragraph{Extreme cases under both settings.}
Extreme cases highlight the differences between the two models.
In our model, all senders adopting the \emph{full revelation} strategy let the buyer know all the values before searching; her utility would be the first best $\Ex{\max_i (v_i - c_i)_+}$.
If no sender sends any signal,  
the buyer gets no additional information. 
This reduces the case to the original Pandora Box problem where
the buyer just searches with the original priors $F_1, \cdots, F_N$.
In contrast, this original problem corresponds to the case of full revelation in the setting of information obfuscation. 
In this setting, if the senders send no signals, the distribution of each $v_i$ degenerates to a point mass on $\Ex{v_i}$, which reduces the buyer's utility to $\max_{i} (\mathbb{E}[v_i] - c_i)_+$.

\paragraph{Buyer's search behavior under both settings.}
The buyer's search behavior also differs in the two models.
In the information revelation setting, the buyer's search depends on the \emph{signals} she receives, updating her posteriors based on these signals. 
In the obfuscation setting, the search depends only on the \emph{signaling schemes} --- once the senders commit to the schemes, the value distribution of each seller becomes an MPC of the original prior, and the receiver searches using these MPCs, observing a signal only after inspecting a seller.

\subsection{Utility Comparison}

\xhdr{Receiver Utility Comparison}
We treat the original Pandora Box problem as the benchmark.
Fixing all boxes priors and search costs, let $U$ be the buyer's utility when she implements the Index Algorithm in the original Pandora Box problem.

\begin{proposition}[Receiver Utility Comparison]
\label{prop:comparison-receiver}
$U$ is weakly lower than the buyer's utility in the information revelation setting, and weakly higher than the buyer's utility in the information obfuscation setting, regardless of the senders' signaling schemes in both settings.
\end{proposition}

\begin{proof}[Proof of \Cref{prop:comparison-receiver}]
Under the information revelation setting, one search strategy available to the receiver is to ignore the signals and to implement the Index Algorithm based on the original priors.  
This searching strategy yields the same overall utility as in the no information setting.
But this utility is no more than that of the Index Algorithm performed on the posteriors, since the posteriors are more informative than the priors and the agent can better estimate the value based on the posteriors than based on the priors.
This proves the first claim.

For the second claim, observe that in the no information setting, the receiver can simulate the obfuscation setting by intentionally ignoring the true value $v_i$ of any inspected box $i$ and instead behaving as if she had observed its posterior mean value.
The utility of this simulation is precisely that in the obfuscation setting, which is no more than the utility of the Index Algorithm on the original priors, due to the optimality of Index Algorithm.
This proves the second claim.
\end{proof}

The second statement was first shown by \citet{ding2023competitive} (Theorem 3.1);
we consider our proof here considerably simpler. 
Intuitively, the more information accessible to the receiver, 
the higher her utility should be.

\xhdr{Sender Utility Comparison}
The next proposition considers a sender choosing between revealing or obfuscating information.

\begin{proposition}[Sender Utility Comparison]
\label{prop:comparison-sender}
Fix all priors and costs, and consider a sender~$i$ who unilaterally considers deploying a signaling scheme,
we have: 
\begin{enumerate}[(i)]
    \item when $c_i=0$, any strategy of obfuscation gives him weakly higher utility than revelation;
    \item when $c_i > \E_{v_i\sim F_i}[v_i] $, any strategy of revelation gives him weakly higher utility than obfuscation.
\end{enumerate}
\end{proposition}

\begin{proof}[Proof of \Cref{prop:comparison-sender}]
When the inspection cost is zero for sender~$i$, the receiver does any inspection for free.
Thus, it is optimal for her to inspect all boxes and chooses the most favorable one.
An information revealing sender can do nothing, since any information revelation only serves to influence the searching order of the agent.
While for an obfuscating sender, he can influence the final choice of the agent through certain obfuscation strategy.
Therefore, an obfuscating sender has weakly higher utility.

When the inspection cost is higher than the expectation of the prior distribution, an obfuscating sender always has a negative index regardless of any obfuscation strategy.
Thus, he can do nothing in this case, since any box with a negative index will never be inspected under the Index Algorithm.  
On the other hand, an information-revealing sender can sacrifice some low values and bundle higher values together to form strictly positive indices, thereby ensuring an expected payoff greater than zero.
Therefore the latter has a weakly higher utility.
\end{proof}

The cost affects the receiver's willingness to search and her optimal search strategy.
For a seller~$i$ with search cost $c_i=0$, the receiver will definitely inspect the seller for free.
Thus, an information-revealing sender can do nothing, since any revelation only serves to influence the searching order of the receiver.
In contrast, when his cost is sufficiently high, an obfuscating sender always has a negative index regardless of the strategy, thus will never be inspected under the Index Algorithm.

%
%
%

%
%
%
%

\end{document}